\documentclass{article}

\usepackage[draft]{todonotes}
\usepackage[utf8]{inputenc}
\usepackage[english]{babel}
\usepackage{amsmath, amsthm}
\usepackage{amsfonts}
\usepackage{pst-node}
\usepackage{braket}
\usepackage{tikz-cd}
\usepackage{amssymb}
\usepackage{mathrsfs}
\usepackage{textcomp}
\usepackage{mathabx}
\usepackage{graphicx}
\usepackage{bm}
\usepackage{caption}
\usepackage{enumerate}
\usepackage{subcaption}
\usepackage{framed}
\usepackage[margin=2.8cm]{geometry}
\usepackage{fancybox}
\usepackage{dsfont}
\usepackage{tikz}
\usepackage{mathtools,xparse}
\usepackage{hyperref}
\hypersetup{
    colorlinks=true, 
    linktoc=all,     
    linkcolor=blue,  
}
\usepackage{scalerel,stackengine}
\stackMath

\DeclarePairedDelimiter{\norm}{\lVert}{\rVert}
\NewDocumentCommand{\normL}{ s O{} m }{%
  \IfBooleanTF{#1}{\norm*{#3}}{\norm[#2]{#3}}_{L_2(\Omega)}%
}

\newtheorem{theorem}{Theorem}[section]
\newtheorem{lemma}[theorem]{Lemma}
\newtheorem{prop}[theorem]{Proposition}

\theoremstyle{definition}
\newtheorem{definition}[theorem]{Definition}

\usepackage{framed}
\usepackage{comment}
\iffalse  
\excludecomment{detail}
\else

\fi

\newcommand{\bb}{\mathbb}
\newcommand{\ud}{\underline}
\newcommand{\cl}{\mathcal}

\newcommand{\m}{\mathrm}
\newcommand{\utr}{\underline{\mathrm{Tr}}}

\newcommand{\bnorm}{\norm[\bigg]}

\newcommand{\rmi}{\mathrm{i}}
\newcommand{\rme}{\mathrm{e}}

\newcommand{\toinfty}{\xrightarrow[n\to\infty]{} }

\theoremstyle{remark}
\newtheorem{remark}[theorem]{Remark}

\numberwithin{equation}{section}
\newlength\dlf  

\newcommand\N{{\mathbb N}}

\newcommand\R{{\mathbb R}}
\newcommand\cK{{\mathcal K}}
\newcommand\cU{{\mathcal U}}
\newcommand{\gS}{{\mathfrak{S}}}
\newcommand{\br}{\mathbf r}

\begin{document}

\title{$\mbox{A reduced Hartree-Fock model of slice-like defects in the Fermi sea}$}
\author{\'{E}ric Canc\`{e}s$^{1,2}$, Ling-Ling Cao$^{1}$, Gabriel Stoltz$^{1,2}$ \\
\small $^1$ CERMICS, Ecole des Ponts ParisTech, 6 \& 8 avenue Blaise Pascal, 77455 Marne-la-Vall\'ee, France \\
\small $^2$ Inria, 2 rue Simone Iff, 75589 Paris, France}
\maketitle

\begin{abstract}
Studying the electronic structure of defects in materials is an important subject in condensed matter physics. From a mathematical point of view, nonlinear mean-field models of localized defects in insulators are well understood. We present here a mean-field model to study a particular instance of extended defects in metals. These extended defects typically correspond to taking out a slice of finite width in the three-dimensional homogeneous electron gas. We work in the framework of the reduced Hartree-Fock model with either Yukawa or Coulomb interactions. Using techniques developed in~\cite{Frank2011energy, frank2013} to study local perturbations of the free-electron gas, we show that our model admits minimizers, and that Yukawa ground state energies and density matrices converge to ground state Coulomb energies and density matrices as the Yukawa parameter tends to zero. These minimizers are unique for Yukawa interactions, and are characterized by a self-consistent equation. We moreover present numerical simulations where we observe Friedel oscillations in the total electronic density.
\end{abstract}

\section{Introduction}
\label{introSec}

The study of the electronic structure of defects in materials is an important topic in condensed matter physics (see e.g.~\cite{RevModPhys.86.253,kaxiras_2003, doi:10.1063/1.3318261, RevModPhys.50.797, stoneham2001theories} and references therein). The case of linear one-body Hamiltonians describing independent electrons has been thoroughly investigated, in particular to study the effect of disorder on transport properties (see e.g.~\cite{Bellissard94,Hundertmark2007} and references therein). Nonlinear mean-field models such as Hartree-Fock or Kohn-Sham type models are much more difficult to handle, because of two major difficulties: (i) their non-convexity, leading to possible symmetry breaking and existence of multiple ground-state densities for a given nuclear configuration, (ii) the long-range of Coulomb interactions. The reduced Hartree-Fock model~\cite{Solovej1991} is of particular interest in mathematical physics. First, being strictly convex in the density, it allows to address the latter difficulty while getting rid of the former one. Second, it can be seen as a good approximation of the (extended) Kohn-Sham LDA model~\cite{KohnSham,DreizlerGross,LewinLiebSeiringer}, which is commonly used in solid-state and condensed matter physics. Most results obtained on the reduced Hartree-Fock model can be extended to the Kohn-Sham LDA model~\cite{AnantharamanCances}, up to possibly some assumptions on the uniqueness of the ground-state density matrix or on the coercivity of the second-order derivative at the ground-state under consideration.

For insulators, a reduced Hartree-Fock model with Coulomb interactions has been proposed in~\cite{Cances2008} to study a local defect in an insulating (or semiconducting) crystal, based on the ideas and techniques from~\cite{Hainzl2005,0305-4470-38-20-014,Hainzl2009,HAINZL2005TheMA}. This model is variational and consists of minimizing some renormalized energy formally obtained by taking the difference between the (infinite) energies of the crystal with the defect, and of the perfect crystal. This approach can be mathematically justified by a thermodynamic limit argument. The zero-frequency dielectric polarizability properties of insulating crystals can be inferred from this model by a homogenization procedure~\cite{Cances2010}. Extensions to the time-dependent setting are discussed in~\cite{CANCES2012887}. The numerical analysis of the steady case is dealt with in~\cite{Gontier2016supercell} (see also \cite{Cances2008b}). Let us also mention the works~\cite{CancesLahbabiLewin2013, Lahbabi2014} in which crystals with stationary random distributions of local defects have been studied.

The above mentioned works are only valid for insulators and semiconductors, and crucially rely on the existence of a spectral gap in the spectrum of the mean-field Hamiltonian of the corresponding perfect crystal. Mean-field model for defects in metals are much more difficult to analyze since small perturbations can cause electrons to escape at infinity. On the other hand, many interesting physical problems, such as electronic transport, occur in metals. The typical setting corresponds to fixing different chemical potentials for the electrons in infinite perfect leads to which a system of interest is connected, as studied for instance in~\cite{AJOO07,Cornean2008, Cornean2012, BJLP15}. 

In~\cite{frank2013}, the authors have considered local perturbations of the Fermi sea of the free-electron gas, for a fixed Fermi level. This setting is natural for unbounded metallic systems for which the charge per unit volume cannot be fixed since a local electronic neutrality is needed. The well-posedness of the dynamics of defects of the Fermi sea of the free-electron gas is proved in~\cite{LewinSabin1}. These works are important milestones in the construction of mathematical sound mean-field models for local defects in real (nonuniform) metals.

Lastly, several correlated-electron models for crystals with defects have been proposed in the physics literature, using either heuristic formulations such as Hubbard models, or approximations to many-body systems as provided by Dynamical Mean-Field Theory (DMFT), Green's function methods (GW, Bethe--Salpeter), or Monte Carlo methods, see e.g. the recent monograph~\cite{MartinReiningCeperley} and references therein.

In this work, we study a particular instance of extended defects in metals within the reduced Hartree-Fock model. More precisely, we consider 2D-translational invariant defects in a 3D homogeneous electron gas. A typical situation is the case when a slice of finite width of the jellium modeling the uniform nuclear distribution is taken out. This gives rise to a model describing the uncharged state of a capacitor composed of two semi-infinite leads separated by some dielectric medium or vacuum. This could be a first step toward the construction of a mean-field model for electronic transport~\cite{Lang2001}. Our mathematical analysis heavily relies on the translation invariance in the directions parallel to the slice. Technically, this allows us to reduce the study of a three-dimensional model to the one of a family of one-dimensional problems labeled by a two-dimensional quasimomentum.

Let us emphasize that our analysis could be adapted to treat local defects in the free electron gas or other types of defects with some sort of symmetry (\emph{e.g.} cylinder-shaped defects). Somehow, the situation we consider in this work is the one for which the technical issues are more acute since the family of effective problems is one dimensional, which raises integrability issues due to Peierls oscillations~\cite{frank2013}.

\medskip

This article is organized as follows. In Section~\ref{sec:2}, we introduce a reduced Hartree-Fock model amenable to describe an extended two-dimensional defect in the three-dimensional Fermi sea, under the assumption that the defect is translation invariant in the $(x,y)$-directions. After introducing the functional setting in Section~\ref{kineticSec}, we define renormalized free kinetic and potential energy functionals for an $(x,y)$-translation invariant defect in Section~\ref{defineDefects}. In Section~\ref{mimimizerSec}, we use these elementary bricks to define reduced Hartree-Fock (rHF) energy functionals for $(x,y)$-translation invariant defects, both for Coulomb and Yukawa interactions, and prove the existence of a ground state. We also show that the Yukawa ground states converge to the Coulomb ground states when the characteristic length of the Yukawa interaction goes to infinity, and uniquely characterize the minimizers for Yukawa interactions. The proof of the results presented in Section~\ref{sec:2} can be read in Section~\ref{sec:proofs}, some technical results being postponed to the appendix.


\section{Construction of the model}
\label{sec:2}

Let us first introduce some notation. Unless otherwise specified, the functions on $\R^d$ considered in this article are complex-valued. Elements of $\R^3$ are denoted by $\br = (r,z)$, where $r= (x,y)\in \bb{R}^2$ and $z \in \R$. We denote respectively by ${\mathscr S}(\R^d)$ the Schwartz space of rapidly decreasing functions on~$\R^d$, and by ${\mathscr S}'(\R^d)$ the space of tempered distributions on~$\R^d$.

Let $\mathfrak{H}$ be a separable Hilbert space. We denote by $\mathcal{L}(\mathfrak{H})$ the space of bounded (linear) operators on $\mathfrak{H}$, by $\mathcal{S}(\mathfrak{H})$ the space of bounded self-adjoint operators on $\mathfrak{H}$, and by $\cl K(\mathfrak{H})$ the space of compact operators on $\mathfrak{H}$.
We denote by $\mathfrak{S}_p(\mathfrak{H})$ the $p$-Schatten class on $\mathfrak{H}$, for $1 \leq p < \infty$: $A\in \cl K(\mathfrak{H})$ is in $\mathfrak{S}_p(\mathfrak{H})$ if and only if $\norm{A}_{\mathfrak{S}_p}=(\m{Tr}(|A|^p))^{1/p} <\infty$. Recall that operators in $\mathfrak{S}_1(\mathfrak{H})$ and $\mathfrak{S}_2(\mathfrak{H})$ are respectively called trace-class and Hilbert--Schmidt.

If $A \in \mathfrak{S}_1(L^2(\R^d))$, there exists a unique function $\rho_A \in L^1(\R^d)$ such that
\[
\forall W \in L^\infty(\R^d), \quad \m{Tr}(AW) = \int_{\R^d} \rho_A W.
\]
The function $\rho_A$ is called the density of the operator $A$. If the integral kernel $A(\br,\br')$ of $A$ is continuous on $\R^d \times \R^d$, then $\rho_A(\br)=A(\br,\br)$ for all $\br \in \R^d$. This relation still stands in some weaker sense for a generic trace-class operator.

An operator $A \in \mathcal{L}(L^2(\R^d))$ is called locally trace-class if the operator $\chi A \chi$ is trace-class for any $\chi \in C^\infty_c(\R^d)$. The density of a locally trace-class operator $A\in \mathcal{L}(L^2(\R^d))$ is the unique function $\rho_A \in L^1_{\rm loc}(\R^d)$ such that
\[
\forall W \in C^\infty_c(\R^d), \quad \m{Tr}(AW) = \int_{\R^d} \rho_A W.
\]

We denote respectively by $\widehat u$ and $\widecheck{u}$ the Fourier transform and the inverse Fourier transform of a tempered distribution $u \in {\mathscr S}'(\R^d)$. We use the normalization convention for which
\[
\forall \phi \in L^1(\R^d), \quad \widehat{\phi}(\zeta):=\frac{1}{(2\pi)^{d/2}}\int_{\bb R^d} \phi(t)\rme^{-\rmi t \cdot \zeta}\,dt \quad \mbox{and} \quad
\widecheck{\phi}(t) := \frac{1}{(2\pi)^{d/2}}\int_{\bb R^d}\phi(\zeta)\rme^{\rmi t\cdot \zeta}\,d\zeta.
\]
With this normalization convention, the Fourier transform defines a unitary operator on $L^2(\R^d)$.

\subsection{Functional setting}
\label{kineticSec}

In Section \ref{directIntegralsec}, we introduce a natural decomposition of $(x,y)$-translation invariant operators based on partial Fourier transform. In Section~\ref{densitymatrices}, we apply it to the special case of $(x,y)$-translation invariant one-body density matrices.

\subsubsection{Decomposition of $(x,y)$-translation invariant operators}
\label{directIntegralsec}

For $r = (x,y)\in \bb R^2$, we denote by $\tau_r$ the translation operator acting on $L_{\mathrm{loc}}^2(\bb R^3)$ as
$$
\forall u\in L_{\mathrm{loc}}^2(\bb R^3),  \quad (\tau_r u)(\cdot,z) = u(\cdot - r,z) \quad \mbox{ for a.a. } z \in \R.
$$
An operator $A$ on $L^2(\bb R^3)$ is called $(x,y)$-translation invariant if it commutes with $\tau_r$ for all $r \in \R^2$. In order to decompose $(x,y)$-translation invariant operators on $L^2(\R^3)$, we introduce the constant fiber direct integral \cite[Section~XIII.16]{ReeSim4}
$$
L^2\left(\bb R^2;L^2(\bb R)\right) \equiv \int_{\bb R^2}^{\oplus} L^2(\bb R) \,dq
$$
with base $\R^2$, and the unitary operator $\cU : L^2(\bb R^3) \to L^2(\bb R^2;L^2(\bb R))$ defined on the dense subspace ${\mathscr S}(\R^3)$ of $L^2(\bb R^3)$ by
\begin{equation}
(\mathcal{U}\Phi)_q(z):= \frac{1}{2\pi}\int_{\bb R^2}\rme^{-\rmi q\cdot r}\Phi(r,z)\,dr.
\label{unitaryT}
\end{equation}
The unitary $\cU$ is simply the partial Fourier transform along the $x$ and $y$ directions. It has the property that $(x,y)$-translation invariant operators on $L^2(\bb R^3)$ are decomposed by $\cU$: for any $A\in \mathcal{L}(L^2(\bb R^3))$ such that $\tau_rA=A\tau_r$, there exists $A_\bullet \in L^\infty(\R^2;\mathcal{L}(L^2(\bb R)))$ such that for all $u \in L^2(\R^3)$,
$$
(\,\mathcal{U}(Au))_q = A_q(\,\mathcal{U}u)_q \quad \mbox{ for a.a. } q \in \R^2.
$$
Hence we use the following notation for the decomposition of $(x,y)$-translation invariant operator $A$
\[
A=\,\mathcal{U}^{-1}\left( \int_{\bb R^2}^{\oplus }A_{q}\,dq\right)\,\mathcal{U}.
\]
In addition, $\norm{A}_{\mathcal{L}(L^2(\bb R^3))} = \norm[\big]{\norm{A_\bullet}_{\mathcal{L}(L^2(\bb R))}}_{L^{\infty}(\bb R^2)}$. Note that, formally, the kernel of $A$ is related to the kernels of the operators $A_q$ by the formula:
\[
A(r,z;r',z') = \frac{1}{(2\pi)^2} \int_{\R^2} A_q(z,z') \, \rme^{\rmi q (r-r')} \, dq.
\]
In particular, if $A$ is positive and locally trace-class, then for almost all $q \in \R^2$, $A_q$ is locally trace-class. The densities of these operators are functions of the variable $z$ only, and are related by the formula
$$
\rho_A(z) = \frac{1}{(2\pi)^2} \int_{\R^2} \rho_{A_q}(z) \, dq.
$$
Likewise, if $A$ is a (not necessarily bounded) self-adjoint operator such that $\tau_r(A+\rmi)^{-1}=(A+\rmi)^{-1}\tau_r$ for all $r\in \bb R^2$, then $A$ is decomposed by $\,\mathcal{U}$ (see~\cite[Theorem XIII.84 and~XIII.85]{ReeSim4}). In particular, the kinetic energy operator $T=-\frac 12 \Delta$ on $L^2(\R^3)$ is decomposed by $\cU$ as follows:
\begin{equation}
T=\,\mathcal{U}^{-1}\left( \int_{\bb R^2}^{\oplus }T_{q}\,dq\right)\,\mathcal{U} \qquad \mbox{with} \qquad T_{q} :=-\frac{1}{2}\frac{d^2}{dz^2}+\frac{|q|^2}{2}.
\label{kineticOpDecomp}
\end{equation}

\subsubsection{One-body density matrices}
\label{densitymatrices}

In Hartree-Fock and Kohn-Sham models, electronic states are described by one-body density matrices (see e.g. \cite{Coleman, Cances2008, frank2013}).
Recall that for a finite system with $N$ electrons, a density matrix is a trace-class self-adjoint operator $\gamma\in \mathcal{S}(L^2(\bb R^3)) \cap \gS_1(L^2(\R^3))$ satisfying the Pauli principle $0\leq \gamma\leq 1$ and the normalization condition $\m{Tr}(\gamma)= \int_{\R^3} \rho_\gamma = N$. The kinetic energy of $\gamma$ is given by $\m{Tr}(-\frac{1}{2}\Delta \gamma):= \frac 12 \m{Tr}(|\nabla|\gamma|\nabla|)$ (see~\cite{CATTO2001687, Cances2008}).

Let us from now on focus on the reduced Hartree-Fock(rHF) model, \textit{i.e.}~the Hartree-Fock model without exchange terms. In this case, the ground state density matrix of a homogeneous electron gas with density $\rho_0$ can be uniquely defined by a thermodynamic limit argument (relying on the strict convexity of the rHF model with respect to the density). It is given by
\begin{equation} \label{eq:gamma0}
\gamma_0 := \mathds 1_{(-\infty,\epsilon_F]}\left(T\right),
\end{equation}
with the Fermi level
$$
\epsilon_F:= \frac 12 (6\pi^2\rho_0)^{2/3},
$$
which is the chemical potential of the electrons. As discussed in the introduction, this Fermi level will be fixed in the sequel. Although $\gamma_0$ is not trace-class, it is locally trace-class and its density is $\rho_0$ by construction. The operator $\gamma_0$ can be seen as the rHF ground-state density matrix of an infinite, locally neutral system, whose nuclear distribution is a jellium of uniform density $\rho_{\rm nuc}^0=\rho_0$. 

Since $T$ is decomposed by $\cU$, so is $\gamma_0$, and we have
\begin{equation}
\gamma_0 = \,\mathcal{U}^{-1}\left( \int_{\bb R^2}^{\oplus} \gamma_{0,q}\, dq \right)\,\mathcal{U},
\label{defGamma0}
\end{equation}
where $\{\gamma_{0,q}\}_{q\in \bb R^2}$ are orthogonal projectors acting on $L^2(\bb R)$: $$\gamma_{0,q}:=\left\{
\begin{aligned}
\mathds 1_{(-\infty,\epsilon_F]}\left(T_{q}\right)\quad &\text{if } q\in \overline{\mathfrak{B}}_{\epsilon_F}, \\
0 \quad\quad\quad &\text{if } q\in \bb R^2 \setminus \overline{\mathfrak{B}}_{\epsilon_F}.
\end{aligned}
\right.$$
Here and in the sequel, $\mathfrak{B}_{R}:=\left\{q\in \bb R^2 \left| \frac{|q|^2}{2}< R\right.\right\}$ and $\overline{\mathfrak{B}}_{R}:=\left\{q\in \bb R^2 \left| \frac{|q|^2}{2}\leq R\right.\right\} $ respectively denote the open and closed balls of $\R^2$ of radius $\sqrt{2R}$ centered at the origin.

If we consider an $(x,y)$-translation invariant perturbation $\rho_{\rm nuc}=\rho_{\rm nuc}^0+\nu$ of the nuclear distribution, and keep the Fermi level $\epsilon_F>0$ fixed, we expect the perturbed ground state density matrix $\gamma_\nu=\gamma_0+Q_\nu$ to be $(x,y)$-translation invariant as well (see Remark~\ref{rmk:trans_inv} below), and therefore the operators $\gamma_\nu$ and $Q_\nu$ to be decomposed by $\cU$:
\[
\gamma_\nu = \,\mathcal{U}^{-1}\left( \int_{\bb R^2}^{\oplus} \gamma_{\nu,q}\, dq \right)\,\mathcal{U} \quad \mbox{and} \quad
Q_\nu = \,\mathcal{U}^{-1}\left( \int_{\bb R^2}^{\oplus} Q_{\nu,q}\, dq \right)\,\mathcal{U}.
\]
We will see that $Q_\nu$ can be characterized as the unique minimizer of a variational problem consisting in minimizing some renormalized free energy functional.

\begin{remark}[On the translation invariance of the defect]
  \label{rmk:trans_inv}
  The problem we consider here consists in characterizing the perturbed state of minimal energy among the class of $(x,y)$-translation invariant perturbations. An interesting question, left open for future work, is to justify through a thermodynamic limit argument that there is indeed no symmetry-breaking. The proof in~\cite{Cances2008} is only valid for local defects in insulating materials, and should therefore be adapted to account for the nonlocal nature of the defect we consider, and the fact that the reference perfect crystal is metallic.
\end{remark}

\subsection{Renormalized free energy functionals}
\label{defineDefects}

Defects that are $(x,y)$-translation invariant are extended (non-local) defects, and therefore, do not fall into the frameworks of  \cite{frank2013, LewinSabin1} (nor {\it a fortiori} of  \cite{Cances2008} since the homogeneous electron gas is a metal). However, the approach consisting in characterizing the ground states as the minimizers of some renormalized free energy functional can still be used.

In Section \ref{perturbedKineticEnergy}, we define a renormalized kinetic free energy per unit area adapted to $(x,y)$-translation invariant perturbations of the homogeneous electron gas. In Section \ref{potentialSec}, we focus on the potential energy contributions, and define renormalized energies per unit area for $(x,y)$-translation invariant systems, both for Yukawa and Coulomb interactions.

\subsubsection{Renormalized kinetic free energy functional}
\label{perturbedKineticEnergy}

Let us start with a formal (non-rigorous) argument. The kinetic energy densities of the operator $\gamma_0$ and of an operator of the form $\gamma=\gamma_0+Q$ can be defined as
\begin{align*}
&t_{\gamma_0}(\br):= \rho_{T^{1/2}\gamma_0T^{1/2}}(\br) ,  \\
&t_\gamma(\br):=\rho_{T^{1/2}\gamma T^{1/2}}(\br) =t_{\gamma_0}(\br)+t_Q(\br) \quad \mbox{with} \quad t_Q(\br):=\rho_{T^{1/2}Q T^{1/2}}(\br).
\end{align*}
By $(x,y)$-translation invariance, the functions $t_{\gamma_0}$, $t_\gamma$ and $t_Q$ are in fact functions of the transverse variable~$z$ only. Fixing the Fermi level $\epsilon_F>0$, we can therefore define a renormalized kinetic free energy per unit area as
\begin{align*}
\underline{T}_{\rm ren}(Q)&:= \left( \int_\R t_\gamma(z) \, dz - \epsilon_F \int_\R \rho_{\gamma}(z) \, dz \right) - \left( \int_\R t_{\gamma_0}(z) \, dz - \epsilon_F \int_\R \rho_{\gamma_0}(z) \, dz \right)\\
& = \int_\R (t_{Q}(z)-\epsilon_F \rho_{Q}(z) ) \, dz.
\end{align*}
Decomposing by $\cU$ and using the fact that $\gamma_{0,q}$ is an orthogonal projector commuting with $T_q$ and such that
$$
(T_q-\epsilon_F)\gamma_{0,q}= - |T_q-\epsilon_F|\gamma_{0,q}, \quad (T_q-\epsilon_F)(1-\gamma_{0,q})= |T_q-\epsilon_F| (1-\gamma_{0,q}),
$$
we obtain
\begin{align}
\underline{T}_{\rm ren}(Q)
&= \int_\R \left( \rho_{T^{1/2}Q T^{1/2}}(\br) -\epsilon_F \rho_{Q}(z) \right) \, dz   \nonumber \\
& =\int_\R  \left[  \frac{1}{(2\pi)^2} \int_{\R^2} \left( \rho_{T_q^{1/2}Q_qT_q^{1/2}}(z)-\epsilon_F \rho_{Q_q}(z) \right) \, dz \right] dq   \nonumber  \\
&= \frac{1}{(2\pi)^2} \int_{\R^2} \left( \m{Tr}(T_q^{1/2}  Q_qT_q^{1/2})-\epsilon_F \m{Tr}(Q_q) \right) dq   \nonumber  \\
&=  \frac{1}{(2\pi)^2} \int_{\R^2} \m{Tr}\left((T_q-\epsilon_F)  Q_q\right) dq   \nonumber  \\
&= \frac{1}{(2\pi)^2} \int_{\R^2} \m{Tr}\left[(T_q-\epsilon_F)  (\gamma_{0,q}+(1-\gamma_{0,q})) Q_q\right] dq\nonumber  \\
&= \frac{1}{(2\pi)^2} \int_{\R^2} \m{Tr}\left(|T_q-\epsilon_F|^{1/2} (Q_q^{++}-Q_q^{--})  |T_q-\epsilon_F|^{1/2} \right) dq,
\label{eq:Tren}
\end{align}
where
\[
Q_q^{++} :=  (1-\gamma_{0,q}) Q_q (1-\gamma_{0,q}) \ge 0 \quad \mbox{and}  \quad Q_q^{--} :=  \gamma_{0,q} Q_q \gamma_{0,q} \le 0.
\]
It follows that the integrand in the right-hand side of~\eqref{eq:Tren} is non-negative. We also observe that
\[
0 \le \gamma_0+Q \le 1 \quad \Leftrightarrow \quad \left( -\gamma_{0,q} \le Q_q \le 1-\gamma_{0,q} \quad \mbox{a.e.} \right) \quad \Leftrightarrow \quad \left( Q_q^2 \le Q_q^{++}-Q_q^{--} \quad \mbox{a.e.} \right),
\]
so that
\begin{align*}
\int_{\R^2}  \left\| |T_q-\epsilon_F|^{1/2}Q_q \right\|_{\mathfrak{S}_2(L^2(\bb R))}^2 \, dq &= \int_{\R^2}  \m{Tr}\left(|T_q-\epsilon_F|^{1/2} Q_q^2 |T_q-\epsilon_F|^{1/2} \right) dq \\
& \le \int_{\R^2}  \m{Tr}\left(|T_q-\epsilon_F|^{1/2} (Q_q^{++}-Q_q^{--})  |T_q-\epsilon_F|^{1/2}\right) dq .
\end{align*}

Reasoning as in \cite{frank2013, LewinSabin1}, the above formal manipulations lead us to introduce
\begin{itemize}
\item the functional space
\begin{align*}
\mathcal{X}_{q} :=  \left\{ Q_{q} \in \mathcal{S}(L^2(\bb R))\left| \,|T_q-\epsilon_F|^{1/2}Q_{q}\in\mathfrak{S}_2(L^2(\bb R)), |T_q-\epsilon_F|^{1/2}Q_{q}^{\pm\pm}|T_q-\epsilon_F|^{1/2}\in \mathfrak{S}_1(L^2(\bb R))\right. \right\},
\end{align*}
which, equipped with the norm
\begin{align*}
\norm{Q_{q}}_{ \mathcal{X}_{q}}:=\norm{Q_{q}}_{\mathcal{L}(L^2(\bb R))} + \bnorm{|T_q-\epsilon_F|^{1/2}Q_{q}}_{ \mathfrak{S}_2\left(L^2(\bb R)\right)}&+ \sum_{\alpha \in \{+,-\}} \bnorm{ |T_q-\epsilon_F|^{1/2}Q_{q}^{\alpha\alpha}|T_q-\epsilon_F|^{1/2}}_{\mathfrak{S}_1(L^2(\bb R))},
\end{align*}
is a Banach space;
\item the convex set $\displaystyle \mathcal{K}_{q}:= \Big\{ Q_{q} \in \mathcal{X}_{q} \, \Big| \, -\gamma_{0,q} \leq Q_{q}\leq 1-\gamma_{0,q}\Big\}$;
\item the linear form
 \begin{equation}
\utr\left(\left(T-\epsilon_F\right)Q\right) := \frac{1}{(2\pi)^2} \int_{\R^2} \m{Tr}\left(|T_q-\epsilon_F|^{1/2} (Q_q^{++}-Q_q^{--})  |T_q-\epsilon_F|^{1/2}\right) dq,
\label{defkin}
\end{equation}
which is well-defined with values in $[0,+\infty]$ whenever $\R^2 \ni q \mapsto Q_q \in {\cal S}(L^2(\R))$ is measurable with $Q_q \in \mathcal{K}_{q}$ for almost all $q \in \R^2$.
\end{itemize}

\begin{definition}(Density matrices with finite renormalized kinetic free energy per unit area) An $(x,y)$-translation invariant density matrix
\[
\gamma = \gamma_0+Q
\]
has a finite renormalized kinetic free energy per unit area if $Q \in \cK$, where
\begin{equation}
\hspace{-2mm}
\mathcal{K}:= \left\{\left. Q = \mathcal{U}^{-1}\left(\int_{\bb R^2}^{\oplus}Q_{q}\, dq \right)\mathcal{U}\, \right| \, q\mapsto Q_{q}\in L^{\infty}\left(\bb R^2;\mathcal{S}(L^2(\bb R))\right),Q_q\in \mathcal{K}_q\,\,\mathrm{a.e.}, \utr\left((T-\epsilon_F)Q\right) <\infty \right\}.
\end{equation}
\end{definition}

It is not obvious {\it a priori} that operators in $\mathcal{K}$, which are not trace-class, nor even compact, have densities. However, it is in fact possible to define the density $\rho_Q$ of any state $Q\in \mathcal{K}$, which will be useful to define renormalized  rHF free energy functionals involving Yukawa or Coulomb interactions (see Section~\ref{mimimizerSec}). The precise result, whose proof relies on Lieb--Thirring type inequalities, can be read in Section~\ref{propoDefofDensity}.

\subsubsection{Coulomb and Yukawa energy functionals}
\label{potentialSec}

The extended defect being $(x,y)$-translation invariant, the renormalized total charge density
$$
\rho := \rho_{\gamma_0+Q}-(\rho_{\rm nuc}^0+\nu) = \rho_Q-\nu
$$
is a function of the variable $z$ only. The Coulomb potential generated by this density is therefore obtained by solving the 1D Poisson equation $-v_{\rho,0}''=2\rho$, which also reads in Fourier representation $|k|^2 \widehat v_{\rho,0}(k) = 2 \widehat \rho(k)$. Formally, the Coulomb energy of $\rho$ per unit area is thus given by
$$
\int_\R \rho(z) \,  v_{\rho,0}(z) \, dz =  \int_\R \overline{\widehat \rho(k)}  \, \widehat v_{\rho,0}(k) = 2  \int_{\bb R}\frac{|\widehat \rho(k)|^2}{|k|^2}\,dk.
$$
This motivates the following definition of the 1D Coulomb space
\begin{equation}
\ud{\mathcal{C}} :=\left\{ \rho \in\mathscr{S}'(\bb R)\,\left| \, \widehat{\rho} \in L_{\rm{loc}}^1(\bb R), \frac{\widehat \rho (k)}{|k|} \in L^2(\bb R)\right. \right\},
\end{equation}
which, endowed with the inner product
\begin{equation}
\ud D(\rho_1,\rho_2) :=2 \int_{\bb R}\frac{\overline{\widehat \rho_1(k)}\widehat{\rho_2}(k)}{|k|^2}\,dk,
\end{equation}
is a Hilbert space. The quantity $\frac 12 \ud D(\rho,\rho) \in [0,+\infty]$ represents the Coulomb energy per unit area of the $(x,y)$-translation invariant renormalized charge density $\rho$.

\begin{remark}
Note that charge densities in $\ud{\mathcal{C}}$ are neutral in some weak sense. In particular, if $\rho \in\ud{\mathcal{C}}\bigcap L^1(\bb R)$, then $\int_\R \rho = (2\pi)^{1/2} \widehat{\rho}(0)=0$ since the function $k \mapsto \frac{|\widehat \rho(k)|^2}{|k|^2}$ has to be integrable in the vicinity of $0$.
\label{neutralityofCharge}
\end{remark}

Likewise, the Yukawa potential of parameter $m > 0$ generated by the renormalized charge density $\rho$ of the extended defect is obtained by solving the 1D Yukawa equation $-v_{\rho,m}''+m^2v_{\rho,m}=2\rho$, and its Yukawa energy per unit area is formally given by
$$
 \int_{\bb R}\frac{{|\widehat{\rho}(k)|^2}}{|k|^2+m^2}\,dk.
$$
This leads us to introduce the Yukawa space of parameter $m$
\begin{equation}
\ud{\mathcal{C}_m} :=\left\{ \rho\in\mathscr{S}'(\bb R)\left| \, \widehat{\rho} \in L_{\rm{loc}}^1(\bb R), \frac{\widehat \rho (k)}{\sqrt{|k|^2+m^2}} \in L^2(\bb R) \right.\right\},
\end{equation}
endowed with the inner product
\begin{equation} \label{eq:Yukawa}
\ud{D_m}(\rho_1,\rho_2) :=2 \int_{\bb R}\frac{\overline{\widehat \rho_1(k)}\widehat{\rho_2}(k)}{|k|^2+m^2}\,dk.
\end{equation}
We will use in the sequel the consistent notation $\ud{D_0}:= \ud{D}$ for Coulomb interactions.

\begin{remark}
For any $m>0$, the Yukawa space $\ud{\mathcal{C}_m}$ actually coincides with the Sobolev space $H^{-1}(\bb R)$ and the norms $\|\cdot\|_{H^{-1}}$ and $\ud{D_m}(\cdot,\cdot)^{1/2}$ are equivalent. However we will consider in the following $m$ as a parameter and will pass to limit $m\to 0$ to make a connection with the Coulomb interaction. We therefore prefer to keep the notation $\ud{\mathcal{C}_m}$.
\end{remark}

\begin{remark}\label{rem:finiteYukawa}
  Proposition~\ref{prop:densityK} implies that the density associated to any $Q \in \cK$ has a finite renormalized Yukawa energy per unit area (using the embedding $L^p(\R)+L^2(\R) \hookrightarrow H^{-1}(\R)$ for $1 < p < 5/3$). On the other hand, its renormalized Coulomb energy can be either finite or infinite.
\end{remark}

\subsection{Formulation and mathematical properties of the model}
\label{mimimizerSec}

We now consider an $(x,y)$-translation invariant nuclear defect $\nu$, typically a sharp trench
$$
\nu=-\rho_{\rm nuc}^0 \mathds 1_{[-a,a]}(z)
$$
for some $a > 0$, where $\mathds 1_{[-a,a]}: \R \to \R$ is the characteristic function of the range $[-a,a]$. Mollified versions of this indicator function can also be considered.

Based on the content of Section~\ref{defineDefects}, we can define the renormalized free energy per unit area associated with a trial density matrix $\gamma=\gamma_0+Q$ by
\begin{equation}
\mathcal{E}_{\nu,m}(Q) = \utr\left(\left(T-\epsilon_F\right)Q\right)+ \frac{1}{2}\ud{D_m} (\rho_{Q} - \nu,\rho_Q-\nu),
\label{energyfuc}
\end{equation}
where the renormalized kinetic free energy per unit area is given by (\ref{defkin}), and when the Yukawa ($m >0$) or Coulomb ($m=0$) potential energy functional  per unit area is given by~\eqref{eq:Yukawa}. For any $Q \in \cK$, the right-hand side of~\eqref{energyfuc} is the sum of two non-negative terms. The former is always finite. The latter is always finite for Yukawa interactions as soon as $\nu \in H^{-1}(\R)$, but can {\it a priori} be infinite for Coulomb interactions. For this reason, we introduce the set
\[
\mathcal{F}_\nu := \left\{Q\in \mathcal{K} \mid \rho_Q - \nu \in \ud{\mathcal{C}}\right\}.
\]
Recall that this set may be empty (see Remark~\ref{rem:finiteYukawa}). We can then state the following result.

\begin{theorem}[Existence of minimizers] $\,$
\begin{enumerate}[(1)]
\item Yukawa interaction: for any $\nu \in H^{-1}(\R)$, the minimization problem
\begin{equation}
\boxed{I_{\nu,m} = \inf\{\mathcal{E}_{\nu,m}(Q), Q \in \mathcal{K}\}}
\label{minimizationpbyoka}
\end{equation}
has a minimizer $Q_{\nu,m}$ and all the minimizers share the same density $\rho_{\nu,m}$.
\item Coulomb interaction: for any $\nu \in L^1(\R)$ such that $|\cdot|\nu(\cdot) \in L^1(\R)$, the set ${\mathcal F}_\nu$ is non-empty, the minimization problem
\begin{equation}
\boxed{I_{\nu,0} = \inf\{\mathcal{E}_{\nu,0}(Q), Q \in \mathcal{F}_{\nu}\}}
\label{minimizationpb}
\end{equation}
has a minimizer $Q_{\nu,0}$, and all the minimizers share the same density $\rho_{\nu,0}$.
\item For any $\nu \in H^{-1}(\R)$, the function $(0,+\infty) \ni m \mapsto I_{\nu,m}\in \bb R_+$ is continuous, non-increasing,
\[
\lim_{m\to 0}I_{\nu,m} \leq I_{\nu,0} \quad \mbox{and} \quad \lim_{m\to +\infty}I_{\nu,m} = 0,
\]
with the convention that $I_{\nu,0}=+\infty$ if ${\mathcal F}_\nu$ is empty. When $\nu \in L^1(\R)$ and $|\cdot|\nu(\cdot) \in L^1(\R)$,
\[
\lim_{m\to 0}I_{\nu,m} = I_{\nu,0}.
\]

Moreover, if $\nu \in L^1(\R)$ and $|\cdot|\nu(\cdot) \in L^1(\R)$, there exists a sequence $(m_k)_{k \in \N}$ of positive real numbers decreasing to zero, and a sequence $(Q_{\nu,m_k})_{k \in \N}$ of elements of $\mathcal K$ such that, for each $k \in \N$, $Q_{\nu,m_k}$ is a minimizer of \eqref{minimizationpbyoka} for $m=m_k$, converging to a minimizer $Q_{\nu,0}$ of~\eqref{minimizationpb} in the following sense:
\begin{align}
& \cU \,\displaystyle Q_{\nu,m_k} \,\cU^{-1} \mathop{\longrightarrow}_{k \to \infty} \cU\, Q_{\nu,0} \,\cU^{-1} \quad \mbox{for the weak-$\ast$ topology of } L^\infty(\R^2;{\cal S}(L^2(\R))); \label{eq:cv_Q_mk_1} \\
& \cU\, |T-\epsilon_F|^{1/2} Q_{\nu,m_k}\, \cU^{-1} \mathop{\longrightarrow}_{k \to \infty}  \cU \,|T-\epsilon_F|^{1/2}Q_{\nu,0} \,\cU^{-1} \quad \mbox{weakly in } L^2(\R^2;\gS_2(L^2(\R))). \label{eq:cv_Q_mk_2}
\end{align}
\end{enumerate}
\label{thmexistence}
\end{theorem}

The proof of Theorem~\ref{thmexistence} can be read in Section \ref{ProofThmExistence}.

\medskip

In the Yukawa case ($m>0$), we are able to characterize the minimizers of~\eqref{minimizationpbyoka}. By Theorem~\ref{thmexistence}, all the minimizers of the problem~\eqref{minimizationpbyoka} share the same density $\rho_{\nu,m} \in \cK$. In view of Proposition~\ref{prop:densityK}, the function $\rho_{\nu,m}$ is in $L^p(\R)+L^2(\R)$ for some $1 < p < 5/3$, thus in $H^{-1}(\R)$ (see Remark~\ref{rem:finiteYukawa}). The Yukawa potential
\begin{equation}
V_{\nu,m}:= \left( -\frac{d^2}{dz^2}+m^2 \right)^{-1} (\rho_{\nu,m}-\nu) = \frac{\rme^{-m|\cdot|}}{m}\star(\rho_{\nu,m}-\nu),
\label{yukawaPotential}
\end{equation}
is therefore well-defined in $H^1(\R)$. In particular, $V_{\nu,m}$ is a continuous function vanishing at infinity. The following result shows that this is sufficient to ensure the uniqueness of the ground-state density matrix in~\eqref{minimizationpbyoka}.

\begin{theorem}[Uniqueness and characterization of the minimizer for the Yukawa case] 
 Let $\nu\in H^{-1}(\bb R)$ and $m > 0$. The minimizer $Q_{\nu,m}$ of the problem~\eqref{minimizationpbyoka} is unique and is the unique solution in $\cK$ to the self-consistent equations:
\begin{equation}
\left\{
\begin{aligned}
\gamma_{\nu,m} &:= \mathds 1_{(-\infty,\epsilon_F]}(T +V_{\nu,m}),\\
V_{\nu,m} &:= \frac{\mathrm{e}^{-m|\cdot|}}{m} \star (\rho_{Q_{\nu,m}}- \nu), \\
Q_{\nu,m} &:= \gamma_{\nu,m} -\gamma_0.
\end{aligned}\right.
\label{self-consist-eq}
\end{equation}
\label{yukawaMinimizerthm}
\end{theorem}
The proof of Theorem~\ref{yukawaMinimizerthm} can be read in Section~\ref{thmYukawaMinimizer}.

\begin{remark}
  Proving that self--consistent equations similar to~\eqref{self-consist-eq} hold for Coulomb interactions is much more challenging. The first step would be to properly define the potential $V_{\nu,0} = \rho_{\nu,0} \star |\cdot|$, as well as the self--adjoint extension of the operator $T+V_{\nu,0}$ (see~\cite{Oliveira2009}). The technique of proof we use in the Yukawa case relies on the fact that the potential is bounded and in~$L^2(\R)$ (although it would be possible to work in more general Lebesgue spaces). It is not obvious at all that $V_{\nu,0}$ satisfies these properties. 
\end{remark}

Let us conclude this section by presenting some numerical simulations illustrating the behavior of the perturbation of the electronic density induced by sharp trenches modeling a capacitor. More precisely, we consider $\nu(z) = -\rho_0\mathds 1_{|z| \leq w}$ for $w>0$. The physical parameters are chosen as $\epsilon_F = 2.0$ and~$w=4$. We refer to~\cite{Cao} for details on how the simulations are performed. We plot in Figures~\ref{FriedelYukawa4} and~\ref{FriedelYukawa2} the total electronic density $\rho_0 + \nu + \rho_{\nu,m}$ for two values of the Yukawa parameter~$m>0$. We can observe Friedel oscillations~\cite{Friedel} in the densities, which can be fitted away from the defect as
\[
\rho_{\nu,m}(z) = \rho_0+a\frac{\cos\left(2\epsilon z + \delta\right)}{|z|^3};
\]
see the values of $a,\delta,\epsilon$ obtained by our fit in the captions of the figures. Remark that the fitted value of~$\epsilon$ is close to the Fermi level, as predicted by~\cite{Friedel}.

\begin{figure}[h!]
\includegraphics[width=0.49\textwidth ]{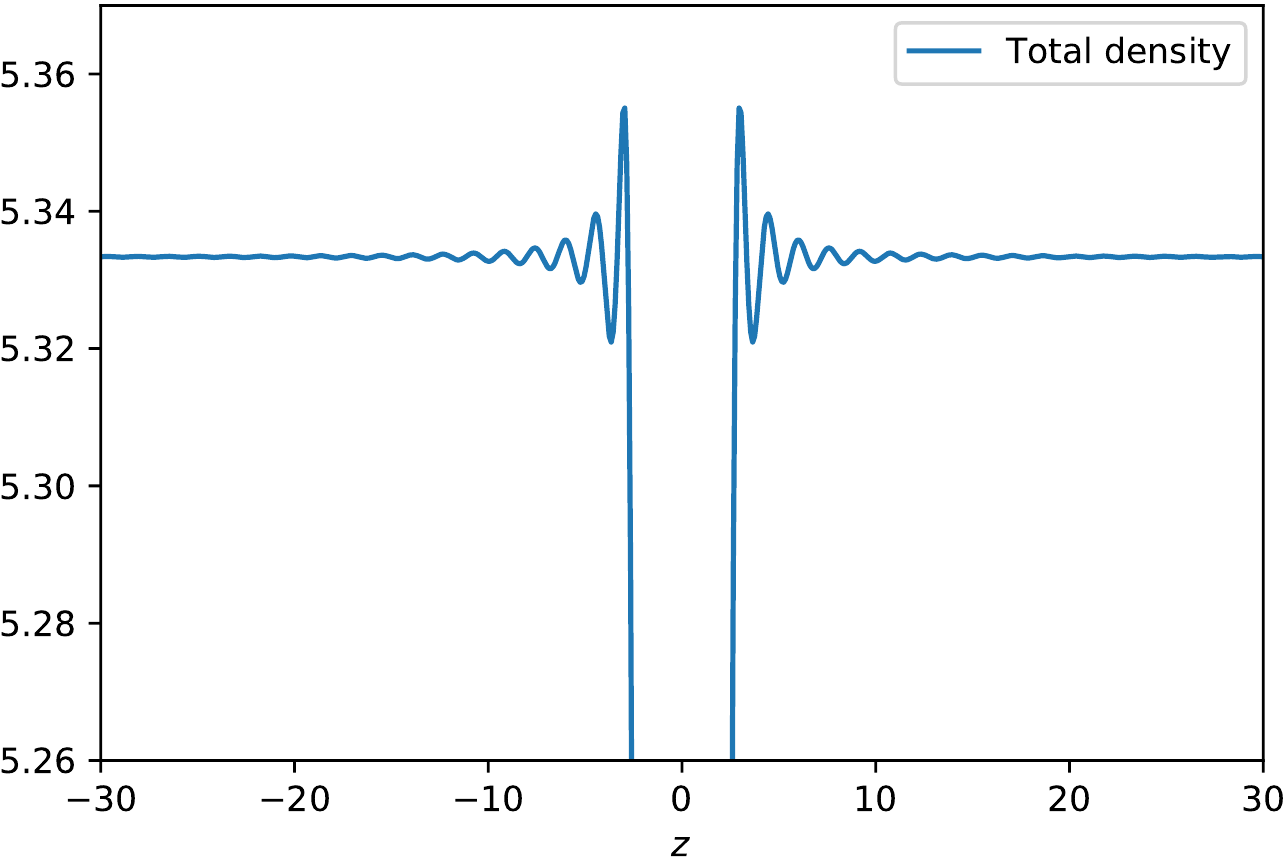}~
\includegraphics[width=0.49\textwidth ]{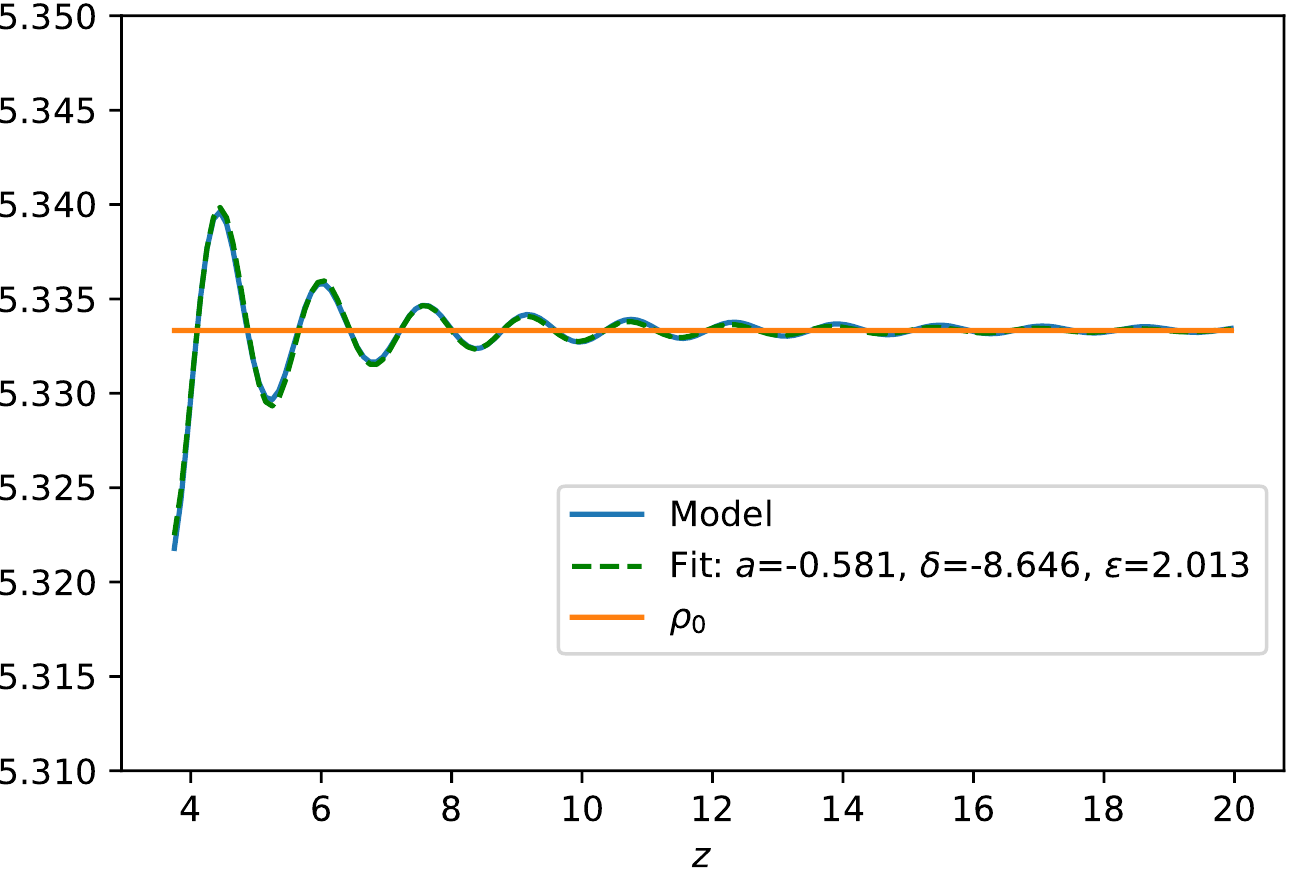}~
\caption{Electronic density $\rho_{\nu,m}$ for $m=4$.}
\label{FriedelYukawa4}
\end{figure}

\begin{figure}[h!]
\includegraphics[width=0.49\textwidth ]{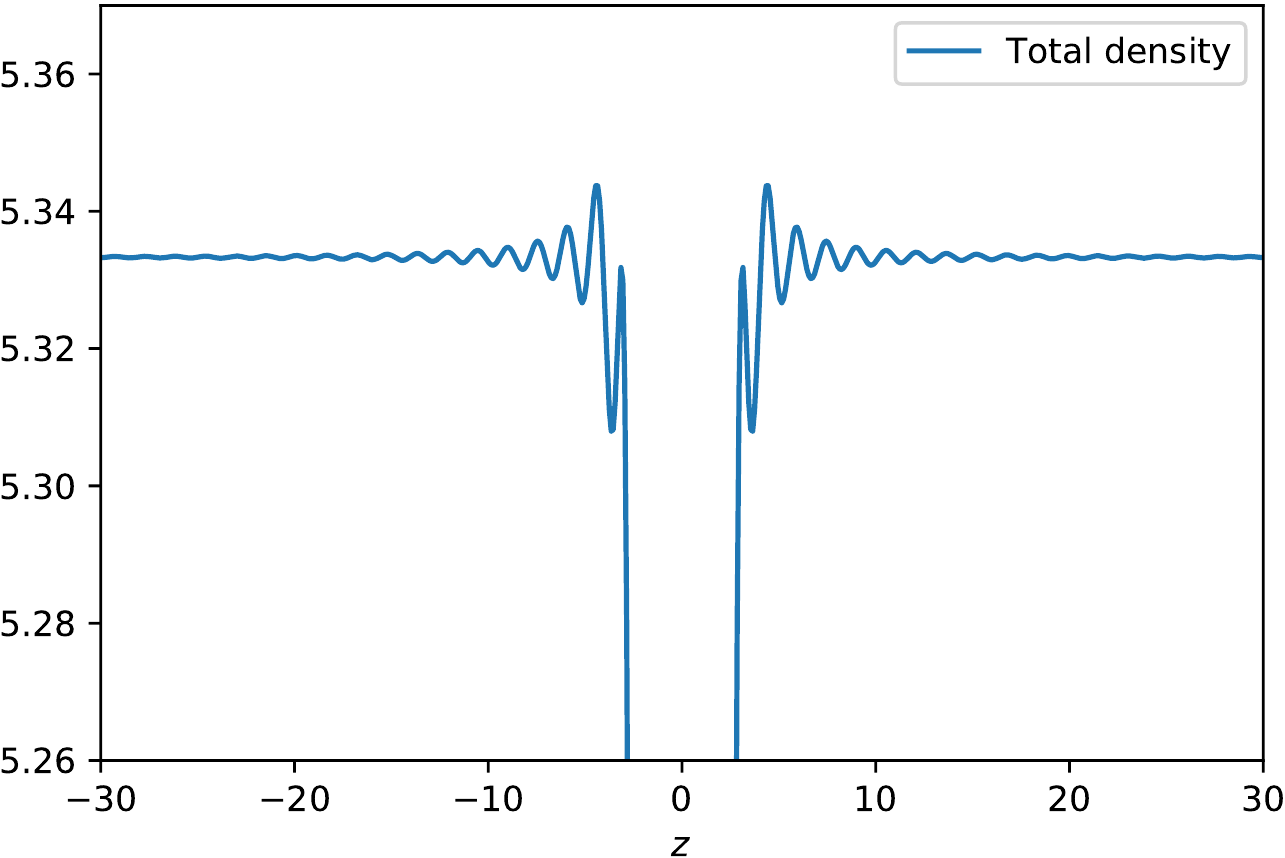}~
\includegraphics[width=0.49\textwidth ]{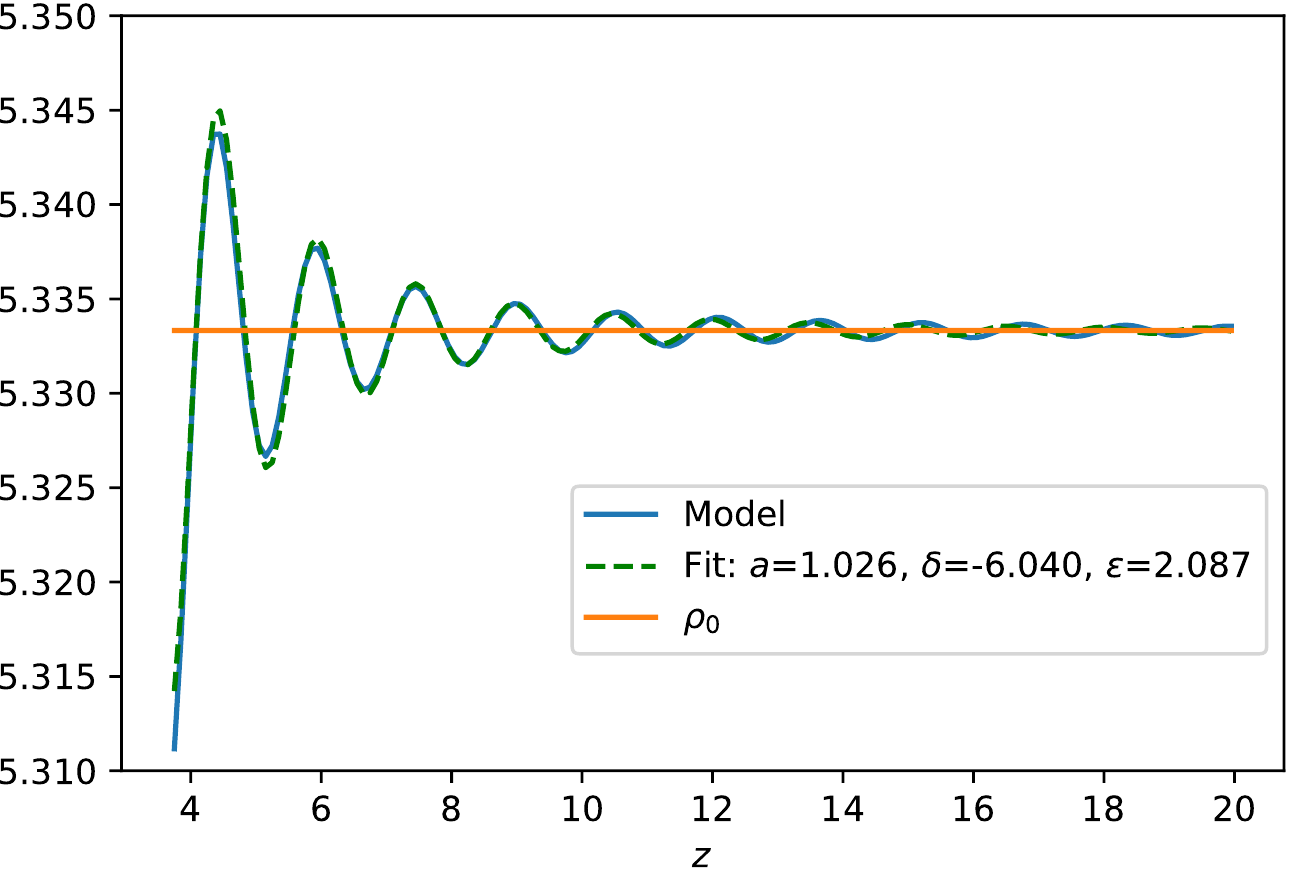}~
\caption{Electronic density $\rho_{\nu,m}$ for $m=2$.}
\label{FriedelYukawa2}
\end{figure}

\section{Proof of the results}
\label{sec:proofs}

Unless otherwise specified, we simply write $\mathfrak{S}_p$ instead of $ \mathfrak{S}_p\left(L^2(\bb R)\right)$ in all the proofs.

\subsection{Proof of Proposition \ref{prop:densityK}}
\label{propoDefofDensity}

The following result shows in which sense the density of operators in $\mathcal{K}$ should be understood.

\begin{prop}(Densities of operators in $\mathcal{K}$)
  \label{prop:densityK}
  Any $Q\in\mathcal{K}$ is locally trace-class, its density $\rho_Q$ is a function of the variable $z$ only, and $\rho_Q \in L^p(\bb R)+ L^2(\bb R)$ for any $1<p<5/3$.

  In addition, for all $1<p<5/3$ and all $c > 0$, there exists two positive constants $\eta_{c,-}, \eta_{p,c,+}$ such that
  \begin{equation}
    \forall Q \in \mathcal{K}, \quad \eta_{c,-} \|\rho_Q^{c,-}\|_{L^2(\R)}^2 + \eta_{p,c,+} \|\rho_Q^{c,+}\|_{L^p(\R)}^p \le \utr\left((T-\epsilon_F)Q\right),
    \label{LTtotalenergybound}
  \end{equation}
  where $\rho_Q = \rho_Q^{c,-}+\rho_Q^{c,+}$ with
  \[
  \rho_Q^{c,-}:=\frac{1}{(2\pi)^2}\int_{\overline{\mathfrak{B}}_{\epsilon_F+c}}\rho_{Q_{q}}\,dq \in L^2(\bb R),  \quad \mbox{and} \quad \rho_Q^{c,+}:= \frac{1}{(2\pi)^2}\int_{\bb R^2\backslash \overline{\mathfrak{B}}_{\epsilon_F+c}}\rho_{Q_{q}}\,dq  \in L^p(\bb R).
  \]
\end{prop}

The proof of Proposition~\ref{prop:densityK} is based on the following technical results, which show that the decomposed kinetic energy of defects actually satisfies Lieb--Thirring-like inequalities~\cite{Lieb2005}. The density associated with the state of the defect can therefore be controlled by the kinetic energy of the defect.

For $q\in \overline{\mathfrak{B}}_{\epsilon_F}$, Lemma \ref{qInMLemma} (resp. \ref{offdiagonal}) provides a lower bound of the densities of diagonal blocks (resp. off-diagonal blocks) of operators in~$\mathcal{K}_{q}$. The proof of these results, obtained by the same techniques as in~\cite{frank2013}, can be read in Section~\ref{lemmaQ} (resp. Section \ref{lemmaOffdiagonal}).

\begin{lemma}
  \label{qInMLemma}
  There exist positive constants $C_{1,{\rm diag}}$, $C_{2,{\rm diag}}$ such that, for all $q\in\overline{\mathfrak{B}}_{\epsilon_F}$ and $Q_{q}\in \mathcal{K}_{q}$,
  \begin{equation}
    \label{qInM}
    \pm \m{Tr}_{L^2(\bb R)}\left(|T_q-\epsilon_F|^{1/2}Q_{q}^{\pm\pm}|T_q-\epsilon_F|^{1/2}\right) \geq C_{1,{\rm diag}}\int_{\bb R} \left|\rho_{Q_{q}^{\pm\pm}}(z)\right|^3 dz + C_{2,{\rm diag}} \sqrt{2\epsilon_F-|q|^2} \int_{\bb R} \rho_{Q_{q}^{\pm\pm}}^2(z)\,dz.
  \end{equation}
\end{lemma}

The absolute value in the integrand of the first integral on the right-hand side is motivated by the fact that $Q_q^{--} \leq 0$, so that $\rho_{Q_q^{--}} \leq 0$.

\begin{lemma}
\label{offdiagonal}
There exist $C_{+-} \in \mathbb{R}_+$ such that, for all $q\in \overline{\mathfrak{B}}_{\epsilon_F}$ and $Q_{q}\in \mathcal{K}_{q}$,
\begin{equation}
\begin{aligned}
& \m{Tr}_{L^2(\bb R)}\left(|T_q-\epsilon_F|^{1/2}\left(Q_{q}^{++}-Q_q^{--}\right)|T_q-\epsilon_F|^{1/2}\right) \\
& \qquad \geq C_{+-} \left(2\epsilon_F-|q|^2\right)^{1/4} \int_{\bb R}\left\lvert\widehat{\rho_{Q_{q}^{\pm\mp}}}(k)\right\rvert^2 \left\lvert |k|-2\sqrt{2\epsilon_F-|q|^2}\right\rvert^{1/2} dk.
\end{aligned}
\label{diag_LT}
\end{equation}
\end{lemma}

For $q\in \bb R^2 \backslash  \overline{\mathfrak{B}}_{\epsilon_F}$, it holds $\gamma_{0,q} \equiv 0$ and  $0 \leq Q_{q} \leq 1$, so that $Q_{q}^{++} \equiv  Q_{q}$ and $ Q_{q}^{--} \equiv 0$. In particular, $\rho_{Q_q}\geq 0$. The following estimate therefore follows from the one dimensional Lieb--Thirring inequality~\cite{Lieb2005}.

\begin{lemma}
  \label{qNotInM}
  There exists a positive constant $C$ such that, for all $ q\in \bb R^2\backslash \overline{\mathfrak{B}}_{\epsilon_F}$ and $Q_{q}\in \mathcal{K}_{q}$,
  \begin{equation}
    \m{Tr}_{L^2(\bb R)}\left(|T_q-\epsilon_F|^{1/2}\left(Q_{q}^{++}-Q_q^{--}\right)|T_q-\epsilon_F|^{1/2}\right) \geq  C\int_{\bb R} \rho_{Q_{q}}^3(z) \,dz + \left(\frac{|q|^2}{2}-\epsilon_F\right)\int_{\bb R} \rho_{Q_{q}}(z) \,dz.
  \end{equation}
\end{lemma}

We are now in position to prove Proposition~\ref{prop:densityK}. Fix $Q\in \mathcal{K}$. Bounds on the densities are obtained by separating the estimates for $q\in \bb R^2 \backslash  \overline{\mathfrak{B}}_{\epsilon_F}$ and $q\in  \overline{\mathfrak{B}}_{\epsilon_F}$. More precisely, defining
\begin{equation}
f(q):= \left||q|^2-2\epsilon_F\right| = \left\{ \begin{aligned}
&|q|^2-2\epsilon_F >  0 \quad \quad \text{if } q\in \bb R^2\backslash  \overline{\mathfrak{B}}_{\epsilon_F},\\
& 2\epsilon_F- |q|^2 \geq  0  \quad \quad \text{if } q\in  \overline{\mathfrak{B}}_{\epsilon_F},
\end{aligned}\right.
\end{equation}
the key ingredient in our proof is the following H\"{o}lder inequality, written in a generic form for real numbers $\beta>1$, $\alpha>0$ and an integration domain $A\subset \bb R^2$:
\begin{equation}
  \label{estimateQNOTINM}
  \begin{aligned}
    &\int_{\bb R}\left(\int_A\left|\rho_{Q_{q}}(z)\right|dq\right)^\beta\,dz \leq \int_{\bb R}\left(\int_{A} \left|\rho_{Q_{q}}(z)\right|^\beta f^{\alpha \beta}(q) \,dq \right)\left(\int_A  f^{-\alpha \beta/(\beta-1)}(q)\,dq\right)^{\beta-1}\,dz.
  \end{aligned}
\end{equation}
We successively consider three situations: total density for $q \in \overline{\mathfrak{B}}_{\epsilon_F}$, density associated with the diagonal blocks of~$Q$ for $q \in  \overline{\mathfrak{B}}_{\epsilon_F}$, and density associated with the off-diagonal blocks of~$Q$ for $q \in  \overline{\mathfrak{B}}_{\epsilon_F}$.

\paragraph{Estimates for the total density on $\bb R^2 \backslash \overline{\mathfrak{B}}_{\epsilon_F}$.}
Lemma~\ref{qNotInM} shows that
\begin{align*}
& \int_{\bb R^2\backslash \overline{\mathfrak{B}}_{\epsilon_F}}\left(C \int_{\bb R} \rho_{Q_{q}}^3(z)\,dz + \left(\frac{|q|^2}{2}-\epsilon_F\right)\int_{\bb R} \rho_{Q_{q}}(z) \,dz\right)dq \\
& \qquad \leq \int_{\bb R^2\backslash \overline{\mathfrak{B}}_{\epsilon_F}} \m{Tr}_{L^2(\bb R)}\left(|T_q-\epsilon_F|^{1/2}\left(Q_{q}^{++}-Q_q^{--}\right)|T_q-\epsilon_F|^{1/2}\right)dq \leq \utr((T-\epsilon_F)Q) < +\infty.
\end{align*}
The above inequality implies that
\begin{equation}
  \label{separateQNInM}
  \int_{\bb R^2\backslash \overline{\mathfrak{B}}_{\epsilon_F}} \int_{\bb R} \rho_{Q_{q}}^3(z) \,dz \,dq \leq \frac1C \utr((T-\epsilon_F)Q),
  \qquad
  \int_{\bb R^2\backslash  \overline{\mathfrak{B}}_{\epsilon_F}}  \int_{\bb R}f(q)\,\rho_{Q_{q}}(z) \,dz \,dq \leq 2 \utr((T-\epsilon_F)Q).
\end{equation}
In order to obtain bounds on the density, we separate the integration domain in two pieces: large values of~$|q|$, and values close to $\{|q|^2=2\epsilon_F\}$ where $f(q)$ vanishes. More precisely, for a given $c>0$, we decompose the integration domain as $\bb R^2 \backslash \overline{\mathfrak{B}}_{\epsilon_F} = \left(\bb R^2\backslash \overline{\mathfrak{B}}_{\epsilon_F+c}\right) \cup \left(\overline{\mathfrak{B}}_{\epsilon_F+c} \backslash  \overline{\mathfrak{B}}_{\epsilon_F}\right)$.

We first set $A = \bb R^2\backslash \overline{\mathfrak{B}}_{\epsilon_F+c} $ in~\eqref{estimateQNOTINM}. The last integral in this inequality reads
\[
\int_A  f^{-\alpha \beta/(\beta-1)}(q)\,dq =\int_{\bb R^2\backslash \overline{\mathfrak{B}}_{\epsilon_F+c} } \left(|q|^2-2\epsilon_F\right)^{-\alpha \beta/(\beta-1)}\,dq = 2\pi \int_{\sqrt{2(\epsilon_F+c)}}^{+\infty} r\left(r^2-2\epsilon_F\right)^{-\alpha \beta/(\beta-1)}\,dr.
\]
The latter integral is finite if and only if $\alpha \beta/(\beta-1)>1$. Moreover, a H\"older inequality combined with~\eqref{separateQNInM} implies that the following integral is finite for $a>1$:
\[
\begin{aligned}
&\int_{\bb R^2\backslash \overline{\mathfrak{B}}_{\epsilon_F+c}} \int_{\bb R}\rho_{Q_{q}}^{3-\frac{2}{a}}(z)f^{1/a}(q)\,dz\,dq  \\
&\quad \leq \left(\int_{\bb R^2\backslash \overline{\mathfrak{B}}_{\epsilon_F+c}}  \int_{\bb R}f(q)\rho_{Q_{q}}(z)\,dz\,dq\right)^{1/a} \left(\int_{\bb R^2\backslash \overline{\mathfrak{B}}_{\epsilon_F+c}} \int_{\bb R} \rho_{Q_{q}}^3(z)\,dz\,dq\right)^{(a-1)/a}\\
&\quad \leq \max\left(2,\frac{1}{C}\right) \utr((T-\epsilon_F)Q).
\end{aligned}
\]
In view of~\eqref{estimateQNOTINM}, this suggests taking $\alpha \beta = 1/a$ and $\beta = 3-2/a$. The condition $\alpha \beta/(\beta-1)>1$ can be rephrased as $1/(2a-2)>1$. The latter inequality is satisfied for $1<a<3/2$, which is equivalent to $1< \beta < 5/3$. For the latter choice, inequality~\eqref{estimateQNOTINM} combined with~\eqref{separateQNInM} then shows that there exists $K_{\beta, \bb R^2 \backslash \overline{\mathfrak{B}}_{\epsilon_F+c}} \in \mathbb{R}_+$ such that
\begin{equation}
  \label{slash1}
  \int_{\bb R}\left(\int_{ \bb R^2 \backslash \overline{\mathfrak{B}}_{\epsilon_F+c}} \rho_{Q_{q}}(z)\,dq\right)^\beta\,dz  \leq K_{\beta, \bb R^2 \backslash \overline{\mathfrak{B}}_{\epsilon_F+c}} \utr((T-\epsilon_F)Q).
\end{equation}

We next set $A = \overline{\mathfrak{B}}_{\epsilon_F+c}\backslash\overline{\mathfrak{B}}_{\epsilon_F}$ in~\eqref{estimateQNOTINM} and follow the same strategy as in the previous case. We still take $\alpha \beta = 1/a$ and $\beta = 3-2/a$, but need now that $\alpha \beta/(\beta-1)<1$ in order to ensure that the last integral in~\eqref{estimateQNOTINM} is finite. This condition is equivalent to $a>3/2$, \textit{i.e.} $5/3<\beta<3$. We therefore consider $\beta = 2$. The inequality~\eqref{estimateQNOTINM} combined with~\eqref{separateQNInM} then shows that there exists $K_{\overline{\mathfrak{B}}_{\epsilon_F+c}\backslash \overline{\mathfrak{B}}_{\epsilon_F}}\in \mathbb{R}_+$ such that
\begin{equation}
  \label{firstPieceDensityEstimationTotal}
  \int_{\bb R}\left(\int_{\overline{\mathfrak{B}}_{\epsilon_F+c}\backslash \overline{\mathfrak{B}}_{\epsilon_F}} \rho_{Q_{q}}(z)\,dq\right)^2\,dz \leq K_{\overline{\mathfrak{B}}_{\epsilon_F+c}\backslash \overline{\mathfrak{B}}_{\epsilon_F}} \utr((T-\epsilon_F)Q).
\end{equation}

\paragraph{Estimates for $q\in \overline{\mathfrak{B}}_{\epsilon_F}$, diagonal blocks.} We write the estimates for $\rho_{Q^{++}_q}$ only, the bounds for $\rho_{Q^{--}_q}$ being similar. Lemma~\ref{qInMLemma} shows that
\begin{equation}
  \label{qinBestim}
  \begin{aligned}
    \int_{\overline{\mathfrak{B}}_{\epsilon_F}}\int_{\bb R} \rho_{Q_{q}^{++}}^3(z) \,dz\,dq  \leq \frac{1}{C_{1,{\rm diag}}} \utr((T-\epsilon_F)Q), \\
    \int_{\overline{\mathfrak{B}}_{\epsilon_F}}\int_{\bb R} f^{1/2}(q) \rho_{Q_{q}^{++}}^2(z) \,dz \,dq \leq \frac{1}{C_{2,{\rm diag}}} \utr((T-\epsilon_F)Q),
  \end{aligned}
\end{equation}
so that, by a H\"older inequality for $a>1$,
\begin{equation}
  \label{halfHolder}
  \begin{aligned}
    \int_{\overline{\mathfrak{B}}_{\epsilon_F}} \int_{\bb R}\rho_{Q_{q}^{++}}^{3-1/a}(z)f^{1/(2a)}(q)\,dz\,dq & \leq \left(\int_{\overline{\mathfrak{B}}_{\epsilon_F}}  \int_{\bb R}f^{1/2}(q)\rho_{Q_{q}^{++}}^2(z)\,dz\,dq\right)^{1/a} \left(\int_{\overline{\mathfrak{B}}_{\epsilon_F}} \int_{\bb R} \rho_{Q_{q}^{++}}^3(z)\,dz\,dq\right)^{(a-1)/a} \\
    & \leq C_{1,{\rm diag}}^{1/a-1}C_{2,{\rm diag}}^{-1/a} \utr((T-\epsilon_F)Q)<+\infty.
  \end{aligned}
\end{equation}
We now consider~\eqref{estimateQNOTINM} with $A = \overline{\mathfrak{B}}_{\epsilon_F}$ and $\rho_{Q_q}$ replaced by $\rho_{Q_{q}^{++}}$. The previous inequality suggests choosing $\beta =3-1/a$ and $\alpha \beta = 1/(2a)$. The last integral in~\eqref{estimateQNOTINM} is finite if and only if $\alpha \beta /(\beta-1)< 1$, which is equivalent to $1/(4a-2) < 1$, \textit{i.e.} $a>3/4$ and $5/3<\beta<3$. We therefore choose $\beta = 2$. The inequality~\eqref{estimateQNOTINM} combined with~\eqref{separateQNInM} then shows that there exists $K_{\overline{\mathfrak{B}}_{\epsilon_F},++} \in \mathbb{R}_+$ such that
\begin{equation}
  \label{diagoestiTotal}
  \int_{\bb R}\left(\int_{ \overline{\mathfrak{B}}_{\epsilon_F}} \rho_{Q_{q}^{++}}(z)\,dq\right)^2\,dz  \leq K_{\overline{\mathfrak{B}}_{\epsilon_F},++} \utr((T-\epsilon_F)Q).
\end{equation}

\paragraph{Estimates for $q\in\overline{\mathfrak{B}}_{\epsilon_F}$, off-diagonal blocks.}
Define, for $k \in \mathbb{R}$ and $q \in \overline{\mathfrak{B}}_{\epsilon_F}$,
\[
g(k,q) = \left\lvert |k|-2\sqrt{2\epsilon_F-|q|^2}\right\rvert^{1/2} \left(2\epsilon_F-|q|^2\right)^{1/4}.
\]
In view of Lemma~\ref{offdiagonal} and the inequality $\left\lvert\widehat{\rho_{Q_{q}^{+-}}}(k)+\widehat{\rho_{Q_{q}^{-+}}}(k) \right\rvert =\left\lvert 2\mathfrak{R}\widehat{\rho_{Q_{q}^{+-}}}(k)\right\rvert\leq 2\left\lvert \widehat{\rho_{Q_{q}^{+-}}}(k)\right\rvert$, it holds
\begin{equation}
\int_{\overline{\mathfrak{B}}_{\epsilon_F}}\int_{\bb R}\left\lvert\widehat{\rho_{Q_{q}^{+-}}}(k)+\widehat{\rho_{Q_{q}^{-+}}}(k)\right\rvert^2g(k,q)\, dk\, dq \leq \frac{4}{C_{+-}} \utr((T-\epsilon_F)Q).
\label{offDia_tmp}
\end{equation}
Note that, by the change of variables $t=|q|/\sqrt{2\epsilon_F}$ and $w=\sqrt{1-t^2}$,
\[
\begin{aligned}
G(k) := \int_{\overline{\mathfrak{B}}_{\epsilon_F}} g(k,q)^{-1}dq & = 2\pi\int_{0}^1\frac{\sqrt{2\epsilon_F} \, t}{\left\lvert (2\epsilon_F)^{-1/2}|k|-2\sqrt{1-t^2}\right\rvert^{1/2} \left(1-t^2\right)^{1/4}} \, dt \\
& = 2\pi\int_{0}^1\frac{\sqrt{2\epsilon_F}\sqrt{w}}{\left\lvert (2\epsilon_F)^{-1/2}|k|-2w\right\rvert^{1/2} } \, dw,
\end{aligned}
\]
from which it is easy to see that $G$ is a bounded positive function tending to 0 as $|k|\to\infty$. Define
\[
\rho_{Q^{\mathrm{offdiag}}}(z):= \frac{1}{(2\pi)^2} \int_{\overline{\mathfrak{B}}_{\epsilon_F}}\left(\rho_{Q_{q}^{+-}}(z)+\rho_{Q_{q}^{-+}}(z)\right)dq.
\]
Using the isometry property of the Fourier transform, the Cauchy--Schwarz inequality and~\eqref{offDia_tmp}, we obtain the following bound on the $L^2$ norm of $\rho_{Q^{\mathrm{offdiag}}}$:
\begin{align}
  \left\|\rho_{Q^{\mathrm{offdiag}}}\right\|_{L^2(\mathbb{R})}^2 = \int_{\bb R}\left|\widehat{\rho_{Q^{\mathrm{offdiag}}}}(k) \right|^2dk &= \frac{1}{(2\pi)^4}\int_{\bb R}\left| \int_{\overline{\mathfrak{B}}_{\epsilon_F}}\left(\widehat{\rho_{Q_{q}^{+-}}}(k)+\widehat{\rho_{Q_{q}^{-+}}}(k)\right) dq\right|^2 dk \nonumber \\
    &\leq  \frac{1}{(2\pi)^4}\int_{\bb R}\left[\int_{\overline{\mathfrak{B}}_{\epsilon_F}}\left\lvert\widehat{\rho_{Q_{q}^{+-}}}(k)+\widehat{\rho_{Q_{q}^{-+}}}(k)\right\rvert^2 g(k,q) \, dq \right]G(k)\,dk \nonumber \\
    & \leq  \frac{1}{(2\pi)^4}\norm{G}_{L^{\infty}(\bb R)} \int_{\bb R}\int_{\overline{\mathfrak{B}}_{\epsilon_F}}\left\lvert\widehat{\rho_{Q_{q}^{+-}}}(k)+\widehat{\rho_{Q_{q}^{-+}}}(k)\right\rvert^2 g(k,q)\,dq\,dk \nonumber \\
    & \leq \frac{1}{4 \pi^4 C_{+-}}\norm{G}_{L^{\infty}(\bb R)} \utr((T-\epsilon_F)Q). \label{offD}
  \end{align}

\paragraph{Conclusion of the proof.}
The estimate~\eqref{LTtotalenergybound} finally follows from~\eqref{slash1}, \eqref{firstPieceDensityEstimationTotal}, \eqref{diagoestiTotal} and \eqref{offD}.

\subsection{Proof of Theorem \ref{thmexistence}}
\label{ProofThmExistence}

We first show in Section~\ref{CoulombSpaceLemmaSec} that the minimization set $\mathcal{F}_\nu$ for Coulomb interactions is not empty. We then prove in Section~\ref{thm:ExistenceMinimizerSec} the existence of minimizers for Yukawa and Coulomb interactions. The uniqueness of the densities relies on a technical result whose proof is postponed to Section~\ref{consistencyOfDensitySec}. Finally, we show in Section~\ref{sec:proof_cv_energies} that the Yukawa ground state converges to the Coulomb ground state when $m \to 0$.

\subsubsection{The set $\mathcal{F}_\nu$ is not empty}
\label{CoulombSpaceLemmaSec}

We prove in this section that the set ${\mathcal F}_\nu = \left\{Q\in \mathcal{K} \mid  \rho_Q - \nu \in \ud{\mathcal{C}}\right\}$ is non-empty for all $\nu \in L^1(\R)$ such that $|\cdot|\nu(\cdot) \in L^1(\R)$. We do so by explicitly constructing an element in~${\mathcal F}_\nu$. We distinguish the cases $\kappa \geq 0$ and $\kappa<0$, where $\kappa$ is the total charge per unit area of the defect:
\[
\kappa= \int_{-\infty}^{+\infty}\nu(z)\,dz.
\]

\paragraph{Non-negative total charge per unit area.} Consider first the case when $\kappa\geq 0$. We introduce an even cut-off function $\chi \in C_c^{\infty}(\bb R)$ such that
\[
0\leq \chi\leq 1, \qquad \int_{\mathbb{R}} \chi^2 = 1.
\]
For a parameter $\mu \in (\epsilon_F,+\infty)$ to be specified and almost all $q\in \bb R^2$, we then define the self-adjoint operator
\begin{equation}
Q_{\mu}:= \mathcal{U}\left(\int_{\bb R^2}^{\oplus}Q_{\mu,q}\,dq \right)\mathcal{U}^{-1},
\qquad
Q_{\mu,q} :=\mathds 1_{(\epsilon_F,\mu]}(T_q) \chi^2 \mathds 1_{(\epsilon_F,\mu]}(T_q).
\label{Qkappapositive}
\end{equation}
Note first that $Q_{\mu,q} = 0$ when $|q|^2>2\mu$. The operator inequality $0\leq \mathds 1_{(\epsilon_F,\mu]}(T_q) \leq  \mathds 1_{(\epsilon_F,+\infty)}(T_q)=1-\gamma_{0,q} $ also implies that $0 \leq Q_{\mu,q} \equiv Q_{\mu,q}^{++}\leq 1-\gamma_{0,q}$. Therefore, $q\mapsto Q_{\mu,q}\in L^{\infty}(\bb R^2;\mathcal{S}(L^2(\bb R)))$. Moreover, the Kato--Seiler--Simon inequality~\cite[Theorem 4.1]{BarryS}
 \begin{equation}
   \label{KatoSeilerSimon}
   \forall p\geq 2, \quad \norm{f(-i\nabla)g}_{\mathfrak{S}_p(L^2(\R))}\leq (2\pi)^{-1/p}\norm{g}_{L^p}\norm{f}_{L^p},
 \end{equation}
shows that
\begin{equation}
  \label{eq:decomposition_F_nu_non_empty}
  |T_q-\epsilon_F|^{1/2}Q_{\mu,q}|T_q-\epsilon_F|^{1/2} =\left(|T_q-\epsilon_F|^{1/2}\mathds 1_{(\epsilon_F,\mu]}(T_q) \chi \right)\left(\chi\mathds 1_{(\epsilon_F,\mu]}(T_q)|T_q-\epsilon_F|^{1/2}\right) \in\mathfrak{S}_1
\end{equation}
as the product of two Hilbert--Schmidt operators. We have therefore proven at this stage that $Q_{\mu,q}\in\mathcal{K}_q$ for all $q\in \bb R^2$.

In addition, since the kernel of the operator $|T_q-\epsilon_F|^{1/2}\mathds 1_{(\epsilon_F,\mu]}(T_q) \chi$ is $(z,z') \mapsto g_q(z-z')\chi(z')$ with
\[
g_q(z) = \frac{1}{2\pi} \int_\R \left|\frac{k^2+|q|^2}{2} - \epsilon_F \right|^{1/2} \mathds 1\left\{ \frac{k^2+|q|^2}{2} \in (\epsilon_F,\mu] \right\} \rme^{\rmi k z} \, dk,
\]
we obtain
\begin{equation}
\label{eq:trace_bar_finite_test_state}
\begin{aligned}
\utr\left((T-\epsilon_F)Q_{\mu}\right) & = \frac{1}{(2\pi)^2}\int_{\bb R^2} \m{Tr}_{L^2(\bb R)}\left(|T_q-\epsilon_F|^{1/2}Q_{\mu,q}|T_q-\epsilon_F|^{1/2}\right) dq = \frac{1}{(2\pi)^2} \int_{\bb R^2}\left(\int_\mathbb{R} \chi^2\right)\left(\int_\mathbb{R} \left|g_q\right|^2\right) dq\\
& = \frac{1}{(2\pi)^3} \int_{\bb R^2} \int_{\epsilon_F < \frac{|q|^2}{2}+\frac{k^2}{2} \leq \mu} \left(\frac{k^2+|q|^2}{2} - \epsilon_F \right) dk \, dq < \infty,
\end{aligned}
\end{equation}
which shows that $Q_{\mu}\in \mathcal{K}$.

It remains to prove that $\rho_{Q_{\mu}} - \nu \in \ud{\mathcal{C}}$. First, it easily follows from~\eqref{Qkappapositive} that $\rho_{Q_{\mu}}$ is smooth and compactly supported. By computations similar to the ones used to establish~\eqref{eq:trace_bar_finite_test_state}, and noting that $Q_{\mu ,q} \in \mathfrak{S}_1$ by a decomposition similar to~\eqref{eq:decomposition_F_nu_non_empty},
\[
\utr(Q_\mu) = \int_\R \rho_{Q_\mu} = \frac{1}{(2\pi)^3}\int_{\bb R^2} \int_{\epsilon_F < \frac{|q|^2}{2}+\frac{k^2}{2} \leq \mu} \,dk\,dq = \frac{\sqrt{2}}{3\pi^2}\left(\mu^{3/2}-\epsilon_F^{3/2}\right).
\]
There exists therefore a (unique) value~$\mu(\kappa)$ such that $\int_{\bb R}\rho_{Q_{\mu(\kappa)}}(z)\,dz = \kappa $. The latter equality is equivalent to~$\widehat{\rho_{Q_{\mu(\kappa)}}}(0)-\widehat{\nu}(0)=0$. Since $\widehat{\rho_{Q_{\mu(\kappa )}}}-\widehat{\nu}$ is $C^1$ and bounded, we can therefore conclude that the function $k \mapsto |k|^{-1} (\widehat{\rho_{Q_{\mu(\kappa)}}}(k)-\widehat{\nu}(k))$ is in $L^2(\mathbb{R})$, \textit{i.e.} $\rho_{Q_{\mu}} - \nu \in \ud{\mathcal{C}}$. This allows to conclude that $Q_{\mu(\kappa)} \in \mathcal{F}_{\nu}$, and so $\mathcal{F}_{\nu}$ is not empty.

\paragraph{Negative total charge per unit area.} We now consider the case when $\kappa <0$. We define the following self-adjoint operator, for a parameter $\alpha \in (0,+\infty)$ to be specified later on:
\[
Q_{\alpha} := \mathcal{U}\left(\int_{\bb R^2}^{\oplus}Q_{\alpha,q}\,dq \right)\mathcal{U}^{-1},
\qquad
Q_{\alpha,q} := - \alpha \gamma_{0,q} \chi^2 \gamma_{0,q}.
\]
It is easy to see that $-\gamma_{0,q}\leq Q_{\alpha,q} = Q_{\alpha,q}^{--} \leq 0$ and that $Q_{\alpha,q} = 0$ when $|q|^2>2\epsilon_F$, so that $q\mapsto Q_{\alpha,q}\in L^{\infty}(\bb R^2;\mathcal{S}(L^2(\bb R)))$. Moreover, $|T_q-\epsilon_F|^{1/2}Q_{\alpha,q}\in\mathfrak{S}_2$ by the Kato--Seiler--Simon inequality~\eqref{KatoSeilerSimon}, so that $|T_q-\epsilon_F|^{1/2}Q_{\alpha,q}|T_q-\epsilon_F|^{1/2}\in\mathfrak{S}_1$. This shows that $Q_{\alpha,q}\in\mathcal{K}_q$ for all $q\in \bb R^2$. In addition, by computations similar to the ones performed for the case $\kappa \geq 0$,
\[
\utr\left((T-\epsilon_F)Q_{\mu}\right) = \frac{\alpha}{(2\pi)^3} \int_{\bb R^2} \int_{0 \leq k^2 + |q|^2 \leq 2\epsilon_F} \left(\epsilon_F - \frac{k^2+|q|^2}{2}\right) dk \, dq < \infty,
\]
so that $Q_{\alpha}\in \mathcal{K}$.

It can be shown same as for the case $\kappa \geq 0$ that $\rho_{Q_\alpha}$ is smooth and compactly supported. Moreover,
\[
\int_{\mathbb{R}} \rho_{Q_\alpha} = -\frac{\alpha}{(2\pi)^3} \int_{\bb R^2} \int_{\frac{|q|^2}{2}+\frac{k^2}{2} \leq \epsilon_F} \,dk\,dq = -\frac{\alpha \sqrt{2} \epsilon_F^{3/2}}{3\pi^2},
\]
so that the choice $\alpha(\kappa) = 3\pi^2 \epsilon_{F}^{-3/2} |\kappa| / \sqrt{2}$ ensures that $\int_{\mathbb{R}} \rho_{Q_{\alpha(\kappa)}} = \kappa$. We can then conclude, by the same reasoning as for the case $\kappa \geq 0$, that $\rho_{Q_{\alpha(\kappa)}} - \nu \in \ud{\mathcal{C}}$, and finally that $\mathcal{F}_{\nu}$ is not empty.

\subsubsection{Existence of ground state for Yukawa (resp. Coulomb) interactions}
\label{thm:ExistenceMinimizerSec}

We write a detailed proof for Coulomb interactions, the result for Yukawa interactions following the same lines. We first construct a candidate minimizer $\overline{Q}$ by as the limit of some minimizing sequence for~\eqref{minimizationpbyoka} and next show that $\overline{Q}$ is an admissible state (\textit{i.e.} $\overline{Q} \in \mathcal{K}$). We finally prove that $\overline{Q}$ is a minimizer and that all minimizers share the same density.

\paragraph{Construction of a candidate minimizer.}
It is easy to see that the functional $\mathcal{E}_{\nu,0}(Q)$ is well defined on the non-empty set~${\mathcal F}_\nu$. Consider a minimizing sequence $\left\{Q_n \right\}_{n\geq 1} \subset \mathcal{F}_\nu$. There exists $C \in \mathbb{R}_+$ such that
\begin{equation}
  \label{uniformBound}
  \forall  n\geq 1, \quad \utr\left(\left(T-\epsilon_F\right)Q_n\right)\leq C, \quad \ud{D} (\rho_{Q_n} - \nu,\rho_{Q_n}-\nu) \leq C.
\end{equation}
Consider any exponent $1 < p < 5/3$. Inequality~\eqref{LTtotalenergybound} shows that the sequence of densities $\{\rho_{Q_{n}}\}_{n\geq 1} $ is uniformly bounded in $L^2(\bb R) + L^p(\bb R)$. Up to extraction, there exist
\[
\overline{Q} := \mathcal{U}^{-1}\left(\int_{\bb R^2}^{\oplus} \overline{Q}_{q}\,dq\right)\mathcal{U},
\qquad
\overline{\overline{\rho}}_Q \in L^2(\bb R) + L^p(\bb R),
\qquad
\widetilde \rho_{Q}-\nu \in \ud{\mathcal{C}}
\]
such that $\rho_{Q_n} \rightharpoonup \overline{\overline{\rho}}_{Q}$ weakly in $L^2(\bb R)+L^p(\bb R)$ and $\rho_{Q_n}-\nu \rightharpoonup \widetilde \rho_{Q}-\nu$ weakly in $\ud{\mathcal{C}}$, while $-\gamma_0 \leq \overline{Q} \leq 1-\gamma_0$ 
and $Q_n \stackrel{\ast}{\rightharpoonup} \overline{Q}$ in the following sense:
\begin{itemize}
\item for any operator-valued function $q\mapsto U_{q} \in L^1\left(\bb R^2; \mathfrak{S}_1\right)$,
  \begin{equation}
    \label{faible1_bis}
    \int_{\bb R^2}\m{Tr}_{L^2(\bb R)}\left(U_{q}Q_{n,q}\right)\,dq \toinfty \int_{\bb R^2}\m{Tr}_{L^2(\bb R)}\left(U_{q}\overline{Q}_{q}\right)dq;
  \end{equation}
\item for any operator-valued function $q\mapsto M_{q} \in L^2\left(\bb R^2; \mathfrak{S}_2\right)$,
  \begin{equation}
    \label{fabile4}
    \int_{\bb R^2}\m{Tr}_{L^2(\bb R)}\left(|T_q-\epsilon_F|^{1/2}Q_{n,q}M_{q}\right)\,dq \toinfty \int_{\bb R^2}\m{Tr}_{L^2(\bb R)}\left(|T_q-\epsilon_F|^{1/2}\overline{Q}_{q}M_{q}\right) dq;
  \end{equation}
  \begin{equation}
    \label{fabile2}
    \int_{\bb R^2}\m{Tr}_{L^2(\bb R)}\left(M_{q}|T_q-\epsilon_F|^{1/2}Q_{n,q}\right)\,dq \toinfty \int_{\bb R^2}\m{Tr}_{L^2(\bb R)}\left(M_{q}|T_q-\epsilon_F|^{1/2}\overline{Q}_{q}\right) dq;
  \end{equation}
  \begin{equation}
    \label{fabile3}
    \int_{\bb R^2}\m{Tr}_{L^2(\bb R)}\left(M_{q}Q_{n,q}|T_q-\epsilon_F|^{1/2}\right)\,dq \toinfty \int_{\bb R^2}\m{Tr}_{L^2(\bb R)}\left(M_{q}\overline{Q}_{q}|T_q-\epsilon_F|^{1/2}\right) dq.
  \end{equation}
\end{itemize}
The weak-$\ast$ convergence~\eqref{faible1_bis} is a consequence of the fact that the sequence $\{Q_n\}_{n \geq 1}$ is uniformly bounded in $L^{\infty}(\bb R^2;\mathcal{L}(L^2(\bb R)))$, whose pre-dual is $L^1\left(\bb R^2; \mathfrak{S}_1\right)$. The weak convergences~\eqref{fabile4} to~\eqref{fabile3} are a consequence of the inequality
\begin{equation}
  \label{infaible2}
  \int_{\bb R^2}\bnorm{ |T_q-\epsilon_F|^{1/2}Q_{n,q}}_{\mathfrak{S}_2}^2\,dq
  \leq \int_{\bb R^2}\m{Tr}_{L^2(\bb R)} \left(|T_q-\epsilon_F|^{1/2}(Q_{n,q}^{++}-Q_{n,q}^{--}) |T_q-\epsilon_F|^{1/2}\right)\,dq \leq C,
\end{equation}
which shows that the sequence of operator-valued functions $q\mapsto  |T_q-\epsilon_F|^{1/2}Q_{n,q}$ and $q\mapsto Q_{n,q}|T_q-\epsilon_F|^{1/2}$ are uniformly bounded in the Hilbert space $L^2\left(\bb R^2; \mathfrak{S}_2\right)$.

\paragraph{The state $\overline{Q}$ belongs to $\mathcal{K}$.}
Note first that the weak convergence of $\left\{q \mapsto \left| T_q-\epsilon_F\right|^{1/2}Q_{n,q}\right\}_{n\geq 1}$ to $q\mapsto \left| T_q-\epsilon_F\right|^{1/2}Q_{q} $ in $L^2(\bb R^2;\mathfrak{S}_2)$ implies
\[
\int_{\bb R^2}\bnorm{ |T_q-\epsilon_F|^{1/2}\overline{Q}_{q}}_{\mathfrak{S}_2}^2\,dq\leq \liminf_{n\to\infty}\int_{\bb R^2}\bnorm{ |T_q-\epsilon_F|^{1/2}Q_{n,q}}_{\mathfrak{S}_2}^2\,dq.
\]
It is therefore sufficient to show that $\utr\left(\left(T-\epsilon_F\right)\overline{Q}\right)<\infty$ in order to conclude that $\overline{Q} \in \mathcal{K}$. Consider to this end an orthonormal basis $\left\{\phi_i\right\}_{i\in \bb N}\subset H^1(\bb R)$ of $L^2(\bb R)$, and define, for $N \geq 1$ and $R > 0$, the family of operators
\[
M_{q,R}^N:= |T_q-\epsilon_F|^{1/2}\left(\sum_{i=1}^N \Ket{\phi_i}\Bra{\phi_i}\right) g_R(q),
\]
where
\[
g_R(q) := \left\{
\begin{aligned}
  1\quad \text{if $|q|^2<R$}\\
  \frac{1}{(1+|q|^2)^2}\quad  \text{if $|q|^2\geq R$}.
\end{aligned}
\right.
\]
A simple computation shows that $q\mapsto M_{q,R}^N$ is in $L^2(\bb R^2;\mathfrak{S}_2)$. Using~\eqref{fabile4} with $q\mapsto (1-\gamma_{0,q})M_{q,R}^N(1-\gamma_{0,q})$,
\[
\begin{aligned}
  0 \leq &\int_{\bb R^2}\m{Tr}_{L^2(\bb R)}\left(|T_q-\epsilon_F|^{1/2}\overline{Q}_{q}|T_q-\epsilon_F|^{1/2}(1-\gamma_{0,q})M_{q,R}^N(1-\gamma_{0,q}) \right)\,dq\\
  &= \int_{\bb R^2}\m{Tr}_{L^2(\bb R)}\left(|T_q-\epsilon_F|^{1/2}\overline{Q}_{q}^{++}|T_q-\epsilon_F|^{1/2}\left(\sum_{i=1}^N \Ket{\phi_i}\Bra{\phi_i}\right) \right) g_R(q)\,dq \\
  & =  \lim_{n\to\infty}\int_{\bb R^2}\sum_{i=1}^N\left\langle \phi_i\left|  |T_q-\epsilon_F|^{1/2}Q_{n,q}^{++}|T_q-\epsilon_F|^{1/2}\right| \phi_i \right\rangle g_R(q)\, dq\\
  &  \leq \liminf_{n\to\infty}\int_{\bb R^2}\m{Tr}_{L^2(\bb R)}\left(|T_q-\epsilon_F|^{1/2}Q_{n,q}^{++}|T_q-\epsilon_F|^{1/2}\right) dq \leq C,\\
\end{aligned}
\]
where the last inequality is a consequence of the uniform bound~\eqref{uniformBound}. We can then pass to the limits $R \to +\infty$ and $N \to +\infty$ with Fatou's Lemma and get
\[
0 \leq \int_{\bb R^2}\m{Tr}_{L^2(\bb R)}\left(|T_q-\epsilon_F|^{1/2}\overline{Q}_{q}^{++}|T_q-\epsilon_F|^{1/2}\right)dq
\leq \liminf_{n \to +\infty} \int_{\bb R^2}\m{Tr}_{L^2(\bb R)}\left(|T_q-\epsilon_F|^{1/2}\overline{Q}_{n,q}^{++}|T_q-\epsilon_F|^{1/2}\right) dq.
\]
A similar reasoning shows that
\[
0\leq -\int_{\bb R^2}\m{Tr}_{L^2(\bb R)}\left(|T_q-\epsilon_F|^{1/2}\overline{Q}_{q}^{--}|T_q-\epsilon_F|^{1/2} \right)dq\leq -\liminf_{n\to\infty} \int_{\bb R^2}\m{Tr}_{L^2(\bb R)}\left(|T_q-\epsilon_F|^{1/2}Q_{n,q} ^{--}|T_q-\epsilon_F|^{1/2}\right)dq.
\]
The combination of the last two inequalities shows that
\begin{equation}
  \label{kineBound}
  0 \leq \utr\left(\left(T-\epsilon_F\right)\overline{Q}\right) \leq \liminf_{n\to\infty}\utr\left(\left(T-\epsilon_F\right)Q_{n}\right) \leq C,
\end{equation}
so that $\overline{Q} \in \mathcal{K}$.

\paragraph{The state $\overline{Q}$ is a minimizer, and its density is uniquely defined.}
The densitiy $\rho_{\overline{Q}}$ is well defined in view of Proposition~\ref{prop:densityK} since $\overline{Q} \in \mathcal{K}$, but it is \textit{a priori} different from $\overline{\overline{ \rho}}_{Q} $ and $\widetilde{\rho}_Q$. The following lemma shows that all these densities actually coincide.

\begin{lemma}(Consistency of densities)
  \label{consistencyOfDensity}
  It holds $\rho_{\overline{Q}}-\nu =  \overline{\overline{\rho}}_{Q}-\nu  =\widetilde{\rho}_Q-\nu $ in $ \mathcal{D}'(\bb R)$. In particular, $\rho_{\overline{Q}} = \overline{\overline{\rho}}_Q$ as elements of~$L^2(\bb R)+L^p(\bb R)$, and $\rho_{\overline{Q}}-\nu = \widetilde{\rho}_Q-\nu$ as elements of~$\ud{\mathcal{C}}$.
\end{lemma}

We postpone the proof of this result to Section \ref{consistencyOfDensitySec}. We use it to obtain that, since $\ud{D}(\cdot,\cdot) $ defines an inner product on $\ud{\mathcal{C}}$,
\begin{equation}
  \label{Dbound}
  \ud{D}(\rho_{\overline{Q}}-\nu,\rho_{\overline{Q}}-\nu) \leq \liminf_{n\to\infty}\ud{D}(\rho_{Q_n}-\nu,\rho_{Q_n}-\nu).
\end{equation}
This shows in particular that $\overline{Q}$ is an admissible state (\textit{i.e.} $\overline{Q} \in \mathcal{F}_\nu$). Moreover, \eqref{kineBound} and~\eqref{Dbound} imply that the minimizing sequence $\left\{Q_n\right\}_{n\geq 1} $ is such that
\begin{equation}
  \label{minimialEnergyINF}
  \mathcal{E}_{\nu,0}(\overline{Q}) \leq \liminf_{n\to\infty}\mathcal{E}_{\nu,0}(Q_n),
\end{equation}
which shows that $\overline{Q}$ is a minimizer of~\eqref{minimizationpb}.

Let us finally prove that all minimizers share the same density. Consider two minimizers $\overline{Q}_1$ and $\overline{Q}_2$. By convexity of $\mathcal{F}_{\nu}$, it holds $\frac{1}{2}\left(\overline{Q}_1+\overline{Q}_2\right) \in \mathcal{F}_{\nu}$. Moreover,
\[
\cl{E}_{\nu,0}\left(\frac{\overline{Q}_1+\overline{Q}_2}{2}\right) =\frac{1}{2}\cl{E}_{\nu,0}\left(\overline{Q}_1\right) +\frac{1}{2}\cl{E}_{\nu,0}\left(\overline{Q}_2\right) - \frac{1}{4}\ud{D}\left(\rho_{\overline{Q}_1}-\rho_{\overline{Q}_2},\rho_{\overline{Q}_1}-\rho_{\overline{Q}_2}\right),
\]
which shows that $\ud{D}\left(\rho_{\overline{Q}_1}-\rho_{\overline{Q}_2},\rho_{\overline{Q}_1}-\rho_{\overline{Q}_2}\right) = 0$. This implies that all miminizers share the same density.

\subsubsection{Convergence of Yukawa to Coulomb}
\label{sec:proof_cv_energies}

\paragraph{Monotonicity of the ground state.}
Fix $\nu \in H^{-1}(\mathbb{R})$ and $m_1 \geq m_2>0$. Note first that, for any $f\in H^{-1}(\bb R)$,
\begin{equation}
  \label{nonincreasingD}
  \ud{D_{m_1}} (f,f) \leq \ud{D_{m_2}} (f,f) \leq \ud{D}(f,f),
\end{equation}
with the convention that $\ud{D}(f,f):= +\infty$ if $f\notin \ud{\cl{C}}$. Since $\rho_Q-\nu \in \ud{\mathcal{C}_m}$ for any $m>0$ when $Q\in \mathcal{K}$ (see Proposition~\ref{prop:densityK}), it holds $0 \leq \cl{E}_{\nu,m_1}(Q) \leq\cl{E}_{\nu,m_2}(Q)\leq \cl{E}_{\nu,0}(Q)$ for any $Q\in \mathcal{K}$. This immediately implies that
\begin{equation}
  \label{compareMinimizer}
  \forall m_1 \geq m_2\geq 0, \qquad 0\leq  I_{\nu,m_1}\leq I_{\nu,m_2}\leq I_{\nu,0},
\end{equation}
which proves that $m \mapsto I_{\nu,m}$ is non-increasing on~$(0,+\infty)$.

\paragraph{Continuity of the ground state.}
Fix $m>0$ and $\delta m > 0$. Denote by $\overline{Q}_m$ (resp. $\overline{Q}_{m+\delta m}$) one of the minimizers of the energy functional $\cl{E}_{\nu,m}$ (resp. $\cl{E}_{\nu,m+\delta m}$) on~$\mathcal{K}$. Then, in view of the monotonicity property~\eqref{compareMinimizer},
\[
\begin{aligned}
0\leq I_{\nu,m}-I_{\nu,m+\delta m} & = \cl{E}_{\nu,m}\left(\overline{Q}_m\right)-\cl{E}_{\nu,m+\delta m}\left(\overline{Q}_{m+\delta m}\right) \\
& \leq \cl{E}_{\nu,m}\left(\overline{Q}_{m+\delta m }\right) -\cl{E}_{\nu,m+\delta m}\left(\overline{Q}_{m+\delta m}\right)\\
& = \frac12 \ud{D_{m}}\left(\rho_{\overline{Q}_{m+\delta m }}-\nu,\rho_{\overline{Q}_{m+\delta m }}-\nu \right) - \ud{D_{m+\delta m}}\left(\rho_{\overline{Q}_{m+\delta m }}-\nu,\rho_{\overline{Q}_{m+\delta m }}-\nu \right) \\
& = \int_\mathbb{R} \frac{(m+\delta m)^2-m^2}{k^2+m^2} \frac{|\widehat{\rho_{\overline{Q}_{m+\delta m }}}(k)-\widehat{\nu}(k)|^2}{k^2+(m+\delta m)^2} \, dk \\
& \leq \frac12 \left[\left(1+\frac{\delta m}{m}\right)^2-1\right] \ud{D_{m+\delta m}}\left(\rho_{\overline{Q}_{m+\delta m }}-\nu,\rho_{\overline{Q}_{m+\delta m }}-\nu \right) \\
& \leq \left[\left(1+\frac{\delta m}{m}\right)^2-1\right] I_{\nu,m+\delta m} \leq \left[\left(1+\frac{\delta m}{m}\right)^2-1\right] I_{\nu,m}.
\end{aligned}
\]
A similar inequality can be obtained for $\delta m < 0$ sufficiently small. This allows to conclude that $m \mapsto I_{\nu,m}$ is continuous on~$(0,+\infty)$.

\paragraph{Limit as $m \to +\infty$.}
Fix $\nu \in H^{-1}(\mathbb{R})$. Note that $0 \in \mathcal{K}$, so that
\[
0 \leq I_{\nu,m} \leq \cl{E}_{\nu,m}(0) = \ud{D_{m}}(\nu,\nu) = 2\int_\mathbb{R} \frac{|\widehat{\nu}(k)|^2}{k^2+m^2} \, dk.
\]
The latter integral converges to~0 as $m \to +\infty$ by dominated convergence. This shows that $I_{\nu,m} \to 0$ as $m \to +\infty$.

\paragraph{Limit as $m \to 0$.}
Note first that \eqref{compareMinimizer} implies that $\lim_{m\to 0}I_{\nu,m}\leq I_{\nu,0}$ for any $\nu \in H^{-1}(\mathbb{R})$. Let us now prove the converse inequality under the conditions $\nu \in L^1(\R)$ and $|\cdot|\nu(\cdot) \in L^1(\R)$. Denote one of the minimizers of~\eqref{minimizationpbyoka} by $Q_{\nu,m}$. By~\eqref{compareMinimizer},
\[
\forall m >0, \qquad I_{\nu,m} = \utr\left(\left(T-\epsilon_F\right)Q_{\nu,m}\right) + \frac{1}{2}\ud{D_m}(\rho_{Q_{\nu,m}}-\nu,\rho_{Q_{\nu,m}}-\nu)\leq I_{\nu,0}.
\]
In particular, $\utr\left(\left(T-\epsilon_F\right)Q_{\nu,m}\right)$ is uniformly bounded. By arguments similar to the ones used to establish~\eqref{kineBound} and Lemma~\ref{consistencyOfDensity}, there exists $Q_{\nu,0} \in \mathcal{K}$ and a subsequence $(Q_{\nu,m_k})_{k \in \bb N}$ with $m_k \to 0$ such that~\eqref{eq:cv_Q_mk_1} and~\eqref{eq:cv_Q_mk_2} hold true, $\rho_{Q_{\nu,m_k}} \rightharpoonup \rho_{Q_{\nu,0}}$ weakly in $L^2(\mathbb{R}) + L^p(\mathbb{R})$ (for a fixed $1 < p < 5/3$) and
\begin{equation}
  \label{eq:bound_kinenergy_for_m_to_0}
  \utr\left((T-\epsilon_F)Q_{\nu,0}\right)\leq \lim_{k \to \infty}\utr\left((T-\epsilon_F)Q_{\nu,m_k}\right).
\end{equation}
We prove in the sequel that $Q_{\nu,0}$ is indeed a minimizer of~$\mathcal{E}_{\nu,0}$.

In order to do so, we need upper bounds on the Coulomb term~$\ud{D}(\rho_{Q_{\nu,0}}-\nu,\rho_{Q_{\nu,0}}-\nu)$. Since $\left\{\frac{\widehat{\rho}_{Q_{\nu,m_{k}}}-\widehat{\nu}}{\sqrt{|\cdot|^2+m_{k}^2}}\right\}_{k \geq 1}$ is bounded in $L^2(\bb R)$, it is possible to extract a subsequence, still denoted by $\left\{\frac{\widehat{\rho}_{Q_{\nu,m_{k}}}-\widehat{\nu}}{\sqrt{|\cdot|^2+m_{k}^2}}\right\}_{k \geq 1}$ with some abuse of notation, which weakly converges in $L^2(\bb R) $ to some function~$\widehat{w}$. For a given function $\Phi \in C^\infty_c(\mathbb{R})$, we introduce the sequence $f_{m_{k}}:=(-\Delta+m_{k}^2)^{1/2}\Phi$ and $f:= |\Delta|^{1/2}\Phi$. A simple computation shows that $f_{m_{k}} \in L^2(\bb R)$ for any $k\geq 1$ and $\left\|f_{m_{k}}-f\right\|_{L^2} \to 0$ as $k \to +\infty$.
Therefore,
\begin{equation}
  \label{convergenceL2}
  \langle \rho_{Q_{\nu,m_{k}}}-\nu, \Phi \rangle =
  \left(\frac{\widehat{\rho}_{Q_{\nu,m_{k}}}-\widehat{\nu}}{\sqrt{|\cdot|^2+m_{k}^2}}, \widehat{f_{m_{k}}}\right)_{L^2(\bb R)} \xrightarrow[k \to \infty]{} \left( \widehat{w}, \widehat{f}\right)_{L^2(\bb R)} = \left( \widehat{w}, |\cdot|\widehat{\Phi}\right)_{L^2(\bb R)}.
\end{equation}
Note that $|\cdot|\widehat{w} \in H^{-1}(\mathbb{R})$, hence its inverse Fourier transform $\mathcal{F}^{-1}\left(|\cdot|\widehat{w}\right)$ is well defined in $\mathcal{S}'(\bb R)$, and
\[
\left( \widehat{w}, |\cdot|\widehat{\Phi}\right)_{L^2(\bb R)} = \langle \mathcal{F}^{-1}\left(|\cdot|\widehat{w}\right), \Phi \rangle.
\]
The weak convergence $\rho_{Q_{\nu,m_k}} \rightharpoonup \rho_{Q_{\nu,0}}$ in $L^2(\mathbb{R}) + L^p(\mathbb{R})$ (for a fixed $1 < p < 5/3$) then allows to identify $\rho_{Q_{\nu,0}}-\nu$ and $\mathcal{F}^{-1}\left(|\cdot|\widehat{w}\right)$ in the sense of distributions. This shows that $\rho_{Q_{\nu,0}}-\nu\in \ud{\mathcal{C}}$, which together with~\eqref{eq:bound_kinenergy_for_m_to_0} implies that $Q_{\nu,0} \in\mathcal{F}_{\nu}$. Moreover, by the properties of weak limits in Hilbert spaces,
\[
\ud{D}\left(\rho_{Q_{\nu,0}}-\nu,\rho_{Q_{\nu,0}}-\nu\right) = 2\left\|\widehat{w}\right\|^2_{L^2} \leq 2 \liminf_{k \to +\infty} \left\| \frac{\widehat{\rho}_{Q_{\nu,m_{k}}}-\widehat{\nu}}{\sqrt{|\cdot|^2+m_{k}^2}}\right\|_{L^2}.
\]
This inequality combined with~\eqref{eq:bound_kinenergy_for_m_to_0} shows that $I_{\nu,0}\leq \mathcal{E}_{\nu,0}(Q_{\nu,0}) \leq \liminf_{m_{k} \to 0} I_{m_{k}}$. We can finally conclude that $Q_{\nu,0}$ is a minimizer of~\eqref{minimizationpb}, and that $\lim_{m\to 0}I_{\nu,m} = I_{\nu,0}$ when $\nu \in L^1(\R)$ and $|\cdot|\nu(\cdot) \in L^1(\R)$.

\subsection{Proof of Theorem~\ref{yukawaMinimizerthm}}
\label{thmYukawaMinimizer}

Theorem~\ref{thmexistence} shows that all minimizers of~\eqref{minimizationpbyoka} share the same density~$\rho_{\nu, m}$. To prove Theorem~\ref{yukawaMinimizerthm}, we proceed as follows. We first construct a potential $V_{\nu, m}$ from $\rho_{\nu,m}-\nu \in H^{-1}(\bb R)$. We show in Section~\ref{sec:min_lin_fctnal} that $V_{\nu, m}$ enjoys some regularity property, and that all the minimizers of (\ref{minimizationpbyoka}) also minimize a linear functional. We next construct a defect state by defining a spectral projector $\overline\gamma_m$ associated with the operator $T+V_{\nu, m}$, and show that $\overline{Q}_m:=\overline \gamma_m-\gamma_0$ is indeed in the kinetic energy space $\mathcal{K}$ (see Section~\ref{DefectStateSec}). We finally show in Section~\ref{sec:form_min} that $\overline{Q}_m$ is the unique minimizer of~\eqref{minimizationpbyoka} and that it satisfies a self-consistent equation. The proofs of some technical results are postponed until Sections~\ref{Secprop:defectstateQ_m} and~\ref{regularizedOpPropProof}.

\subsubsection{The minimizers of~\eqref{minimizationpbyoka} also mimimize a linear functional} 
\label{sec:min_lin_fctnal}
Given a minimizer $Q_{\nu, m}$ of~\eqref{minimizationpbyoka}, we consider the unique potential $V_{\nu, m}\in H^1(\bb R)$ satisfying the one-dimensional Yukawa equation
 \begin{equation}
-V_{\nu, m}''+m^2V_{\nu, m} = 2\left(\rho_{\nu,m}-\nu\right) \in H^{-1}(\bb R).
\label{potentialRHO}
\end{equation}
The potential $V_{\nu, m}$ defined in~\eqref{potentialRHO} has an explicit expression $V_{\nu, m} = \left(\rho_{\nu,m}-\nu\right)\star \frac{\mathrm{e}^{-m|\cdot|}}{m}$. Let us first give some properties of the potential~$V_{\nu,m}$, which will be useful for the subsequent analysis. Since $H^1(\bb R) \subset L^{\infty}(\bb R)$, we have $V_{\nu,m}\in L^{\infty}(\bb R)$ and that $\lim_{|z| \to +\infty}V_{\nu, m}(z) = 0$. We denote by $c_{V_{\nu, m}}:= \norm{V_{\nu, m}}_{L^{\infty}}$. Note that $V_{\nu,m}$ also belongs to $L^p(\bb R)$ for $p\in [2,+\infty]$ by interpolation. Finally, by the Kato--Seiler--Simon inequality (\ref{KatoSeilerSimon}), it is easy to see that
\[
\left\| V_{\nu, m}\left(1-\frac{d^2}{dz^2}\right)^{-1}\right\|_{\mathfrak{S}_2}\leq \frac{1}{2}\norm{V_{\nu, m}}_{L^2(\bb R)}. 
\]
In particular $V_{\nu, m}$ is $\left(-\frac{d^2}{dz^2}\right)$--compact. 

Let us now show that any minimizer $Q_{\nu, m}$ of~\eqref{minimizationpbyoka} minimizes a linear functional on $\mathcal{K}$. Since $Q_{\nu, m}$ minimizes~\eqref{minimizationpbyoka}, it holds, for an arbitrary state $Q\in \mathcal{K}$ and $0\leq t \leq 1$,
\begin{align*}
\mathcal{E}_m\left((1-t){Q}_{\nu, m}+tQ\right) -\mathcal{E}_m\left(Q_{\nu, m}\right)  \geq 0.
\end{align*}
A simple calculation then shows that $Q_{\nu, m}$ minimizes the following linear functional on $\mathcal{K}$:
\[
F(Q) := \utr\left(\left(T-\epsilon_F\right) Q\right) +\ud{D_m}(\rho_{\nu,m}-\nu,\rho_{Q}),
\]
which can also be written as follows in view of the definition and the regularity of $V_{\nu, m}$:
\begin{equation}
F(Q) =\utr\left(\left(T-\epsilon_F\right) Q\right) + \int_{\bb R}V_{\nu, m}(z)\rho_{Q}(z)\,dz.
\label{minimizedF}
\end{equation}

\subsubsection{Construction of a defect state $\overline{Q}_m$}
\label{DefectStateSec}
We construct a defect state $\overline{Q}_{\nu, m}$ as follows. First of all, by the Kato-Rellich theorem (see for example \cite[Theorem 9.10]{Helffer}), for each $q\in\bb R^2$, $T_q+V_{\nu, m} $ is a self-adjoint operator on $L^2(\bb R)$ with domain $ H^2(\bb R)$ and form domain $H^1(\bb R)$. Let us introduce a spectral projector as follows:
\[
\overline\gamma_m := \mathds 1_{(-\infty,\epsilon_F)}\left(T+V_{\nu, m}\right)=\mathcal{U}^{-1}\left(\int_{\bb R^2}^{\oplus}\overline\gamma_{m,q}\,dq \right)\mathcal{U} ,
\qquad 
\overline\gamma_{m,q}:= \mathds 1_{(-\infty,\epsilon_F)}\left(T_q+V_{\nu, m}\right).
\]
By construction, $\overline \gamma_{m,q} = 0$ and $\gamma_{0,q} = 0$ when $q\in \bb R^2\backslash \overline{\mathfrak{B}}_{\epsilon_F+c_{V_{\nu, m}}}$.

Let us study the spectral structure of $T_q+V_{\nu, m}$ as a function of $|q|$ (See Figure~\ref{FigSpectreT_qV_m}). It is clear that $T_q+V_{\nu, m}$ depends analytically on $|q|$, and has essential spectrum $\left[|q|^2/2, +\infty\right)$ since $V_{\nu,m}$ is $T_q$--compact. The potential $V_{\nu, m}$ introduces at most countably many eigenvalues below $|q|^2/2$, and $|q|^2/2$ is the only possible accumulation point of these eigenvalues.
\begin{figure}[h!]
\centering
\includegraphics[height=0.5\textwidth ]{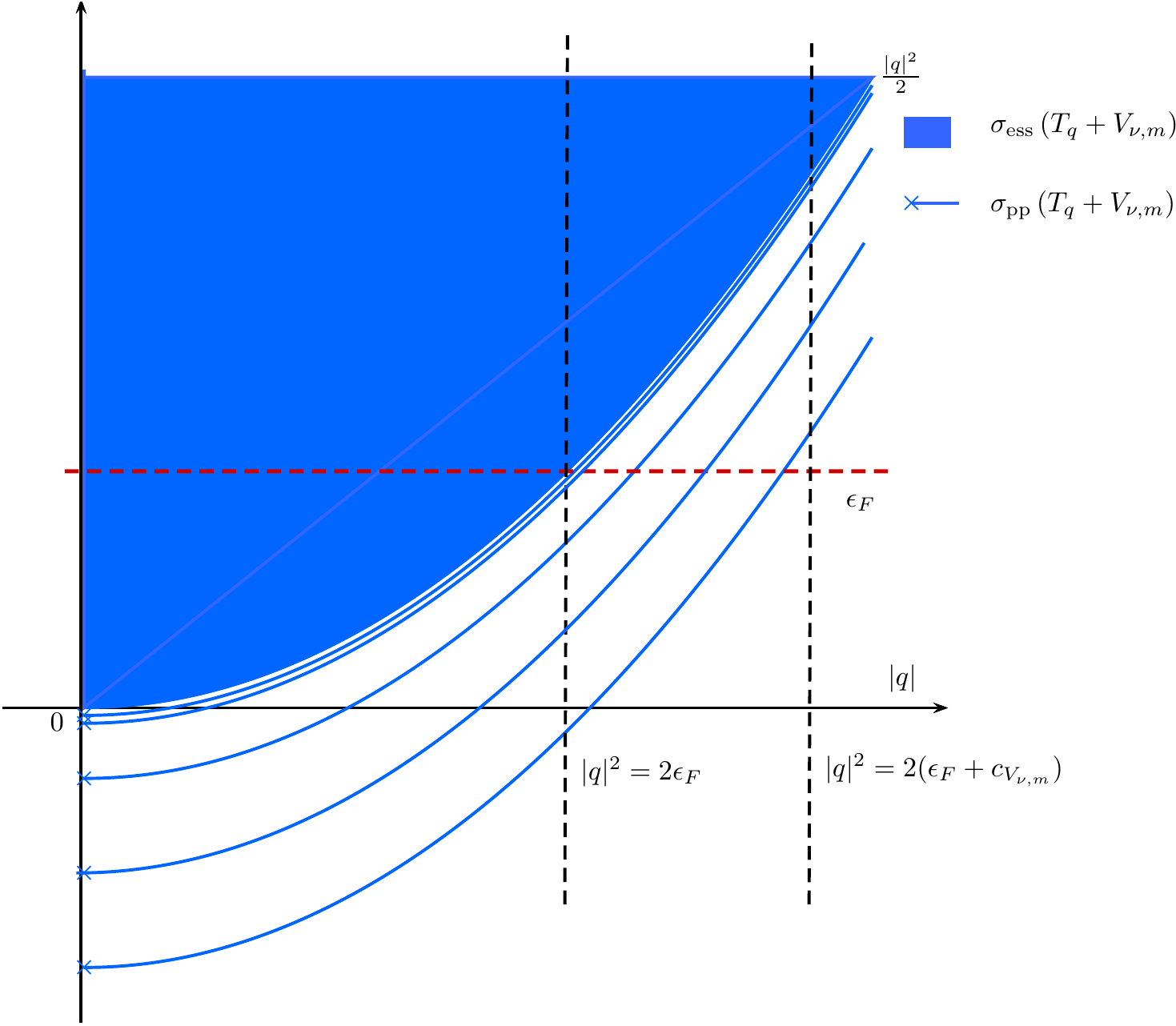}~
\caption{The spectrum of $T_q+V_{\nu, m}$ as a function of $|q|$. When $|q|^2 > 2\epsilon_F$, there are at most finitely many states in $\overline \gamma_{m,q}$. When $|q|^2 > 2(\epsilon_F+c_{V_{\nu, m}})$, there is no state in $\overline \gamma_{m,q}$.  }
\label{FigSpectreT_qV_m}
\end{figure}
Denoting by $\sigma_{\mathrm{pp}}(A)$ the set of eigenvalues of an operator~$A$, we partition the set $\left\{q\in \bb R^2\right\}$ by distinguishing whether the Fermi level~$\epsilon_F$ is an eigenvalue of $T_q+V_{\nu, m}$  or not:
\begin{equation}
 \mathfrak{M}_{\mathrm{pp}}:=\left\{ q\in \bb R^2\left|\, \epsilon_F\in \sigma_{\mathrm{pp}}(T_q+V_{\nu, m})\right.\right\} , \quad \mathfrak{M}_{\mathrm{pp}}^{\complement}:= \bb R^2 \backslash \mathfrak{M}_{\mathrm{pp}}:= \left\{ q\in \bb R^2\left|\, \epsilon_F\notin \sigma_{\mathrm{pp}}(T_q+V_{\nu, m})\right.\right\}.
\label{ppset}
\end{equation}
It is easy to see that $\mathfrak{M}_{\mathrm{pp}}$ has Lebesgue measure zero, since $T_0 + V_{\nu,m}$ has at most countably many eigenvalues, and $\epsilon_F\in \sigma_{\mathrm{pp}}(T_q+V_{\nu, m})$ if and only if $\epsilon_F-|q|^2/2 \in \sigma_{\mathrm{pp}}(T_0+V_{\nu, m})$. The reason for the separation of $\mathbb{R}^2$ into $\mathfrak{M}_{\mathrm{pp}}$ and $\mathfrak{M}_{\mathrm{pp}}^{\complement}$ is technical: we will need a regularized operator to approximate the spectral projector $\overline{\gamma}_{m,q}$ for $q\in \overline{\mathfrak{B}}_{\epsilon_F}$. For this reason we need $\epsilon_F$ not to be an eigenvalue of $T_q+V_{\nu, m}$ (recall~\cite[Theorem VIII.24]{ReeSim80}). 
Remark that when $q\in \overline{\mathfrak{B}}_{\epsilon_F}$, \emph{i.e.} $|q|^2 \leq 2\epsilon_F$, elements of $\mathfrak{M}_{\mathrm{pp}}^{\complement}$ are associated with non-negative eigenvalues of the operator $T_0 + V_{\nu,m} = -\frac{1}{2}\frac{d^2}{dz^2}+V_{\nu, m}$ embedded in its essential spectrum. It can be shown that even in dimension~1 there exist potentials~$V\in H^1(\bb R)$ such that $- \frac{1}{2}\frac{d^2}{dz^2}+V$ has positive eigenvalues embedded in the essential spectrum (see~\cite{vonNeuman}). In our case, it is highly non trivial to prove the absence of positive eigenvalues of the operator $-\frac{1}{2}\frac{d^2}{dz^2}+V_{\nu, m}$. Standard techniques to this end, such as Mourre estimates~\cite{Mourre1981} or Carleman estimates~\cite{Koch2006}, which involve estimating the decay property of $z V_{\nu, m}'(z)$ when $|z|\to +\infty$, are not immediately applicable as the density $\rho_{\nu, m}-\nu $ is only in $H^{-1}(\mathbb{R})$ \textit{a priori}. In order to avoid this difficulty we construct a test defect state by eliminating the points $q\in \mathfrak{M}_{\mathrm{pp}}$:
\begin{equation}
\overline{Q}_m:= \mathcal{U}^{-1}\left(\int_{\bb R^2}^{\oplus}\overline{Q}_{m,q}\,dq\right)\cl{U},
\label{defectstateQ_m}
\end{equation}
where, recalling that $\overline \gamma_{m,q}= 0$ when $q\in \bb R^2\backslash \overline{\mathfrak{B}}_{\epsilon_F+c_{V_{\nu,m}}} $,
\begin{equation}
\overline{Q}_{m,q}:=\left\{ \begin{aligned}
\overline\gamma_{m,q}- \gamma_{0,q}, & \quad \text{for } q\in \overline{\mathfrak{B}}_{\epsilon_F+c_{V_{\nu,m}}}\cap \mathfrak{M}_{\mathrm{pp}}^{\complement}, \\
0, & \quad \text{for } q\in \bb R^2\backslash\left(\overline{\mathfrak{B}}_{\epsilon_F+c_{V_{\nu,m}}}\cap \mathfrak{M}_{\mathrm{pp}}^{\complement}\right) .
\end{aligned}\right.
\label{zeroBarQ}
\end{equation}
Denoting by $P^\perp = 1-P$ for an orthogonal projector~$P$, an easy algebraic calculation shows that
\begin{equation}
- \left(\overline{Q}_{m,q}\right)^2  = \overline{\gamma}_{m,q}^{\perp}\overline{Q}_{m,q}\overline{\gamma}_{m,q}^{\perp} - \overline\gamma_{m,q}\overline{Q}_{m,q}\overline\gamma_{m,q},
\label{projetorsAlgebra1}
\end{equation}
and
\begin{equation}
\left(\overline{Q}_{m,q}\right)^2 =\gamma_{0,q}^{\perp}\overline{Q}_{m,q} \gamma_{0,q}^{\perp}- \gamma_{0,q}\overline{Q}_{m,q}\gamma_{0,q} = \overline{Q}_{m,q}^{++}-\overline{Q}_{m,q}^{--}.
\label{projetorsAlgebra3}
\end{equation}

\begin{prop}
\label{prop:defectstateQ_m}
The defect state $\overline{Q}_m$ defined in (\ref{defectstateQ_m}) belongs to the kinetic energy space $\mathcal{K}$. Moreover, the energy functional $F(\overline{Q}_m )$ can be written as
\[
\begin{aligned}
& F(\overline{Q}_m ) =\utr\left((T-\epsilon_F)\overline{Q}_m \right) + \int_{\bb R}V_{\nu,m} \rho_{\overline{Q}_m}\\
&= \frac{1}{(2\pi)^2}\int_{\bb R^2}\m{Tr}_{L^2(\bb R)}\left(\left|T_q-\epsilon_F+V_{\nu,m} \right|^{1/2}\left(\overline{\gamma}_{m,q}^{\perp}\overline{Q}_{m,q}\overline{\gamma}_{m,q}^{\perp}-\overline{\gamma}_{m,q}\overline{Q}_{m,q}\overline{\gamma}_{m,q}\right)\left|T_q-\epsilon_F+V_{\nu,m} \right|^{1/2}\right) dq. 
\end{aligned}
\]
\end{prop}

The proof of this proposition, which can be read in Section~\ref{Secprop:defectstateQ_m}, is quite technical and constitutes the core of the proof of Theorem~\ref{yukawaMinimizerthm}.

\subsubsection{Form of the minimizers}
\label{sec:form_min}

We prove that all minimizers satisfy~\eqref{self-consist-eq} by showing that $\overline{Q}_m\in\mathcal{K}$ minimizes $F(Q)$ defined in~\eqref{minimizedF}. This implies that the only difference between a minimizer~$Q_{\nu,m}$ of~\eqref{minimizationpbyoka} and $\overline{Q}_m$ can be due to bound states at the Fermi level~$\epsilon_F$. We show in fact that $Q_{\nu,m} = \overline{Q}_m$. 

We first need the following lemma on the density of finite-rank operators in $\mathcal{K}$.

\begin{lemma}[Density of direct sums of finite-rank operators in $\mathcal{K}$]
  \label{lemmaFiniterankOperator}
  Every $Q := \mathcal{U}^{-1}\left(\int_{\bb R^2}^{\oplus}Q_q\,dq\right)\mathcal{U} \in \mathcal{K}$ can be approximated by a sequence of operators $Q_n:= \mathcal{U}^{-1}\left(\int_{\bb R^2}^{\oplus}Q_{n,q}\,dq\right)\mathcal{U}$ with the following properties:
  \begin{itemize}
  \item for all $n \geq 1$ and for almost all $q\in\bb R^2$, $Q_{n,q}\in \mathcal{K}_q$ is finite-rank, and
    \[
    q\mapsto Q_{n,q}\in L^{\infty}(\bb R^2; \mathcal{S}(L^2(\bb R))),\qquad q\mapsto |T_q-\epsilon_F|^{1/2}\left(Q_{n,q}^{++}-Q_{n,q}^{--}\right) |T_q-\epsilon_F|^{1/2} \in L^1(\bb R^2; \mathfrak{S}_1);
    \]
  \item the renormalized kinetic energy converges as
    \begin{equation}
      \begin{aligned}
        \lim_{n\to\infty}\int_{\bb R^2}\m{Tr}_{L^2(\bb R)}&\left( |T_q-\epsilon_F|^{1/2}\left(Q_{n,q}^{++}-Q_{n,q}^{--} \right)|T_q-\epsilon_F|^{1/2}\right)dq \\
        &= \int_{\bb R^2}\m{Tr}_{L^2(\bb R)}\left( |T_q-\epsilon_F|^{1/2}\left(Q_{q}^{++}-Q_{q}^{--} \right)|T_q-\epsilon_F|^{1/2}\right)dq.
      \end{aligned}
      \label{limTrace}
    \end{equation}
  \end{itemize}
\end{lemma}

\begin{proof}
The proof is similar to the proof of Lemma~\ref{continuityLemma} and is based on~\cite[Lemma 3.2]{frank2013}. First of all, we show that there exists a sequence of operators $\widetilde{Q}_{n}:= \mathcal{U}^{-1}\left(\int_{\bb R^2}^{\oplus}\widetilde Q_{n,q}\,dq\right)\mathcal{U} \in \mathcal{K}$ satisfying~\eqref{limTrace} such that $q\mapsto \widetilde{Q}_{n,q}\in L^{\infty}(\bb R^2;\mathfrak{S}_2)$. Introduce the negligible set $\Omega = \{q \in \mathbb{R}^2 | Q_q \not \in \mathcal{K}_q\}$. For $q\in\bb R^2 \backslash \Omega$ and $n\geq 1$, define $P_{n,q}:= \mathds 1_{\left[\frac{1}{n},n\right]}\left(\left|T_q-\epsilon_F\right|\right)$ and $\widetilde{Q}_{n,q} := P_{n,q}Q_qP_{n,q}$. It can be checked that $\widetilde{Q}_{n,q}$ and $\widetilde{Q}_{n,q}\left|T_q-\epsilon_F\right|^{1/2}$ are Hilbert--Schmidt by writing $\widetilde{Q}_{n,q} = P_{n,q}Q_q \left|T_q-\epsilon_F\right|^{1/2}\frac{P_{n,q}}{\left|T_q-\epsilon_F\right|^{1/2}}$, and similarly that $\widetilde{Q}_{n,q}^{\pm\pm}$ and $\left|T_q-\epsilon_F\right|^{1/2}\widetilde{Q}_{n,q}^{\pm\pm}\left|T_q-\epsilon_F\right|^{1/2}$ are trace-class. Moreoever, by the uniform boundedness of $0\leq P_{n,q}\leq 1$ the fact that $P_n$ commutes with $|T_q-\epsilon_F|^{1/2}$, it is easy to see that $q\mapsto \widetilde{Q}_{n,q}\in L^{\infty}(\bb R^2;\mathfrak{S}_2)$ and $q\mapsto |T_q-\epsilon_F|^{1/2}\left(\widetilde{Q}_{n,q}^{++}-\widetilde{Q}_{n,q}^{--}\right) |T_q-\epsilon_F|^{1/2} \in L^1(\bb R^2; \mathfrak{S}_1)$, with
\[
0\leq \pm\m{Tr}_{L^2(\bb R)}\left(|T_q-\epsilon_F|^{1/2}\widetilde{Q}_{n,q}^{\pm\pm} |T_q-\epsilon_F|^{1/2}\right)\leq \pm\m{Tr}_{L^2(\bb R)}\left(|T_q-\epsilon_F|^{1/2}Q_{q}^{\pm\pm}|T_q-\epsilon_F|^{1/2}\right).
\]

Let us now turn to the convergence of the relative kinetic energies. First, for all $q\in \bb R^2\backslash \Omega$, using that $P_n$ converges strongly to~1 in~$L^2(\mathbb{R})$ (which implies that $P_n A P_n$ converges to~$A$ in $\mathfrak{S}_p$ for $A \in \mathfrak{S}_p$),
\[
\lim_{n\to +\infty}\norm[\bigg]{ |T_q-\epsilon_F|^{1/2}(\widetilde{Q}_{n,q}^{\pm\pm}-Q_{q}^{\pm\pm}) |T_q-\epsilon_F|^{1/2}}_{\mathfrak{S}_1}=0.
\]
Therefore, by the dominated convergence theorem,
\begin{align*}
\lim_{n\to\infty}\int_{\bb R^2}\pm\m{Tr}_{L^2(\bb R)}\left(|T_q-\epsilon_F|^{1/2}\widetilde{Q}_{n,q}^{\pm\pm} |T_q-\epsilon_F|^{1/2}\right)dq  =\int_{\bb R^2}\pm\m{Tr}_{L^2(\bb R)}\left(|T_q-\epsilon_F|^{1/2}Q_{q}^{\pm\pm}|T_q-\epsilon_F|^{1/2}\right)dq.
\end{align*}

In order to conclude the proof, it remains to approximate the Hilbert--Schmidt operators $\widetilde{Q}_{n,q}$ by finite-rank operators $Q_{n,q} \in \mathcal{K}_q$. This can be done as in~\cite[Theorems~5 and~6]{Hainzl2009}, more details can be found in~\cite{Cao}.
\end{proof}

Consider $m>0$ and $ \nu\in H^{-1}(\bb R)$, and let us prove that $ \mathds 1_{(-\infty,\epsilon_F]}(T +V_{\nu,m}) - \gamma_0$ is the unique minimizer of~\eqref{minimizationpbyoka}. For $Q \in \mathcal{K}$ a direct sum of smooth finite-rank operators (\emph{i.e.} as the ones in Lemma~\ref{lemmaFiniterankOperator}), \eqref{minimizedF} becomes
\begin{equation}
\begin{aligned}
& F(Q) = \utr\left(\left(T-\epsilon_F\right) Q\right) + \int_{\bb R}V_{\nu,m}(x)\rho_{Q}(x)\,dx \\
&=\frac{1}{(2\pi)^2}\int_{\bb R^2} \m{Tr}_{L^2(\bb R)}\left(\left(T_q-\epsilon_F+V_{\nu,m} \right)Q_q\right) dq\\
&= \frac{1}{(2\pi)^2}\int_{\bb R^2} \m{Tr}_{L^2(\bb R)}\left(\left|T_q-\epsilon_F+V_{\nu,m} \right|^{1/2}\left(\overline{\gamma}_{m,q}^{\perp}Q_q\overline{\gamma}_{m,q}^{\perp}-\overline{\gamma}_{m,q}Q_q\overline{\gamma}_{m,q}\right)\left|T_q-\epsilon_F+V_{\nu,m} \right|^{1/2}\right) dq.
\end{aligned}
\label{FQformula}
\end{equation}
Therefore, by~\eqref{FQformula} and Proposition~\ref{prop:defectstateQ_m},
\begin{equation}
\begin{aligned}
F(Q) - F(\overline{Q}_m) = \frac{1}{(2\pi)^2}\int_{\bb R^2}\m{Tr}_{L^2(\bb R)}& \left(\left|T_q-\epsilon_F+V_{\nu,m} \right|^{1/2}\left(\overline{\gamma}_{m,q}^{\perp}\left(Q_q-\overline{Q}_{m,q}\right)\overline{\gamma}_{m,q}^{\perp} \right.\right.\\
&\left.\left.\qquad
-\overline{\gamma}_{m,q}\left(Q_q-\overline{Q}_{m,q}\right)\overline{\gamma}_{m,q}\right)\left|T_q-\epsilon_F+V_{\nu,m} \right|^{1/2}\right) dq.
\end{aligned}
\label{Fqfinal}
\end{equation}
Remark that, by~\eqref{zeroBarQ}, $-\overline{\gamma}_{m,q} \leq Q_q-\overline{Q}_{m,q}\leq 1-\overline{\gamma}_{m,q}$ for any $q \in \mathbb{R}^2 \backslash \Omega$ with $\Omega = \{q \in \mathbb{R}^2 | Q_q \not \in \mathcal{K}_q\}$ negligible, so that
\[
\m{Tr}_{L^2(\bb R)}\left(\left|T_q-\epsilon_F+V_{\nu,m} \right|^{1/2}\left(\overline{\gamma}_{m,q}^{\perp}\left(Q_q-\overline{Q}_{m,q}\right)\overline{\gamma}_{m,q}^{\perp}-\overline{\gamma}_{m,q}\left(Q_q-\overline{Q}_{m,q}\right)\overline{\gamma}_{m,q}\right)\left|T_q-\epsilon_F+V_{\nu,m} \right|^{1/2}\right)  \geq 0.
\]
This implies that 
\[
F(Q)\geq F(\overline{Q}_m).
\]
This inequality can be extended to any $Q\in \mathcal{K}$ by~\eqref{limTrace} in Lemma~\ref{lemmaFiniterankOperator}. This shows that $Q$ is a minimizer of~\eqref{minimizationpbyoka} if and only if
\[
Q = \overline{Q}_m + \delta,
\]
where 
\[
\delta := \mathcal{U}^{-1}\left(\int_{\bb R^2}^{\oplus } \delta_q\,dq \right)\mathcal{U},
\]
with $0 \leq \delta_q \leq \mathds 1_{\left\{T_q+V_{\nu,m}= \epsilon_F\right\}} $. Since the set $\mathfrak{M}_{\mathrm{pp}}=\left\{ q\in \bb R^2\left|\, \epsilon_F\in \sigma_{\mathrm{pp}}(T_q+V_{\nu, m})\right.\right\}$ defined in (\ref{ppset}) has zero Lebesgue measure as discussed in Section~\ref{DefectStateSec}, it follows that $\delta = \mathcal{U}^{-1}\left(\int_{\bb R^2}^{\oplus } \delta_q\,dq \right)\mathcal{U} = 0$, and so the unique minimizer of~\eqref{minimizationpbyoka} is $\mathds 1_{(-\infty,\epsilon_F)}(T +V_{\nu,m}) - \gamma_0 = \mathds 1_{(-\infty,\epsilon_F]}(T +V_{\nu,m}) - \gamma_0$.

\subsubsection{Proof of Proposition \ref{prop:defectstateQ_m}}
\label{Secprop:defectstateQ_m}

The statement of the proposition is obtained by a limiting procedure, where we approximate the defect state $\overline{Q}_m$ by a family of regular operators with a spectral gap around~$\epsilon_F$, relying on the idea used in~\cite[Section~4.2]{frank2013}. Since $\overline{Q}_m$ is defined by a direct integral over $\overline{\mathfrak{B}}_{\epsilon_F+c_{V_{\nu,m}}}$ and $\mathfrak{M}_{\mathrm{pp}}$ is negligible, it suffices to construct regularizations of $\overline{Q}_{m,q}$ for $q \in \overline{\mathfrak{B}}_{\epsilon_F+c_{V_{\nu,m}}}\cap \mathfrak{M}_{\mathrm{pp}}^{\complement}$. 

\paragraph{Construction of a regularized kinetic operator.} 
In order to achieve uniform estimates locally in~$q$, we define a gap-opening function $h_{q}^{\eta}$ which locally does not depend on~$q$ for $q \in \overline{\mathfrak{B}}_{\epsilon_F+c_{V_{\nu,m}}}\cap \mathfrak{M}_{\mathrm{pp}}^{\complement}$. Let us first define a partition of $|q|$ for $q\in \overline{\mathfrak{B}}_{\epsilon_F+c_{V_{\nu,m}}}\cap \mathfrak{M}_{\mathrm{pp}}^{\complement}$. Given $\eta >0$ small enough, we first note that there exist integers $N_1\leq N_2$ such that $ N_1 \eta  \leq 2\epsilon_F < (N_1+1)\eta$ and $ N_2 \eta  \leq 2(\epsilon_F + c_{V_{\nu,m}}) < (N_2+1)\eta$. We then consider the following covering:
\begin{equation}
\begin{aligned}
\left[0, \sqrt{2(\epsilon_F+c_{V_{\nu,m}})}\right]\subset \left[0,\sqrt{\eta}\right)\cup  \left[\sqrt{\eta},\sqrt{2\eta}\right)\cup \cdots \left[\sqrt{N_2\eta},\sqrt{(N_2+1)\eta}\right),
\label{partition_2ef}
\end{aligned}
\end{equation}
and an associated gap-opening function:
\[
h_{q}^{\eta}(k):= \left\{
\begin{aligned}
- \eta \mathds 1\left(\frac{k^2}{2}\leq \epsilon_F - \frac{n\eta }{2}\right)+ \eta \mathds 1\left( \frac{k^2}{2} >\epsilon_F - \frac{n\eta }{2} \right), &\text{ if  } \, |q| \in \left[\sqrt{n\eta} , \sqrt{(n+1)\eta}\right),\, n = 0, 1, 2, \cdots, N_1,\\
0,\quad  &\text{ if   }\, |q|\in \left[\sqrt{(N_1+1)\eta}, \sqrt{(N_2+1)\eta}\right).
 \end{aligned}\right.
\]
The corresponding approximate kinetic energy operators are defined as 
\[
K_{q}^{\eta}:=T_q+  h_{q}^{\eta}\left(-\rmi\frac{d}{dz}\right)= -\frac{1}{2}\frac{d^2}{dz^2}+ \frac{|q|^2}{2}+ h_{q}^{\eta}\left(-\rmi\frac{d}{dz}\right).
\]
Remark that $K_q^{\eta}$ has purely absolutely continuous spectrum, and that $\gamma_{0,q}:= \mathds 1_{(-\infty,\epsilon_F]}\left(T_{q}\right)  = \mathds 1_{(-\infty,\epsilon_F]}\left(K_q^{\eta}\right)$ for all $q\in \bb R^2$. Moreover, for $ n = 0, 1,\cdots, N_1$ and $ |q| \in \left[\sqrt{n\eta} , \sqrt{(n+1)\eta}\right)$, the operator $K_q^{\eta}$ has a spectral gap $\left(\epsilon_F-\eta + \frac{|q|^2-n\eta}{2} , \epsilon_F+\eta + \frac{|q|^2-n\eta}{2} \right)$; see Figure~\ref{FigSpectreT_qV_m_smoothed} for a sketch of the spectrum. 
\begin{figure}[h!]
\centering
\includegraphics[height=0.5\textwidth ]{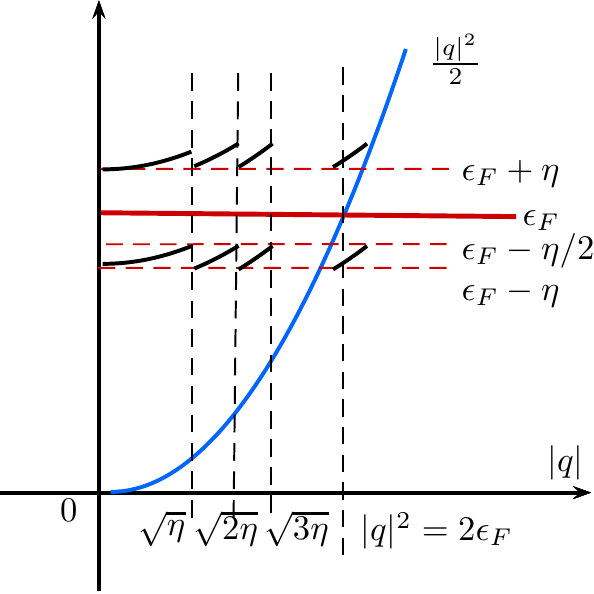}~
\caption{Spectrum of $K_q^{\eta}$ around $\epsilon_F$. The dark thick lines are the upper/lower bounds of the essential spectrum of $K_q^{\eta}$ around $\epsilon_F$.}
\label{FigSpectreT_qV_m_smoothed}
\end{figure}

The regularized kinetic energy operator is now defined to be $K_q^{\eta} + V_{\nu,m}$. Since the potential $V_{\nu,m}$ is a compact perturbation of~$K_q^{\eta}$ (as $(1+K_q^{\eta})^{-1}V_{\nu,m} \in \mathfrak{S}_2$ by the Kato--Seiler--Simon inequality~\eqref{KatoSeilerSimon}), the essential spectrum of $K_q^{\eta} + V_{\nu,m}$ is the same as for $K_q^{\eta}$, and at most countably many eigenvalues can be introduced in the spectral gap and below~$|q|^2/2$, which can only accumulate at the edges $\epsilon_F-\eta + \frac{|q|^2-n\eta}{2}$ and $\epsilon_F+\eta + \frac{|q|^2-n\eta}{2}$ of the spectral gap or at~$|q|^2/2$. In order to exclude these eigenvalues, we discard the values of~$q$ in the set 
\[
\mathfrak{M}_{\mathrm{pp}}^{\eta}:= \left\{ \left. q\in \overline{\mathfrak{B}}_{\epsilon_F+c_{V_{\nu,m}}}\cap \mathfrak{M}_{\mathrm{pp}}^{\complement} \right|\, \epsilon_F \in \sigma_{\mathrm{disc}}(K_q^{\eta}+V_{\nu,m})\right\}.
\]
We show next that this set is negligible. 

\paragraph{The set $\mathfrak{M}_{\mathrm{pp}}^{\eta}$ has Lebesgue measure zero.}
Let us show that there are only finitely many values of~$|q|$ for $q \in \mathfrak{M}_{\mathrm{pp}}^{\eta}$ for a given $\eta >0$. We distinguish two cases: 
\begin{enumerate}[(i)]
\item in the regions $|q| \in [\sqrt{n\eta}, \sqrt{(n+1)\eta} )$ for $0 \leq n \leq N_1$, the construction of $K_q^{\eta}+V_{\nu,m}$ makes sure that there can be at most finitely many values of~$|q|$ such that $\epsilon_F \in \sigma_{\mathrm{disc}}(K_q^{\eta}+V_{\nu,m})$;
\item in the region $|q|\in \left[\sqrt{(N_1+1)\eta}, \sqrt{(N_2+1)\eta}\right)$, it holds $K_q^{\eta}= T_q$. The condition $\epsilon_F \in \sigma_{\mathrm{disc}}(K_q^{\eta}+V_{\nu,m})$ is therefore equivalent to $\epsilon_F-|q|^2/2 \in \sigma_{\mathrm{disc}}(T_0+V_{\nu,m})$. Since $(N_1+1)\eta > 2\epsilon_F$, there exists $\delta > 0$ such that $\epsilon_F-|q|^2/2 \leq -\delta$ for $|q|\in \left[\sqrt{(N_1+1)\eta}, \sqrt{(N_2+1)\eta}\right)$. Since $\sigma_{\mathrm{disc}}(T_0+V_{\nu,m}) \cap (-\infty,-\delta]$ is at most finite, this shows that $\mathfrak{M}_{\mathrm{pp}}^{\eta} \cap \left[\sqrt{(N_1+1)\eta}, \sqrt{(N_2+1)\eta}\right)$ is at most finite.
\end{enumerate}
Since $\mathfrak{M}_{\mathrm{pp}}^{\eta}$ is the union of a finite number of circles in~$\mathbb{R}^2$, it has Lebesgue measure zero.

\paragraph{Construction of a regularized defect state. }
For $\eta>0 $ small enough, we define a regularized operator in analogy with~\eqref{defectstateQ_m} and~\eqref{zeroBarQ} by excluding the set $ \mathfrak{M}_{\mathrm{pp}}^{\eta}$:
\begin{equation}
\overline{Q}_m^{\eta}:= \mathcal{U}^{-1}\left(\int_{\bb R^2}^{\oplus}\overline{Q}_{m,q}^{\eta}\,dq\right)\cl{U},
\label{smoothedSP}
\end{equation}
where, with $\Pi_{V,q}^{\eta}:=\mathds 1_{(-\infty,\epsilon_F)}\left(K_q^{\eta}+V_{\nu,m}\right)$,
\begin{equation}
\overline{Q}_{m,q}^{\eta}:= \left\{
\begin{aligned}
\Pi_{V,q}^{\eta} -\mathds 1_{(-\infty,\epsilon_F]}\left(K_q^{\eta}\right)\quad & \text{when $q\in \overline{\mathfrak{B}}_{\epsilon_F+c_{V_{\nu,m}}}\cap \mathfrak{M}_{\mathrm{pp}}^{\complement}\cap \left(\mathfrak{M}_{\mathrm{pp}}^{\eta}\right)^\complement$}, \\
 0 \quad & \text{when $q\in \left( \bb R^2\backslash \overline{\mathfrak{B}}_{\epsilon_F+c_{V_{\nu,m}}} \right) \cup \mathfrak{M}_{\mathrm{pp}} \cup \mathfrak{M}_{\mathrm{pp}}^{\eta}$}.
 \end{aligned}
 \right.
\end{equation}
Moreover, since $\mathds 1_{(-\infty,\epsilon_F]}\left(K_q^{\eta}\right) = \gamma_{0,q}$ and similarly to~\eqref{projetorsAlgebra3},
\begin{equation}
\left(\overline{Q}_{m,q}^{\eta}\right)^2  = \overline{Q}_{m,q}^{\eta,++}-\overline{Q}_{m,q}^{\eta,--}.
\label{projetorsAlgebra4}
\end{equation}

A first observation is that, for any $q\in \overline{\mathfrak{B}}_{\epsilon_F+c_{V_{\nu,m}}}\cap \mathfrak{M}_{\mathrm{pp}}^{\complement}$,
\begin{equation}
  \overline{Q}_{m,q}^{\eta} \xrightarrow[\eta \to 0]{}  \overline{Q}_{m,q} \text{ strongly}.
\label{strongConvergenceQ}
\end{equation}
This follows from the convergence of $K_q^{\eta}+V_{\nu,m}$ to $T_q+V_{\nu,m}$ in the norm resolvent sense when $\eta \to 0$ (which comes from the estimate $\norm{h_q}_{L^{\infty}}\leq \eta$), and the result~\cite[Theorem VIII.24]{ReeSim80} on the convergence of spectral projectors. 

The following lemma summarizes some properties of the regularized operator $\overline{Q}_m^{\eta}$.

\begin{lemma}
\label{regularizedOpProp}
For any $\eta >0$, the regularized operator $\overline{Q}_m^{\eta}$ defined in~\eqref{smoothedSP} belongs to the kinetic energy space~$\mathcal{K}$. Moreover, the following properties hold:
\begin{enumerate}[(1)]
\item for all $\eta>0$ and $q\in \overline{\mathfrak{B}}_{\epsilon_F+c_{V_{\nu,m}}}$,
\[
\overline{Q}_{m,q}^{\eta}\in \mathfrak{S}_2,\qquad \left(-\frac{d^2}{dz^2}\right)\overline{Q}_{m,q}^{\eta}\in \mathfrak{S}_2,\qquad (K_q^{\eta}+V_{\nu,m})\overline{Q}_{m,q}^{\eta}\in \mathfrak{S}_2.
\]
In particular $\overline{Q}_{m,q}^{\eta}\in \mathcal{K}_q$.
\item for any $\eta>0$ and $q\in \overline{\mathfrak{B}}_{\epsilon_F+c_{V_{\nu,m}}}$, 
\begin{align}
&\m{Tr}_{L^2(\bb R)}\left(|T_q-\epsilon_F|^{1/2}\left(\overline{Q}_{m,q}^{\eta,++}-\overline{Q}_{m,q}^{\eta,--}\right)|T_q-\epsilon_F|^{1/2}\right) + \m{Tr}_{L^2(\bb R)}\left(\left|h_q^{\eta}\left(-\rmi \frac{d}{dz}\right)\right|\left(\overline{Q}_{m,q}^{\eta,++}-\overline{Q}_{m,q}^{\eta,--}\right)\right) \notag \\
&\qquad  + \int_{\bb R}V_{\nu,m}\rho_{\overline{Q}_{m,q}^{\eta}} \notag \\
&= -\m{Tr}_{L^2(\bb R)}\left(\left|K_q^{\eta}+V_{\nu,m}-\epsilon_F\right|^{1/2}\left(\overline{Q}_{m,q}^{\eta}\right)^2\left|K_q^{\eta}+V_{\nu,m}-\epsilon_F\right|^{1/2}\right) \leq 0.
\label{boundsupRegular}
\end{align}
\end{enumerate}
\end{lemma}

We postpone the proof of this lemma to Section~\ref{regularizedOpPropProof}. Let us however already remark that the terms in the energy functional~\eqref{boundsupRegular} are indeed well defined. In particular,
\[
\left|h_q^{\eta}\left(-\rmi \frac{d}{dz}\right)\right|\left(\overline{Q}_{m,q}^{\eta,++}-\overline{Q}_{m,q}^{\eta,--}\right)\in\mathfrak{S}_1,
\]
since $\left|h_q^{\eta}\left(-\rmi \frac{d}{dz}\right)\right|$ is bounded and $(1-\gamma_{0,q})\left(\overline{Q}_{m,q}^{\eta}\right)^2=\overline{Q}_{m,q}^{\eta,++} \in \mathfrak{S}_1$ and $\gamma_{0,q}\left(\overline{Q}_{m,q}^{\eta}\right)^2=-\overline{Q}_{m,q}^{\eta,--} \in \mathfrak{S}_1$ (relying on $\overline{Q}_{m,q}^{\eta}\in \mathfrak{S}_2$). 

\paragraph{The regularized defect state $\overline{Q}_m^{\eta}$ converges to $\overline{Q}_m $ in $ \mathcal{K}$.}
The remainder of this section is devoted to showing that there exists a sequence $(\eta_i)_{i \geq 1}$ converging to~0 as $i \to +\infty$ such that $\overline{Q}_m^{\eta_i}\in \mathcal{K}$ converges to the defect state $\overline{Q}_m$ in the sense of~\eqref{eq:cv_Q_mk_2}; from which we can conlude that $\overline{Q}_m$ is in $\mathcal{K}$. The general idea is to use the fact that the total energy of $\overline{Q}_m^{\eta_i}$ is uniformly bounded. By identifying the limit when $\eta_i \to 0$ with the defect state $\overline{Q}_m$ we obtain the desired result. This can be summarized in the following lemma, which crucially relies on the properties stated in Lemma~\ref{regularizedOpProp}.

\begin{lemma} There exists a sequence $(\eta_i)_{i \geq 1}$ such that the following convergences hold as $i \to +\infty$:
\begin{equation}
\overline{Q}_{m,(\boldsymbol{\cdot})}^{\eta_i}\left|K_{(\boldsymbol{\cdot})}^{\eta_i}+V_{\nu,m}-\epsilon_F\right|^{1/2}\xrightharpoonup[i\to +\infty]{} \overline{Q}_{m,(\boldsymbol{\cdot})}\left|T_{(\boldsymbol{\cdot})}+V_{\nu,m}-\epsilon_F\right|^{1/2} \text{ weakly in } L^2\left(\overline{\mathfrak{B}}_{\epsilon_F+c_{V_{\nu,m}}};\mathfrak{S}_2\right),
\label{re1}
\end{equation}
\begin{equation}
\overline{Q}_{m,(\boldsymbol{\cdot})}^{\eta_i}\left|T_{(\boldsymbol{\cdot})}-\epsilon_F\right|^{1/2} \xrightharpoonup[i \to +\infty]{} \overline{Q}_{m,(\boldsymbol{\cdot})}\left|T_{(\boldsymbol{\cdot})}-\epsilon_F\right|^{1/2}  \text{ weakly in } L^2\left(\overline{\mathfrak{B}}_{\epsilon_F+c_{V_{\nu,m}}};\mathfrak{S}_2\right).
\label{re2}
\end{equation}
Moreover, the defect state $\overline{Q}_m$ belongs to $\mathcal{K}$ and 
\begin{equation}
\utr\left((T-\epsilon_F)\overline{Q}_m\right) =\frac{1}{(2\pi)^2}\int_{\overline{\mathfrak{B}}_{\epsilon_F+c_{V_{\nu,m}}}}\!\!\!\!\!\!\!\!\m{Tr}_{L^2(\bb R)}\left(|T_q-\epsilon_F|^{1/2}\left(\overline{Q}_{m,q}^{ ++}-\overline{Q}_{m,q}^{--}\right)|T_q-\epsilon_F|^{1/2}\right)\,dq  < \infty.
\label{re3}
\end{equation}
\label{lemma10}
\end{lemma}

The proof of this result can be read in Section~\ref{lemma10Sec}. In order to conclude the proof of Proposition~\ref{prop:defectstateQ_m}, it suffices to obtain the claimed expression for $F(\overline{Q}_m )$. This is a consequence of~\eqref{projetorsAlgebra1} and~\eqref{re3}, combined with arguments similar to the ones used to establish~\eqref{boundsupRegular} in Section~\ref{regularizedOpPropProof}.

\subsubsection{Proof of Lemma \ref{regularizedOpProp}}
\label{regularizedOpPropProof}
We start by stating a useful lemma. 
\begin{lemma}
  \label{boundedQR}
  Fix $\eta >0$ and an integer $0 \leq n \leq N_2$ (appearing in the covering~\eqref{partition_2ef}). There exist positive constants $a_1(n, \eta)$, $a_2(n, \eta)$, $b_1(n,\eta)$, $b_2(n,\eta)$ and $c(n,\eta)$ such that, for any $q\in \overline{\mathfrak{B}}_{\epsilon_F+c_{V_{\nu,m}}}\cap \mathfrak{M}_{\mathrm{pp}}^{\complement} \cap \left(\mathfrak{M}_{\mathrm{pp}}^{\eta}\right)^\complement$ with $|q| \in \left[ \sqrt{n\eta},\sqrt{(n+1)\eta}\right)$, there is a smoothed closed contour~$\mathcal{C}_q \subset \mathbb{C}$ enclosing the spectra of $K_q^{\eta}$ and $K_q^{\eta}+V_{\nu,m}$ below the Fermi level~$\epsilon_F$ without intersecting them, for which
    \begin{equation}
      \forall \zeta \in\mathcal{C}_q, \qquad \left|K_q^{\eta}+V_{\nu,m}-\zeta\right|\geq c(n,\eta) >0,
    \end{equation} 
    and
    \begin{equation}
      \forall \zeta \in\mathcal{C}_q, \qquad a_1(n,\eta)\left(1+\frac{|q|^2}{2}-\frac{1}{2}\frac{d^2}{dz^2}\right)\leq |K_q^{\eta}+V_{\nu,m}-\zeta| \leq a_2(n,\eta)\left(1+\frac{|q|^2}{2}-\frac{1}{2}\frac{d^2}{dz^2}\right),
      \label{bd1}
  \end{equation}
    \begin{equation}
      \forall \zeta \in\mathcal{C}_q, \qquad b_1(n, \eta)\left(1+\frac{|q|^2}{2}-\frac{1}{2}\frac{d^2}{dz^2}\right) \leq |K_q^{\eta}-\zeta| \leq b_2(n,\eta)\left(1+\frac{|q|^2}{2}-\frac{1}{2}\frac{d^2}{dz^2}\right).
      \label{bd2}
    \end{equation}
\end{lemma}

We postpone the proof of this result to Section~\ref{boundedQRSec}. Let us first prove the statement~(1) of Lemma~\ref{regularizedOpProp}. Given $q\in \overline{\mathfrak{B}}_{\epsilon_F+c_{V_{\nu,m}}}\cap \mathfrak{M}_{\mathrm{pp}}^{\complement} \cap \left(\mathfrak{M}_{\mathrm{pp}}^{\eta}\right)^\complement$, and with the curve $\mathcal{C}_q$ introduced in Lemma~\ref{boundedQR}, Cauchy's formula gives 
\begin{align*}
\overline{Q}_{m,q}^{\eta}&= -\frac{1}{2\rmi \pi} \oint_{\mathcal{C}_q}\left(\frac{1}{K_q^{\eta}+V_{\nu,m}-\zeta}-\frac{1}{K_q^{\eta}-\zeta}\right)\,d\zeta =\frac{1}{2\rmi \pi} \oint_{\mathcal{C}_q}\frac{1}{K_q^{\eta}+V_{\nu,m}-\zeta}V_{\nu,m}\frac{1}{K_q^{\eta}-\zeta}\,d\zeta.
\end{align*}
Since the function $k \mapsto f_q(k) := \left(\frac{k^2}{2}+ \frac{|q|^2}{2}+1\right)^{-1}$ belongs to $L^2(\bb R)$ with $\|f_q\|_{L^2} \leq \|f_0\|_{L^2}$, the Kato--Seiler--Simon inequality~\eqref{KatoSeilerSimon} together with~\eqref{bd1} and~\eqref{bd2} implies that, for $|q| \in \left[ \sqrt{n\eta},\sqrt{(n+1)\eta}\right)$, there exists $C(n,\eta) \in \mathbb{R}_+$ such that 
\begin{align*}
& \bnorm{\left(1+\frac{|q|^2}{2}-\frac12\frac{d^2}{dz^2}\right)\overline{Q}_{m,q}^{\eta}}_{\mathfrak{S}_2} =  \bnorm{ \frac{1}{2\rmi\pi}\left(1+\frac{|q|^2}{2}-\frac12\frac{d^2}{dz^2}\right)\oint_{\mathcal{C}_q}\left(\frac{1}{K_q^{\eta}+V_{\nu,m}-\zeta}V_{\nu,m}\frac{1}{K_q^{\eta}-\zeta}\right)d\zeta}_{\mathfrak{S}_2} \notag \\
&\qquad \leq C(n,\eta)\norm[\Bigg]{V_{\nu,m}\left(1+\frac{|q|^2}{2}-\frac12\frac{d^2}{dz^2}\right)^{-1}}_{\mathfrak{S}_2}\leq \frac{C(n,\eta)}{\sqrt{2\pi}}\norm{V_{\nu,m}}_{L^2}\norm{f_q}_{L^2}. 
\end{align*}
Similar computations show that $\overline{Q}_{m,q}^{\eta}, \left(-\frac{d^2}{dz^2}\right)\overline{Q}_{m,q}^{\eta}, \left(K_q^{\eta}+V_{\nu,m}\right)\overline{Q}_{m,q}^{\eta}$ and $|T_q-\epsilon_F|^{1/2}\overline{Q}_{m,q}^{\eta}$ all belong to~$\mathfrak{S}_2$, while $|T_q-\epsilon_F|^{1/2}\overline{Q}_{m,q}^{\eta,\pm\pm}|T_q-\epsilon_F|^{1/2}$ is in~$\mathfrak{S}_1$. This implies that $\overline{Q}_{m,q}^{\eta}\in \mathcal{K}_q$. 

The proof of~\eqref{boundsupRegular} in the statement~(2) is obtained by an easy adaptation of~\cite[Lemma 4.2]{frank2013}. We approximate $\overline{Q}_{m,q}^{\eta}\in \mathcal{K}_q$ by finite-rank operators for which we prove the equality, and then rely on a density argument based on Lemma~\ref{continuityLemma} to conclude. 

Let us finally prove that $\overline{Q}_m^{\eta}\in \mathcal{K}$. It remains to show that
\begin{align*}
\int_{\overline{\mathfrak{B}}_{\epsilon_F+c_{V_{\nu,m}}}}\m{Tr}_{L^2(\bb R)}\left(|T_q-\epsilon_F|^{1/2}\left(\overline{Q}_{m,q}^{\eta,++}-\overline{Q}_{m,q}^{\eta,--}\right)|T_q-\epsilon_F|^{1/2}\right)dq<\infty.
\end{align*}
Consider $q \in \overline{\mathfrak{B}}_{\epsilon_F+c_{V_{\nu,m}}}\cap \mathfrak{M}_{\mathrm{pp}}^{\complement} \cap \left(\mathfrak{M}_{\mathrm{pp}}^{\eta}\right)^\complement$ and an associated integer $0 \leq n \leq N_2$ such that $|q| \in \left[ \sqrt{n\eta},\sqrt{(n+1)\eta}\right)$ for the covering~\eqref{partition_2ef}. In view of~\eqref{projetorsAlgebra4},~\eqref{bd1} and~\eqref{bd2}, there exists $R(n,\eta) \in \mathbb{R}_+$ such that
\[
\begin{aligned}
& \m{Tr}_{L^2(\bb R)}\left(|T_q-\epsilon_F|^{1/2}\left(\overline{Q}_{m,q}^{\eta,++}-\overline{Q}_{m,q}^{\eta,--}\right)|T_q-\epsilon_F|^{1/2}\right) = \bnorm{\overline{Q}_{m,q}^{\eta}|T_q-\epsilon_F|^{1/2}}_{\mathfrak{S}_2}^2\\
& \qquad\qquad = \norm[\Bigg]{\frac{1}{2\rmi\pi}\oint_{\mathcal{C}_q}\left(\frac{1}{K_q^{\eta}+V_{\nu,m}-\zeta}V_{\nu,m}\frac{1}{K_q^{\eta}-\zeta}|T_q-\epsilon_F|^{1/2}\right)d\zeta }_{\mathfrak{S}_2}^2\\
& \qquad\qquad \leq R(n,\eta)\norm[\Bigg]{\left(1+\frac{|q|^2}{2}-\frac12\frac{d^2}{dz^2}\right)^{-1}V_{\nu,m}}_{\mathfrak{S}_2}^2 \leq \frac{R(n,\eta)}{\sqrt{2\pi}}\norm{V_{\nu,m}}_{L^2}^2\norm{f_q}_{L^2}^2.
\end{aligned}
\]
Therefore, since $\mathfrak{M}_{\mathrm{pp}}$ and $\mathfrak{M}_{\mathrm{pp}}^\eta$ are negligible, and $\|f_q\|_{L^2} \leq \|f_0\|_{L^2}$,
\[
\begin{aligned}
& \int_{\overline{\mathfrak{B}}_{\epsilon_F+c_{V_{\nu,m}}}}\m{Tr}_{L^2(\bb R)}\left(|T_q-\epsilon_F|^{1/2}\left(\overline{Q}_{m,q}^{\eta,++}-\overline{Q}_{m,q}^{\eta,--}\right)|T_q-\epsilon_F|^{1/2}\right)dq \\
&\qquad\qquad\qquad = \int_{\overline{\mathfrak{B}}_{\epsilon_F+c_{V_{\nu,m}}}\cap \mathfrak{M}_{\mathrm{pp}}^{\complement} \cap \left(\mathfrak{M}_{\mathrm{pp}}^{\eta}\right)^\complement}\m{Tr}_{L^2(\bb R)}\left(|T_q-\epsilon_F|^{1/2}\left(\overline{Q}_{m,q}^{\eta,++}-\overline{Q}_{m,q}^{\eta,--}\right)|T_q-\epsilon_F|^{1/2}\right)dq \\
&\qquad\qquad\qquad \leq \max_{n=0,1\cdots,N_2+1} \frac{R(n,\eta)}{\sqrt{2\pi}}\norm{V_{\nu,m}}_{L^2}^2\int_{\overline{\mathfrak{B}}_{\epsilon_F+c_{V_{\nu,m}}}}\norm{f_q}_{L^2}^2\,dq < +\infty,
\end{aligned}
\]
which allows to conclude.

\subsubsection{Proof of Lemma \ref{lemma10} }
\label{lemma10Sec}

Let us prove~\eqref{re1} by showing that 
\[
\int_{ \overline{\mathfrak{B}}_{\epsilon_F+c_{V_{\nu,m}}}}\bnorm{\overline{Q}_{m,q}^{\eta}\left|K_q^{\eta}+V_{\nu,m}-\epsilon_F\right|^{1/2}}_{\mathfrak{S}_2}^2 \,dq \geq 0
\]
is upper bounded uniformly with respect to $\eta \in [0,1]$. Remark that the state $\overline{Q}_{m}^{\eta}$ belongs to $\mathcal{K}$ for all $\eta>0$ by Lemma~\ref{regularizedOpProp}. Since $Q_{\nu,m}$ minimizes the functional~\eqref{minimizedF},
\begin{equation}
\begin{aligned}
F(\overline{Q}_{m}^{\eta})& = \utr\left(\left(T-\epsilon_F\right) \overline{Q}_{m}^{\eta}\right) + \int_{\bb R}V_{\nu,m}(z)\rho_{\overline{Q}_{m}^{\eta}}(z)\,dz\\
& =\frac{1}{(2\pi)^2}\int_{ \overline{\mathfrak{B}}_{\epsilon_F+c_{V_{\nu,m}}}}\left(\bnorm{\overline{Q}_{m,q}^{\eta}\left|T_q-\epsilon_F\right|^{1/2}}_{\mathfrak{S}_2}^2 +\int_{\bb R}V_{\nu,m}(z)\rho_{\overline{Q}_{m,q}^{\eta}}(z)\,dz\right)dq \geq F(Q_{\nu,m}).
\end{aligned}
\label{energybound1}
\end{equation}
In addition,
\[
\m{Tr}_{L^2(\bb R)}\left(\left|h_q^{\eta }\left(-\rmi \frac{d}{dz}\right)\right|\left(\overline{Q}_{m,q}^{\eta,++}-\overline{Q}_{m,q}^{\eta,--}\right)\right) =  \m{Tr}_{L^2(\bb R)}\left(\left|h_q^{\eta}\left(-\rmi \frac{d}{dz}\right)\right|^{1/2}\left(\overline{Q}_{m,q}^{\eta}\right)^2\left|h_q^{\eta}\left(-\rmi \frac{d}{dz}\right)\right|^{1/2}\right)\geq 0.
\]
Combining this inequality with~\eqref{boundsupRegular} and \eqref{energybound1} leads to the uniform bound
\begin{equation}
\begin{aligned}
 (2\pi)^2F(Q_{\nu, m}) & \leq  \int_{ \overline{\mathfrak{B}}_{\epsilon_F+c_{V_{\nu,m}}}}\left(\bnorm{\overline{Q}_{m,q}^{\eta}\left|T_q-\epsilon_F\right|^{1/2}}_{\mathfrak{S}_2}^2+\int_{\bb R}V_{\nu,m}(z)\rho_{\overline{Q}_{m,q}^{\eta}}(z)\,dz\right)dq \\
 & \leq  -\int_{ \overline{\mathfrak{B}}_{\epsilon_F+c_{V_{\nu,m}}}}\norm[\Bigg]{\overline{Q}_{m,q}^{\eta}\left|K_q^{\eta}+V_{\nu,m}-\epsilon_F\right|^{1/2}}_{\mathfrak{S}_2}^2 dq \leq 0.
\end{aligned}
\label{re4}
\end{equation}
It is therefore possible to extract a sequence $\left\{\eta_{i}\right\}_{i\geq 1}$ such that $\overline{Q}_{m,(\boldsymbol{\cdot})}^{\eta_i}\left|K_{(\boldsymbol{\cdot})}^{\eta_i}+V_{m}-\epsilon_F\right|^{1/2}$ converges weakly in $L^2(\overline{\mathfrak{B}}_{\epsilon_F+c_{V_{\nu,m}}};\mathfrak{S}_2)$ to some operator $S_{(\boldsymbol{\cdot})}\in L^2(\overline{\mathfrak{B}}_{\epsilon_F+c_{V_{\nu,m}}};\mathfrak{S}_2)$. We next show that for all $q\in \overline{\mathfrak{B}}_{\epsilon_F+c_{V_{\nu,m}}}$, $S_q = \overline{Q}_{m,q}\left|T_q+V_{\nu,m}-\epsilon_F\right|^{1/2}$. By the same reasoning as in~\cite[Section~4.2]{frank2013}, it can be shown that, for all $q\in\overline{\mathfrak{B}}_{\epsilon_F+c_{V_{\nu,m}}}$, the following strong convergence holds in~$L^2(\mathbb{R})$: for any $\varphi \in H^2(\mathbb{R})$,
\begin{equation}
\label{eq:strong}
\overline{Q}_{m,q}^{\eta_i}\left|K_q^{\eta_i}+V_{\nu,m}-\epsilon_F\right|^{1/2}\varphi \xrightarrow[i \to +\infty]{}\overline{Q}_{m,q}\left|T_q+V_{\nu,m}-\epsilon_F\right|^{1/2} \varphi\quad \text{in } L^2(\mathbb{R}).
\end{equation}
Now, consider $\phi\in H^2(\bb R)$ and the associated rank-1 operator $K:= \Ket{\phi}\Bra{\phi}\in \mathfrak{S}_1$. In particular $K$ can be seen as a constant function of $L^2(\overline{\mathfrak{B}}_{\epsilon_F+c_{V_{\nu,m}}};\mathfrak{S}_2)$. On the one hand, by the weak convergence of $\overline{Q}_{m,(\boldsymbol{\cdot})}^{\eta_i}\left|K_{(\boldsymbol{\cdot})}^{\eta_i}+V_{\nu,m}-\epsilon_F\right|^{1/2}$ to $S_{(\boldsymbol{\cdot})}$ in $L^2(\overline{\mathfrak{B}}_{\epsilon_F+c_{V_{\nu,m}}};\mathfrak{S}_2)$,
\begin{equation}
\begin{aligned}
&\int_{\overline{\mathfrak{B}}_{\epsilon_F+c_{V_{\nu,m}}}}\left\langle \phi\left|\, \overline{Q}_{m,q}^{\eta_i}\left|K_q^{\eta_i}+V_{\nu,m}-\epsilon_F\right|^{1/2} \right|\,\phi\right\rangle dq = \int_{\overline{\mathfrak{B}}_{\epsilon_F+c_{V_{\nu,m}}}}\m{Tr}_{L^2(\bb R)}\left(\overline{Q}_{m,q}^{\eta_i}\left|K_q^{\eta_i}+V_{\nu,m}-\epsilon_F\right|^{1/2}K\right) dq \\
& \qquad \xrightarrow[i \to +\infty]{} \int_{\overline{\mathfrak{B}}_{\epsilon_F+c_{V_{\nu,m}}}}\m{Tr}_{L^2(\bb R)}\left(S_qK\right) dq = \int_{\overline{\mathfrak{B}}_{\epsilon_F+c_{V_{\nu,m}}}}\left\langle\phi\left|\, S_q \right|\,\phi \right\rangle \, dq .
\end{aligned}
\label{weak1}
\end{equation}
On the other hand, by the strong convergence~\eqref{eq:strong}, it holds, for all $q\in\overline{\mathfrak{B}}_{\epsilon_F+c_{V_{\nu,m}}}$,
\[
\left\langle\phi\left|\,\overline{Q}_{m,q}^{\eta_i}\left|K_q^{\eta_i}+V_{\nu,m}-\epsilon_F\right|^{1/2}\right|\phi\right\rangle \xrightarrow[i\to +\infty]{}\left\langle\phi\left| \,\overline{Q}_{m,q}\left|T_q+V_{\nu,m}-\epsilon_F\right|^{1/2}\right|\, \phi\right\rangle.
\]
A simple computation shows that there exists $C_\phi \in \mathbb{R}_+$ such that $\left|\left\langle\phi\left| \,\overline{Q}_{m,q}^{\eta_i}\left|K_q^{\eta_i}+V_{\nu,m}-\epsilon_F\right|^{1/2}\right|\,\phi\right\rangle \right|\leq C_\phi$ for all $i \geq 1$ and $q \in \overline{\mathfrak{B}}_{\epsilon_F+c_{V_{\nu,m}}}$. By the dominated convergence theorem, 
\[
\int_{\overline{\mathfrak{B}}_{\epsilon_F+c_{V_{\nu,m}}}}\left\langle\phi\left|\,\overline{Q}_{m,q}^{\eta_i}\left|K_q^{\eta_i}+V_{\nu,m}-\epsilon_F\right|^{1/2}\right| \,\phi\right \rangle dq \xrightarrow[i\to +\infty]{}\int_{\overline{\mathfrak{B}}_{\epsilon_F+c_{V_{\nu,m}}}}\left\langle\phi\left| \,\overline{Q}_{m,q}\left|T_q+V_{\nu,m}-\epsilon_F\right|^{1/2}\right|\,\phi\right\rangle dq.
\]
The comparison of the previous limit and~\eqref{weak1} implies that $\overline{Q}_{m,(\boldsymbol{\cdot})}\left|T_{(\boldsymbol{\cdot})}+V_{\nu,m}-\epsilon_F\right|^{1/2} = S_{(\boldsymbol{\cdot})} \in L^2\left(\overline{\mathfrak{B}}_{\epsilon_F+c_{V_{\nu,m}}};\mathfrak{S}_2\right)$.

We next prove~\eqref{re2} by showing that
\[
\int_{\overline{\mathfrak{B}}_{\epsilon_F+c_{V_{\nu,m}}}}\bnorm{|T_q-\epsilon_F|^{1/2}\overline{Q}_{m,q}^{\eta}}_{\mathfrak{S}_2}^2 \,dq\geq 0 
\]
is upper bounded uniformly with respect to~$\eta\in [0,1]$. Introduce 
\[
\rho_{\overline{Q}_{m}^{\eta}}(z) = \int_{\mathbb{R}^2}\rho_{\overline{Q}_{m,q}^{\eta}}(z)\,dq = \int_{\overline{\mathfrak{B}}_{\epsilon_F+c_{V_{\nu,m}}}}\rho_{\overline{Q}_{m,q}^{\eta}}(z)\,dq,
\]
which is well defined by Proposition~\ref{prop:densityK} since $\overline{Q}_{m}^{\eta} \in \mathcal{K}$. Moreover, in view of~\eqref{LTtotalenergybound} with $c = c_{V_{\nu,m}}$, there exists a constant $R \in \mathbb{R}_+$ such that 
\[
\left\|\rho_{\overline{Q}_{m}^{\eta}}\right\|_{L^2}^2 \leq R \int_{\overline{\mathfrak{B}}_{\epsilon_F+c_{V_{\nu,m}}}}\bnorm{|T_q-\epsilon_F|^{1/2}\overline{Q}_{m,q}^{\eta}}_{\mathfrak{S}_2}^2 dq.
\]
In view of~\eqref{re4},
\[
\int_{\overline{\mathfrak{B}}_{\epsilon_F+c_{V_{\nu,m}}}}\bnorm{|T_q-\epsilon_F|^{1/2}\overline{Q}_{m,q}^{\eta}}_{\mathfrak{S}_2}^2 dq \leq -\int_{\bb R}V_{\nu,m}\rho_{\overline{Q}_{m}^{\eta}},
\]
so that a Cauchy--Schwarz inequality leads to
\[
0\leq \int_{\overline{\mathfrak{B}}_{\epsilon_F+c_{V_{\nu,m}}}} \bnorm{|T_q-\epsilon_F|^{1/2}\overline{Q}_{m,q}^{\eta}}_{\mathfrak{S}_2}^2 dq \leq R \norm{V_{\nu,m}}_{L^2}^2.
\]By arguments similar to the ones used above, we deduce that it is possible to further extract a subsequence (still denoted by $(\eta_i)_{i \geq 1}$ with some abuse of notation) such that $\overline{Q}_{m,(\boldsymbol{\cdot})}^{\eta_i}\left|T_{(\boldsymbol{\cdot})}-\epsilon_F\right|^{1/2}\xrightharpoonup[i\to +\infty]{}\overline{Q}_{m,(\boldsymbol{\cdot})}\left|T_{(\boldsymbol{\cdot})}-\epsilon_F\right|^{1/2}$ weakly in $L^2(\overline{\mathfrak{B}}_{\epsilon_F+c_{V_{\nu,m}}};\mathfrak{S}_2)$ and
\[
0  \leq \int_{\overline{\mathfrak{B}}_{\epsilon_F+c_{V_{\nu,m}}}}\bnorm{\overline{Q}_{m,q}\left|T_q-\epsilon_F\right|^{1/2}}_{\mathfrak{S}_2}^2\,dq \leq  \liminf_{i \to +\infty} \int_{\overline{\mathfrak{B}}_{\epsilon_F+c_{V_{\nu,m}}}}\bnorm{\overline{Q}_{m,q}^{\eta_i}\left|T_q-\epsilon_F\right|^{1/2}}_{\mathfrak{S}_2}^2 dq < \infty.
\]
This allows us to conclude that $\overline{Q}_m\in\mathcal{K}$ using~\eqref{projetorsAlgebra3}, and shows that~\eqref{re3} holds.

\subsubsection{Proof of Lemma \ref{boundedQR}}
\label{boundedQRSec}

Since $V_{\nu,m}$ is a compact perturbation of $K_q^\eta$ for any $q \in \mathbb{R}^2$, the only possible accumulation points of the discrete spectrum of $K_q^{\eta}+V_{\nu,m}$ are $|q|^2/2$, $\epsilon_F-\eta+\frac{|q|^2-n\eta}{2}$ and $\epsilon_F+\eta+\frac{|q|^2-n\eta}{2}$, where $0 \leq n \leq N_2$ is such that $\sqrt{n\eta} \leq |q| < \sqrt{(n+1)\eta}$. For the remainder of the proof, we fix an integer $0 \leq n \leq N_2$ and an element $q_0 \in \mathbb{R}^2$ with $|q_0|=\sqrt{n\eta}$. The objects we introduce are depicted in Figure~\ref{FigSpectralN}.
\begin{figure}[h!]
\centering
\includegraphics[height=0.5\textwidth]{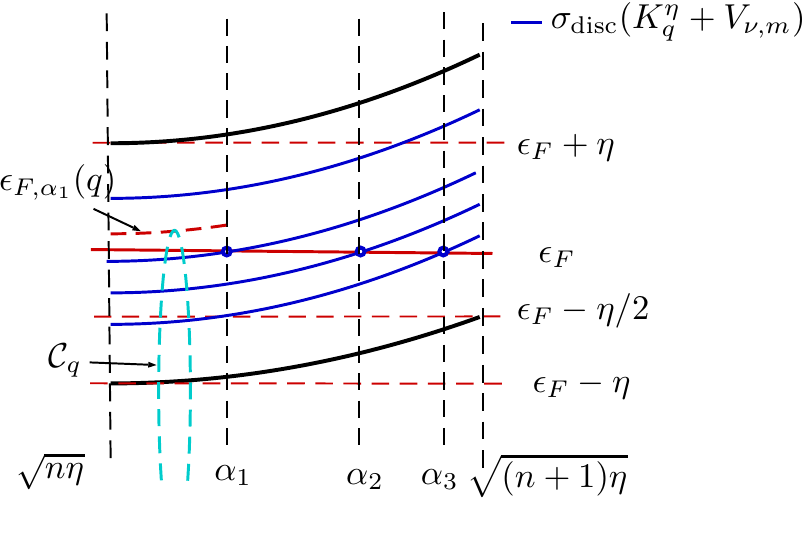}~
\caption{Spectral structure of $K_q^{\eta}+V_m$ around the Fermi level $\epsilon_F$ for $|q| \in \left[\sqrt{n\eta},\sqrt{(n+1)\eta}\right)$. The dark blue lines correspond to elements of $\sigma_{\mathrm{disc}}\left(K_q^{\eta}+V_{\nu,m}\right)$ which are obtained from those at $q_0$ such that $|q_0| = \sqrt{n\eta}$ by adding $(|q|^2-|q_0|^2)/2$. The points~$q$ such that $|q|=\alpha_i$ are removed. The curve $\mathcal{C}_q$ (dashed light blue line) is defined to pass through $\epsilon_{F,\alpha_1}(q)$ for $|q| \in \left[\sqrt{n\eta},\alpha_1\right)$ . }
\label{FigSpectralN}
\end{figure}

First of all remark that the gap-opening function $h_q^{\eta}$ is independent of $|q|$ in each interval $\left[\sqrt{n\eta},\sqrt{(n+1)\eta}\right)$. Introduce the integer $M\geq 0$ such that
\[
\sigma_{\mathrm{disc}}\left(K_{q_0}^{\eta}+V_{\nu,m}\right) \cap \left[\epsilon_F - \frac{3\eta}{4},\epsilon_F+\frac{\eta}{4}\right] = \left\{\varepsilon_i\right\}_{0 \leq i\leq M}, 
\qquad 
\epsilon_F -\frac{3\eta}{4} \leq \varepsilon_0 < \cdots < \varepsilon_M \leq \epsilon_F + \frac{\eta}{4}.
\]
Note that we restrict the spectral window around~$\epsilon_F$ in order to discard the possible accumulation points of the discrete spectrum. We assume that $q_0 \not \in \mathfrak{M}_{\rm pp}$; the case when $q_0\in \mathfrak{M}_{\rm pp} $ can however be handled by a simple modification of our arguments. Note that, since $K_{q}^{\eta} = K_{q_0}^{\eta} + \frac{|q|^2-|q_0|^2}{2}$,
\[
\sigma_{\mathrm{disc}}\left(K_{q}^{\eta}+V_{\nu,m}\right) \cap \left[\epsilon_F - \frac{\eta}{4},\epsilon_F+\frac{\eta}{4}\right] \subset \frac{|q|^2-|q_0|^2}{2} + \sigma_{\mathrm{disc}}\left(K_{q_0}^{\eta}+V_{\nu,m}\right) \cap \left[\epsilon_F - \frac{3\eta}{4},\epsilon_F+\frac{\eta}{4}\right],
\]
and the distances between the eigenvalues of $K_q^{\eta}+V_{\nu,m}$ in the spectral region $\left[\epsilon_F-\eta/4,\epsilon_F+\eta/4\right]$ are independent of $q$. 

Note next that we eliminated the values of~$q$ for which $\epsilon_F \in \sigma_{\mathrm{disc}}(K_q^{\eta}+V_{\nu,m})$ by the definition of the set $\mathfrak{M}_{\mathrm{pp}}^{\eta}$. The number of the corresponding values of~$|q|$ is at most equal to $M$, which corresponds to a further partition of the interval $\left[\sqrt{n\eta},\sqrt{(n+1)\eta}\right)$ into finitely many subintervals $[\alpha_i,\alpha_{i+1})$ with $\alpha_0:= |q_0| = \sqrt{n\eta}$.  

Let us concentrate on the region $\mathscr{R} = \{ q \in \mathbb{R}^2, n\eta \leq |q| <\alpha_1 \}$. As $\sigma_{\mathrm{disc}}\left(K_{q}^{\eta}+V_{\nu,m}\right)$ intersects the Fermi level $\epsilon_F$ at $|q|=\alpha_1$, there exists a unique integer $0\leq  j \leq M$ such that $ \varepsilon_{j} +\frac{|\alpha_1|^2 -|q_0|^2}{2} = \epsilon_F $. We now consider curves $\mathcal{C}_q \subset \mathbb{C}$ for $q \in \mathscr{R}$ satisfying the following properties:
\begin{itemize}
\item the top of the curve passes through 
\[
\epsilon_{F,\alpha_1}(q):= \frac{\varepsilon_{j+1}+\varepsilon_{j}}{2} + \frac{|q|^2-|q_0|^2}{2},\quad |q|\in \left[\sqrt{n\eta},\alpha_1\right),
\]
with $\varepsilon_{M+1}:= \epsilon_F+\eta/4$. Note that, when $j=M$, it holds $\varepsilon_M \leq \epsilon_F$, so that the above definition makes sense since there is no other eigenvalue in the spectral region~$(\varepsilon_M,\epsilon_F+\eta/4]$. Note also that the distance between $\epsilon_{F,\alpha_1}(q)$ and $\sigma_{\mathrm{disc}}\left(K_q^{\eta}+V_{\nu,m}\right) \cap [\epsilon_F-\eta/2,\epsilon_F]$ is independent of $|q|$, and that $\mathds 1_{(-\infty,\epsilon_{F,\alpha_1}(q)]}(K_q^{\eta}+V_{\nu,m}) = \mathds 1_{(-\infty,\epsilon_F]}(K_q^{\eta}+V_{\nu,m})$, \emph{i.e.} the states below the Fermi level have not been changed. 
\item the curve encloses the spectrum of $K_q^{\eta}+V_{\nu,m}$ below $\epsilon_F$ without intersecting it.
\end{itemize} 
It is then easy to see that there exists a positive constant $c_1(n,\eta)$ such that $\left|K_q^{\eta}+V_{\nu,m}-\zeta\right|\geq c_1(n,\eta) > 0$ uniformly for $|q|\in \left[\sqrt{n\eta},\alpha_1\right)$ and $\zeta\in\mathcal{C}_q$. 

The same procedure can be applied to other intervals. As the number of intervals is finite, there exists a constant $c(n,\eta)$ (which depends only on the properties of the discrete spectrum of $K_{q_0}^{\eta}+V_{\nu,m}$ in the spectral region $[\epsilon_F-3\eta/4,\epsilon_F+\eta/4]$) such that, for any $q \in \mathbb{R}^2$ with $\sqrt{n\eta} \leq |q| < \sqrt{(n+1)\eta}$, there is a curve $\mathcal{C}_q \subset \mathbb{C}$ such that $\left|K_q^{\eta}+V_{\nu,m}-\zeta\right|\geq c(n,\eta) >0$ for any $\zeta\in\mathcal{C}_q$. The proof of inequalities~\eqref{bd1} and~\eqref{bd2} is then easily obtained by replacing $-\frac{1}{2}\frac{d^2}{dz^2}$ with $-\frac{1}{2}\frac{d^2}{dz^2}+\frac{|q|^2}{2}$ in~\cite[Lemma~1]{Cances2008}.

\appendix
\section{Proofs of some technical results}

\subsection{Proof of Lemma \ref{qInMLemma}}
\label{lemmaQ}

Let us first state a useful technical result showing that finite-rank operators are dense in $\mathcal{X}_q$ and $\mathcal{K}_q$. It is a direct adaptation of~\cite[Lemma 3.2]{frank2013}.

\begin{lemma}
  \label{continuityLemma}
  Fix $q\in \bb R^2$ and consider $Q_{q} \in \mathcal{X}_{q}$. There exists a sequence $(Q_{n,q})_{n \geq 1} \subset \mathcal{X}_q$ of finite-rank operators such that $T_{q}Q_{n,q}\in \mathcal{L}(L^2(\bb R))$, and
  \begin{itemize}
  \item $Q_{n,q} \to Q_{q}$ strongly (\textit{i.e.}, $Q_{n,q}f \to Q_{q}f$ strongly in $L^2(\bb R)$ for any $f\in L^2(\bb R) $);
  \item $\displaystyle \lim_{n\to +\infty}\norm[\bigg]{ |T_q-\epsilon_F|^{1/2}(Q_{n,q}^{\pm\pm}-Q_{q}^{\pm\pm}) |T_q-\epsilon_F|^{1/2}}_{\mathfrak{S}_1(L^2(\bb R))}=0$.
  \end{itemize}
  Moreover, if $Q_{q} \in \mathcal{K}_{q}$, the sequence $(Q_{n,q})_{n \geq 1}$ can be chosen in~$\mathcal{K}_{q}$.
\end{lemma}

We can now provide the proof of Lemma~\ref{qInMLemma}, which follows the proof of~\cite[Lemma~3.3]{frank2013} and uses in particular ideas from~\cite{rumin2011balanced}. Let $0 \leq \gamma \leq 1 $ be a smooth enough self-adjoint finite-rank operator on $L^2(\bb R)$ (with density~$\rho_\gamma$). For $q\in \partial \overline{\mathfrak{B}}_{\epsilon_F}$, $T_q-\epsilon_F = -\frac{1}{2}\frac{d^2}{dz^2}$ and~\eqref{qInM} boils down to the one-dimensional Lieb--Thirring inequality~\cite{Lieb2005}. For $q \in \mathfrak{B}_{\epsilon_F}$, denote by $\rho_{e,q}$ is the density of the finite-rank operator $P_{e,q\,}\gamma \,P_{e,q}$ where $P_{e,q}:= \mathds 1_{[e,+\infty)}\left(|T_q-\epsilon_F|\right)$. Then, by the same manipulations as in~\cite{frank2013},
\[
\m{Tr}_{L^2(\bb R)}\left(|T_q-\epsilon_F|^{1/2}\gamma|T_q-\epsilon_F|^{1/2}\right) = \int_0^{+\infty}\left(\int_{\bb R}\rho_{e,q}(z)\,dz\right)de \geq \int_{\bb R}R_q(\rho_{\gamma}(z))\,dz,
\]
with
\[
R_q(t) = \int_0^{f_{q}^{-1}( t)}\left( \sqrt{t}-\sqrt{f_q(e)}\right)^2 de,
\]
where (introducing $a_+ = \max(0,a)$) the nonnegative function
\[
f_q(e) = (2\pi)^{-1}\left\lvert \left\{p\in \bb R\,\left.,\left\lvert\frac{|p|^2}{2}+\frac{|q|^2}{2}-\epsilon_F\right\rvert\leq e\right.\right\}\right\rvert = (2\pi)^{-1}\left( \sqrt{2e+2\epsilon_F- |q|^2}-\sqrt{\left(-2e+2\epsilon_F- |q|^2\right)_{+}}\right)
\]
is increasing, hence invertible.

Let us now provide a global lower bound on the function~$t \mapsto R_q(t)$ by considering the asymptotic behavior of this function in the regimes $t \to 0$ and $t \to +\infty$. We work in fact with the rescaled parameter $T = t/\omega_q$ and the rescaled energy $E = e/\omega_q^2$ with $\omega_q = \sqrt{2\epsilon_F-|q|^2} > 0$ (since $q\in \mathfrak{B}_{\epsilon_F}$). Note indeed that
\[
f_q(e) = \omega_q F\left(\frac{e}{\omega_q^2}\right), \qquad F(E) = \frac{1}{2\pi} \left(\sqrt{1+2E}-\sqrt{(1-2E)_+}\right),
\]
and
\[
R_q(t) = \int_0^{+\infty} \left[\left(\sqrt{t}-\sqrt{f_q(e)}\right)_+\right]^2 de = \omega_q^{3} \int_0^{+\infty} \left[\left(\sqrt{\frac{t}{\omega_q}}-\sqrt{F(E)}\right)_+\right]^2 dE.
\]
Since $F(E) \sim \pi^{-1}E$ as $E \to 0$ and $F(E) \sim \pi^{-1}\sqrt{E/2}$ as $E \to +\infty$, a simple computation shows that
\[
\int_0^{+\infty} \left[\left(\sqrt{T}-\sqrt{F(E)}\right)_+\right]^2 dE \mathop{\sim}_{T \to 0} \frac{\pi}{6} T^2,
\qquad
\int_0^{+\infty} \left[\left(\sqrt{T}-\sqrt{F(E)}\right)_+\right]^2 dE \mathop{\sim}_{T \to \infty} \frac{2\pi^2}{15} T^3.
\]
There exist therefore two positive constants $C_{1,{\rm diag}},C_{2,{\rm diag}}$ such that
\[
\forall T \geq 0, \qquad \int_0^{+\infty} \left[\left(\sqrt{T}-\sqrt{F(E)}\right)_+\right]^2 dE \geq C_{1,{\rm diag}} T^2 + C_{2,{\rm diag}} T^3,
\]
from which we deduce that
\[
\forall t \geq 0, \qquad R_q(t) \geq  C_{1,{\rm diag}} \omega_q t^2 + C_{2,{\rm diag}}  t^3.
\]
The final result is obtained by a continuity argument and the density result of Lemma~\ref{continuityLemma}.

\subsection{Proof of Lemma \ref{offdiagonal}}
\label{lemmaOffdiagonal}

We follow the proofs of~\cite[Theorem~2.1 and Lemma~3.4]{frank2013}. Inequality~\eqref{diag_LT} is trivial for $q\in \partial\overline{\mathfrak{B}}_{\epsilon_F}$. Fix $q\in \mathfrak{B}_{\epsilon_F}$ and $Q_{q}\in \mathcal{K}_{q}$. Since $Q_q^{+-} = (Q_q^{-+})^*$, it holds $\rho_{Q_{q}^{-+}}+\rho_{Q_{q}^{+-}} = 2\mathfrak{R}\rho_{Q_{q}^{-+}}$. It suffices therefore to obtain estimates for $\rho_{Q_q^{-+}}$. We rely on a duality argument. Consider to this end $V\in L^2(\bb R)$. By proceeding as in~\cite{frank2013},
\begin{equation}
  \label{offDiagDual}
  \left|\int_{\bb R} V\rho_{Q_{q}^{-+}}\right| \leq \bnorm{\frac{1-\gamma_{0,q}}{|T_q-\epsilon_F|^{1/4}}V\frac{\gamma_{0,q}}{|T_q-\epsilon_F|^{1/4}}}_{\mathfrak{S}_2} \sqrt{ \m{Tr}_{L^2(\bb R)}\left(|T_q-\epsilon_F|^{1/2}\left(Q_{q}^{++}-Q_q^{--}\right)|T_q-\epsilon_F|^{1/2}\right) }.
\end{equation}
Now,
\begin{equation}
\label{potentialestimates}
\bnorm{\frac{1-\gamma_{0,q}}{|T_q-\epsilon_F|^{1/4}}V\frac{\gamma_{0,q}}{|T_q-\epsilon_F|^{1/4}}}_{\mathfrak{S}_2}^2
=\frac{1}{2\pi}\int_{\bb R} |\widehat{V}(k)|^2\Phi_q(k)\,dk,
\end{equation}
with
\[
\Phi_q(k) := \int_{
\substack{\frac{|q|^2}{2}+\frac{(s-k)^2}{2} \geq \epsilon_F\\
\frac{|q|^2}{2}+\frac{s^2}{2} \leq \epsilon_F}} \frac{ds}{\sqrt{\left(\frac{|q|^2}{2}+\frac{(s-k)^2}{2}-\epsilon_F\right)\left(\epsilon_F-\frac{|q|^2}{2}-\frac{s^2}{2}\right)}}.
\]
In fact, denoting by $\omega_q = \sqrt{2\epsilon_F-|q|^2} > 0$ (since $q\in \mathfrak{B}_{\epsilon_F}$),
\[
\Phi_q(k) = \frac{1}{\omega_q} \Psi\left(\frac{k}{\omega_q}\right),
\qquad
\Psi(t) = 2 \int_{\substack{(m-t)^2\geq 1\\ m^2\leq 1}} \frac{dm}{\sqrt{((m-t)^2-1)(1-m^2)}}.
\]
The bounds~\eqref{offDiagDual} and~\eqref{potentialestimates} then lead to
\[
\left|\int_{\bb R} V\rho_{Q_{q}^{-+}}\right|^2 \leq\frac{1}{2\pi} \left(\int_{\bb R}\left|\widehat{V}(k)\right|^2\Phi_q(k)\,dk\right) \m{Tr}_{L^2(\bb R)}\left(|T_q-\epsilon_F|^{1/2}\left(Q_{q}^{++}-Q_q^{--}\right)|T_q-\epsilon_F|^{1/2}\right).
\]
From the estimates on~$\Psi$ provided by~\cite[Lemma~3.4]{frank2013} (namely $\Psi(t) \sim -\log|t-2|$ as $t\to 2$ and $\Psi(t) \sim 2\pi/t$ as $t\to+\infty$), there exists a positive constant~$R$ such that
\begin{align}
  \label{phiqk}
\forall 0 < w_q \leq \sqrt{2\epsilon_F}, \,\forall k \in \bb R, \quad  0 \leq \Phi_q(k) \leq \frac{R}{\left\lvert |k|-2\omega_q\right\rvert^{1/2}\sqrt{\omega_q}},
\end{align}
which implies
\begin{align*}
\frac{2\pi}{R} \sqrt{\omega_q} \int_{\bb R}|\widehat{\rho_{Q_{q}^{-+}}}(k)|^2 \left\lvert |k|-2\omega_q\right\rvert^{1/2} \,dk & \leq 2\pi\int_{\bb R}\frac{\left|\widehat{ \rho_{Q_{q}^{-+}}}(k)\right|^2}{\Phi_q(k)} \,dk \\
& \leq \m{Tr}_{L^2(\bb R)}\left(|T_q-\epsilon_F|^{1/2}\left(Q_{q}^{++}-Q_q^{--}\right)|T_q-\epsilon_F|^{1/2}\right).
\end{align*}
This gives the claimed result.

\subsection{Proof of Lemma \ref{consistencyOfDensity}}
\label{consistencyOfDensitySec}

We prove the result for Coulomb interactions. The statement of the lemma and its proof for Yukawa interactions are obtained by a straightforward adaptation.

\paragraph{Equality of $\widetilde\rho_{Q}$ and $\overline{\overline{ \rho}}_{Q}$.} Let us first show that $\widetilde\rho_{Q}-\nu = \overline{\overline{ \rho}}_{Q}-\nu$ in $\mathcal{D}'$. Fix $w\in \mathcal{D}(\bb R)$. The weak convergence $\rho_{Q_n}\rightharpoonup \overline{\overline{\rho}}_Q$ in $ L^2(\bb R) + L^p(\bb R)$ implies that
\[
\langle \rho_{Q_n}-\nu, w\rangle\toinfty \langle  \overline{\overline{ \rho}}_{Q}-\nu,w\rangle.
\]
Note next that
\begin{align*}
\langle \rho_{Q_n}-\nu,w\rangle & = \int_{\bb R}\left(\rho_{Q_n}-\nu\right)w = \int_{\bb R}\overline{\left(\widehat{\rho}_{Q_n}-\widehat{\nu}\right)}(k)\widehat{w}(k)\,dk \\
& = 2 \int_{\bb R}\frac{\overline{\left(\widehat{\rho}_{Q_n}-\widehat{\nu}\right)}(k)\widehat{f}(k)}{|k|^2}\,dk = \ud{D}\left(\rho_{Q_n}-\nu, f\right),
\end{align*}
where we introduced $f = -w''/2$. Note that $f\in \ud{\mathcal{C}}$ since $\widehat{f}\in L_{\mathrm{loc}}^1(\bb R)$ and $k \mapsto \frac{\widehat{f}(k)}{|k|} = \frac{1}{2}|k|\widehat{w}(k)$ belongs to $L^2(\mathbb{R})$ because $\norm{k\widehat{w}}_{L^2(\bb R)}^2 = \norm{w'}_{L^2(\bb R)}^2 < +\infty$. The convergence $\ud{D}\left(\rho_{Q_n}-\nu, f\right)\toinfty \ud{D}\left(\widetilde \rho_{Q}-\nu, f\right) $ then implies that $\langle \rho_{Q_n}-\nu,w\rangle\toinfty \langle \widetilde \rho_{Q}-\nu,w\rangle$. The uniqueness of the limit in the sense of distributions finally shows that $ \overline{\overline{ \rho}}_{Q}-\nu = \widetilde{\rho}_Q-\nu $ in $\mathcal{D}'(\bb R)$.

\paragraph{Equality of $\rho_{\overline{Q}}$ and $\overline{\overline{\rho}}_Q$.}
Fix $w \in \mathcal{D}(\bb R)$. The weak convergence $\rho_{Q_n} \rightharpoonup \overline{\overline{ \rho}}_{Q} $ in $L^2(\bb R)+L^p(\bb R)$ implies
\[
\langle \rho_{Q_n}, w\rangle \toinfty \langle \overline{\overline{ \rho}}_{Q}, w\rangle.
\]
It therefore suffices to prove that the operator-valued function $q \mapsto wQ_{n,q}$ belongs to $L^1(\mathbb{R}^2 ;\mathfrak{S}_1)$ and that
\begin{equation}
  \label{eq:cv_Tr_wQ}
  \frac{1}{(2\pi)^2}\int_{\bb R^2}\m{Tr}_{L^2(\bb R)}\left(wQ_{n,q}\right)\,dq =\ud{\m{Tr}}\left(wQ_n\right) =\langle \rho_{Q_n},w\rangle \xrightarrow[n\to\infty]{} \langle \rho_{\overline{Q}},w\rangle\ = \ud{\m{Tr}}\left(w\overline{Q}\right) = \langle \rho_{\overline{Q}},w\rangle.
\end{equation}
In order to prove the above convergence, we split the integration domain for $q \in \mathbb{R}^2$ into three parts as $\bb R^2 = \overline{\mathfrak{B}}_{\epsilon_F+c} \cup (\mathfrak{B}_{R}\backslash \overline{\mathfrak{B}}_{\epsilon_F+c}) \cup (\bb R^2 \backslash \mathfrak{B}_R)$, where $R>2\epsilon_F$ is large enough and $0 < c < R-\epsilon_F$.

Consider first the case when $q\in \bb R^2 \backslash \mathfrak{B}_R$. For these values of~$q$, the operator $|T_q-\epsilon_F|^{-1/2}$ is bounded, with operator norm smaller than $\left(\frac{|q|^2}{2}-\epsilon_F\right)^{-1}$. Moreover,
\[
\forall q \in \bb R^2 \backslash \mathfrak{B}_R, \qquad \left\| |T_q-\epsilon_F|^{-1/2}w|T_q-\epsilon_F|^{-1/2} \right\| \leq \norm{w}_{L^{\infty}}\left|R-\epsilon_F\right|^{-1}.
\]
Note also that $|T_q-\epsilon_F|^{1/2}Q_{n,q}|T_q-\epsilon_F|^{1/2} = |T_q-\epsilon_F|^{1/2}Q_{n,q}^{++}|T_q-\epsilon_F|^{1/2} \in \mathfrak{S}_1$ since $|q|^2 > 2\epsilon_F$. Therefore, the operator $wQ_{n,q}$ is trace-class and
\begin{align*}
\left| \m{Tr}_{L^2(\bb R)}\left(w Q_{n,q}\right) \right| & = \left|\m{Tr}_{L^2(\bb R)}\left(|T_q-\epsilon_F|^{-1/2}w|T_q-\epsilon_F|^{-1/2}|T_q-\epsilon_F|^{1/2}Q_{n,q}|T_q-\epsilon_F|^{1/2}\right)\right|\\
& \leq \norm{w}_{L^{\infty}}\left|R-\epsilon_F\right|^{-1} \left\| |T_q-\epsilon_F|^{1/2}Q_{n,q}^{++}|T_q-\epsilon_F|^{1/2} \right\|_{\mathfrak{S}_1}.
\end{align*}
Integrating over $q \in \bb R^2 \backslash \mathfrak{B}_R$ and relying on the uniform bound~\eqref{uniformBound}, we finally obtain
\begin{equation}
  \label{lastPart}
  \left| \int_{ \bb R^2 \backslash \mathfrak{B}_R}\m{Tr}_{L^2(\bb R)}\left(w Q_{n,q}\right)dq \right| \leq C\norm{w}_{L^{\infty}}\left|R-\epsilon_F\right|^{-1}.
\end{equation}
This term therefore vanishes as $R \to +\infty$. Note that a similar inequality holds if $Q_{n,q}$ is replaced by $\overline{Q}_{n,q}$.

Consider next the case when $q\in \mathfrak{B}_{R}\backslash \overline{\mathfrak{B}}_{\epsilon_F+c} $. The Kato--Seiler--Simon inequality~\eqref{KatoSeilerSimon} shows that $w|T_q-\epsilon_F|^{-1/2}$ is Hilbert--Schmidt, and $q\mapsto w|T_q-\epsilon_F|^{-1/2}$ is in $L^2(\mathfrak{B}_{R}\backslash \overline{\mathfrak{B}}_{\epsilon_F+c}; \mathfrak{S}_2)$. The convergence~\eqref{fabile2} then shows that
\begin{equation}
  \label{firstpartConvergence}
  \begin{aligned}
    &\int_{\mathfrak{B}_{R}\backslash \overline{\mathfrak{B}}_{\epsilon_F+c}}\m{Tr}_{L^2(\bb R)}\left(wQ_{n,q}\right) dq = \int_{\mathfrak{B}_{R}\backslash \overline{\mathfrak{B}}_{\epsilon_F+c}}\m{Tr}_{L^2(\bb R)}\left(w|T_q-\epsilon_F|^{-1/2} |T_q-\epsilon_F|^{1/2}Q_{n,q}\right)dq\\
    &\qquad \toinfty \int_{\mathfrak{B}_{R}\backslash \overline{\mathfrak{B}}_{\epsilon_F+c}}\m{Tr}_{L^2(\bb R)} \left(w|T_q-\epsilon_F|^{-1/2} |T_q-\epsilon_F|^{1/2}\overline{Q}_{q}\right) dq  = \int_{\mathfrak{B}_{R}\backslash \overline{\mathfrak{B}}_{\epsilon_F+c}}\m{Tr}_{L^2(\bb R)}\left(w\overline{Q}_{q}\right) dq.
  \end{aligned}
\end{equation}

Consider finally the case when $q\in \overline{\mathfrak{B}}_{\epsilon_F+c} $. Define $\Pi_{1,q}:= \mathds 1_{(-\infty,2\epsilon_F]}\left(T_{q}\right)$ and $\Pi_{2,q}:= 1-\Pi_{1,q}$. We decompose the operator $wQ_{n,q}$ as $w\Pi_{2,q}Q_{n,q} + \Pi_{1,q}w\Pi_{1,q}Q_{n,q} + \Pi_{2,q}w\Pi_{1,q}Q_{n,q} $. We show successively that these three operators are trace-class, and characterize their limits as $n \to +\infty$. Note first that $w\Pi_{2,q}Q_{n,q} = w\Pi_{2,q}|T_q-\epsilon_F|^{-1/2} |T_q-\epsilon_F|^{1/2} Q_{n,q}$ is the product of two Hilbert--Schmidt operators. In fact, a simple computation based on the Kato--Seiler--Simon inequality~\eqref{KatoSeilerSimon} shows that $q \mapsto w\Pi_{2,q}|T_q-\epsilon_F|^{-1/2} \in L^{2}\left(\overline{\mathfrak{B}}_{\epsilon_F+c};\mathfrak{S}_2\right)$. Therefore, by~\eqref{fabile2},
\begin{equation}
  \label{b}
  \int_{ \overline{\mathfrak{B}}_{\epsilon_F+c}}\m{Tr}_{L^2(\bb R)}\left(w\Pi_{2,q}Q_{n,q}\right)dq \toinfty
  \int_{ \overline{\mathfrak{B}}_{\epsilon_F+c}}\m{Tr}_{L^2(\bb R)}\left(w\Pi_{2,q}\overline{Q}_{q}\right)dq.
\end{equation}
For the second operator, we denote by $w_+$ (resp. $w_-$) the positive (resp. negative) part of $w$, so that $w= w_+-w_-$. Since $\Pi_{1,q}\sqrt{w^{\pm}}\in\mathfrak{S}_2$, it follows that $\Pi_{1,q}w^{\pm}\Pi_{1,q}\in\mathfrak{S}_1$. A simple computation shows that $q\mapsto \Pi_{1,q}w_{\pm}\Pi_{1,q}\in L^1\left(\overline{\mathfrak{B}}_{\epsilon_F+c};\mathfrak{S}_1\right)$, so that $q\mapsto \Pi_{1,q}w\Pi_{1,q}\in L^1\left(\overline{\mathfrak{B}}_{\epsilon_F+c};\mathfrak{S}_1\right)$. Therefore, in view of~\eqref{faible1_bis},
\begin{equation}
  \label{d}
  \int_{ \overline{\mathfrak{B}}_{\epsilon_F+c}}\m{Tr}_{L^2(\bb R)}\left(\Pi_{1,q}w\Pi_{1,q}Q_{n,q}\right) dq
  \xrightarrow[n\to\infty]{}
  \int_{ \overline{\mathfrak{B}}_{\epsilon_F+c}}\m{Tr}_{L^2(\bb R)}\left(\Pi_{1,q}w\Pi_{1,q}\overline{Q}_{q}\right)dq.
\end{equation}
For the last operator, we rely on the following lemma.

\begin{lemma}
  \label{lem:decomposition_Pi2_w_Pi1}
  The operator-valued function $q \mapsto \Pi_{2,q}w\Pi_{1,q}$ belongs to $L^\infty\left(\overline{\mathfrak{B}}_{\epsilon_F+c},\mathfrak{S}_1\right)$.
\end{lemma}

In particular, $q \mapsto \Pi_{2,q}w\Pi_{1,q}$ belongs to $L^1\left(\overline{\mathfrak{B}}_{\epsilon_F+c},\mathfrak{S}_1\right)$, so that, by~\eqref{faible1_bis},
\begin{equation}
  \label{c}
  \int_{ \overline{\mathfrak{B}}_{\epsilon_F+c}} \m{Tr}_{L^2(\bb R)}\left(\Pi_{2,q}w\Pi_{1,q}Q_{n,q}\right) dq
  \xrightarrow[n\to\infty]{}
  \int_{ \overline{\mathfrak{B}}_{\epsilon_F+c}}\m{Tr}_{L^2(\bb R)}\left(\Pi_{2,q}w\Pi_{1,q}\overline{Q}_{q}\right)dq.
\end{equation}

We finally obtain, by summing~\eqref{b}, \eqref{c} and~\eqref{d}, that
\begin{equation}
  \label{smallpart}
  \int_{ \overline{\mathfrak{B}}_{\epsilon_F+c} }\m{Tr}_{L^2(\bb R)}\left(w Q_{n,q}\right) dq \xrightarrow[n\to\infty]{} \int_{ \overline{\mathfrak{B}}_{\epsilon_F+c} }\m{Tr}_{L^2(\bb R)}\left(w \overline{Q}_{q}\right) dq.
\end{equation}
The combination of~\eqref{lastPart}, \eqref{firstpartConvergence} and~\eqref{smallpart} shows that~\eqref{eq:cv_Tr_wQ} holds, which allows to conclude the proof of the equality $\rho_{\overline{Q}}-\nu =\overline{\overline{\rho}}_Q-\nu$ in the sense of distributions.

\medskip

Let us conclude this section by providing the proof of Lemma~\ref{lem:decomposition_Pi2_w_Pi1}.

\begin{proof}[Proof of Lemma~\ref{lem:decomposition_Pi2_w_Pi1}]
Consider $q \in \overline{\mathfrak{B}}_{\epsilon_F+c}$. We decompose the operator $\Pi_{2,q}w\Pi_{1,q}$ as follows:
\[
\begin{aligned}
\Pi_{2,q}w\Pi_{1,q} & = \Pi_{2,q}(T_q-\epsilon_F)^{-1} (T_q-\epsilon_F)w\Pi_{1,q} \\
& = \Pi_{2,q}(T_q-\epsilon_F)^{-1} w (T_q-\epsilon_F)\Pi_{1,q} - \frac12 \Pi_{2,q} (T_q-\epsilon_F)^{-1} \left[\frac{d^2}{dz^2},w\right]\Pi_{1,q}.
\end{aligned}
\]
By the Kato--Seiler--Simon inequality~\eqref{KatoSeilerSimon}, $q \mapsto \Pi_{2,q}(T_q-\epsilon_F)^{-1} \sqrt{w_\pm}$ and $q \mapsto \sqrt{w_\pm} (T_q-\epsilon_F)\Pi_{1,q}$ both belong to $L^\infty(\overline{\mathfrak{B}}_{\epsilon_F+c},\mathfrak{S}_2)$, so that $q \mapsto \Pi_{2,q}(T_q-\epsilon_F)^{-1} w (T_q-\epsilon_F)\Pi_{1,q}$ is in $L^\infty(\overline{\mathfrak{B}}_{\epsilon_F+c},\mathfrak{S}_1)$. Moreover,
\[
\Pi_{2,q} (T_q-\epsilon_F)^{-1} \left[\frac{d^2}{dz^2},w\right]\Pi_{1,q} = 2 \Pi_{2,q} (T_q-\epsilon_F)^{-1} w' \frac{d}{dz} \Pi_{1,q} + \Pi_{2,q} (T_q-\epsilon_F)^{-1} w''\Pi_{1,q}.
\]
The decomposition
\[
\begin{aligned}
\Pi_{2,q} (T_q-\epsilon_F)^{-1} w' \frac{d}{dz} \Pi_{1,q} & = \left[\Pi_{2,q} (T_q-\epsilon_F)^{-1} \sqrt{(w')_+}\right]\left[\sqrt{(w')_+} \frac{d}{dz} \Pi_{1,q}\right] \\
& \ \ - \left[\Pi_{2,q} (T_q-\epsilon_F)^{-1} \sqrt{(w')_-}\right]\left[\sqrt{(w')_-} \frac{d}{dz} \Pi_{1,q}\right]
\end{aligned}
\]
shows that $q \mapsto \Pi_{2,q} (T_q-\epsilon_F)^{-1} w' \frac{d}{dz}\Pi_{1,q}$ belongs to $L^\infty(\overline{\mathfrak{B}}_{\epsilon_F+c},\mathfrak{S}_1)$ as the sum of products of operator-valued functions in $L^\infty(\overline{\mathfrak{B}}_{\epsilon_F+c},\mathfrak{S}_2)$. It can also similarly be shown that $q \mapsto \Pi_{2,q} (T_q-\epsilon_F)^{-1} w''\Pi_{1,q}$ is in $L^\infty(\overline{\mathfrak{B}}_{\epsilon_F+c},\mathfrak{S}_1)$, which proves the statement of the lemma.
\end{proof}

\bibliographystyle{plain}
\bibliography{rHF_Extended_Defect_Fermi_Sea_revised}

\begin{thebibliography}{10}

\bibitem{AnantharamanCances}
A.~Anantharaman and \'{E}. Canc\`es.
\newblock Existence of minimizers for {Kohn--Sham} models in quantum chemistry.
\newblock {\em Ann. Institut Henri Poincar\'e (C) Non Linear Analysis},
  26(6):2425--2455, 2009.

\bibitem{AJOO07}
W.~Aschbacher, V.~Jak\v{s}i\'{c}, Y.~Pautrat, and C.-A. Pillet.
\newblock Transport properties of quasi-free fermions.
\newblock {\em J. Math. Phys.}, 48(3):032101, 2007.

\bibitem{Bellissard94}
J.~Bellissard, A.~van Elst, and H.~Schulz-Baldes.
\newblock The noncommutative geometry of the quantum {H}all effect.
\newblock {\em J. Math. Phys.}, 35(10):5373--5451, 1994.

\bibitem{BJLP15}
L.~Bruneau, V.~Jak\v{s}i\'{c}, Y.~Last, and C.-A. Pillet.
\newblock Landauer-{B}\"{u}ttiker and {T}houless conductance.
\newblock {\em Comm. Math. Phys.}, 338(1):347--366, 2015.

\bibitem{Cances2008}
\'{E}. Canc\`{e}s, A.~Deleurence, and M.~Lewin.
\newblock A new approach to the modeling of local defects in crystals: The
  reduced {Hartree-Fock} case.
\newblock {\em Commun. Math. Phys.}, 281(1):129--177, 2008.

\bibitem{Cances2008b}
\'{E}. Canc\`{e}s, A~Deleurence, and M~Lewin.
\newblock Non-perturbative embedding of local defects in crystalline materials.
\newblock {\em J. Phys. Condens. Matter}, 20(29):294213, 2008.

\bibitem{CancesLahbabiLewin2013}
\'{E}. Canc\`{e}s, S.~Lahbabi, and M.~Lewin.
\newblock Mean-field models for disordered crystals.
\newblock {\em J. Math. Pures. Appl.}, 100(2):241–274, 2013.

\bibitem{Cances2010}
{\'E}.~Canc{\`e}s and M.~Lewin.
\newblock The dielectric permittivity of crystals in the reduced
  {Hartree--Fock} approximation.
\newblock {\em Arch. Ration. Mech. Anal.}, 197(1):139--177, Jul 2010.

\bibitem{CANCES2012887}
\'{E}. Canc\`{e}s and G.~Stoltz.
\newblock {A mathematical formulation of the random phase approximation for
  crystals}.
\newblock {\em Ann. I. H. Poincare-An.}, 29(6):887 -- 925, 2012.

\bibitem{Cao}
L.~Cao.
\newblock {\em Analyse math\'ematique du transport thermo- \'electronique dans
  les solides d\'esordonn\'es}.
\newblock Universit\'e Paris Est, 2019.

\bibitem{CATTO2001687}
I.~Catto, C.~Le Bris, and P-L. Lions.
\newblock On the thermodynamic limit for {H}artree–{F}ock type models.
\newblock {\em Annales de l'Institut Henri Poincare (C) Non Linear Analysis},
  18(6):687 -- 760, 2001.

\bibitem{Coleman}
A.~J. Coleman.
\newblock Structure of {F}ermion density matrices.
\newblock {\em Rev. Mod. Phys.}, 35:668--686, 1963.

\bibitem{Cornean2008}
H.~D. Cornean, P.~Duclos, G.~Nenciu, and R.~Purice.
\newblock Adiabatically switched-on electrical bias and the
  {L}andauer–{B}üttiker formula.
\newblock {\em J. Math. Phys.}, 49(10):102106, 2008.

\bibitem{Cornean2012}
H.~D. Cornean, P.~Duclos, and R.~Purice.
\newblock Adiabatic non-equilibrium steady states in the partition free
  approach.
\newblock {\em Ann. Henri Poincar{\'e}}, 13(4):827--856, 2012.

\bibitem{Oliveira2009}
C.~R. {de Oliveira} and A.~A. Verri.
\newblock Self-adjoint extensions of {C}oulomb systems in 1, 2 and 3
  dimensions.
\newblock {\em Ann. Phys.}, 324(2):251 -- 266, 2009.

\bibitem{Lang2001}
M.~Di~Ventra and N.~D. Lang.
\newblock Transport in nanoscale conductors from first principles.
\newblock {\em Phys. Rev. B}, 65:045402, 2001.

\bibitem{DreizlerGross}
R.M. Dreizler and E.K.U. Gross.
\newblock {\em {Density Functional Theory}}.
\newblock Springer Berlin Heidelberg, 1990.

\bibitem{Frank2011energy}
R.~L. Frank, M.~Lewin, E.~H. Lieb, and R.~Seiringer.
\newblock Energy cost to make a hole in the {F}ermi sea.
\newblock {\em Phys. Rev. Lett.}, 106:150402, 2011.

\bibitem{frank2013}
R.L. Frank, M.~Lewin, E.H. Lieb, and R.~Seiringer.
\newblock {A positive density analogue of the Lieb–Thirring inequality}.
\newblock {\em Duke Math. J.}, 162(3):435--495, 2013.

\bibitem{RevModPhys.86.253}
C.~Freysoldt, B.~Grabowski, T.~Hickel, J.~Neugebauer, G.~Kresse, A.~Janotti,
  and C.G. Van~de Walle.
\newblock First-principles calculations for point defects in solids.
\newblock {\em Rev. Mod. Phys.}, 86:253--305, 2014.

\bibitem{Friedel}
J.~Friedel.
\newblock The distribution of electrons round impurities in monovalent metals.
\newblock {\em The London, Edinburgh, and Dublin Philosophical Magazine and
  Journal of Science}, 43(337):153--189, 1952.

\bibitem{Gontier2016supercell}
D.~Gontier and S.~Lahbabi.
\newblock Supercell calculations in the reduced {H}artree-{F}ock model for
  crystals with local defects.
\newblock {\em Appl. Math. Res. Express}, pages 1--64, 2016.

\bibitem{Hainzl2005}
Ch. Hainzl, M.~Lewin, and {\'E}.~S{\'e}r{\'e}.
\newblock {Existence of a stable polarized vacuum in the Bogoliubov-Dirac-Fock
  approximation}.
\newblock {\em Commun. Math. Phys.}, 257(3):515--562, 2005.

\bibitem{0305-4470-38-20-014}
Ch. Hainzl, M.~Lewin, and {\'E}.~S{\'e}r{\'e}.
\newblock Self-consistent solution for the polarized vacuum in a no-photon
  {QED} model.
\newblock {\em J. Phys. A-Math. Gen.}, 38(20):4483, 2005.

\bibitem{Hainzl2009}
Ch. Hainzl, M.~Lewin, and {\'E}.~S{\'e}r{\'e}.
\newblock Existence of atoms and molecules in the mean-field approximation of
  no-photon quantum electrodynamics.
\newblock {\em Arch. Ration. Mech. An.}, 2009.

\bibitem{HAINZL2005TheMA}
Ch. Hainzl, M.~Lewin, and J.P. Solovej.
\newblock {The mean-field approximation in quantum electrodynamics. The
  no-photon case}.
\newblock {\em Commun. Pur. Appl. Math.}, LX:546--596, 2007.

\bibitem{Helffer}
B.~Helffer.
\newblock {\em Spectral Theory and its Applications}.
\newblock Cambridge University Press, 2013.

\bibitem{Hundertmark2007}
D.~Hundertmark.
\newblock A short introduction to {A}nderson localization.
\newblock In {\em Analysis and Stochastics of Growth Processes and Interface
  Models}, pages 194--218. Oxford Univ. Press, Oxford, 2008.

\bibitem{kaxiras_2003}
E.~Kaxiras.
\newblock {\em Atomic and Electronic Structure of Solids}.
\newblock Cambridge University Press, 2003.

\bibitem{Koch2006}
H.~Koch and D.~Tataru.
\newblock Carleman estimates and absence of embedded eigenvalues.
\newblock {\em Commun. Math. Phys.}, 267(2):419--449, Oct 2006.

\bibitem{KohnSham}
W.~Kohn and L.~J. Sham.
\newblock Self-consistent equations including exchange and correlation effects.
\newblock {\em Phys. Rev.}, 140:A1133--A1138, 1965.

\bibitem{doi:10.1063/1.3318261}
A.~V. Krasheninnikov and K.~Nordlund.
\newblock Ion and electron irradiation-induced effects in nanostructured
  materials.
\newblock {\em J. Appl. Phys.}, 107(7):071301, 2010.

\bibitem{Lahbabi2014}
S.~Lahbabi.
\newblock {The reduced Hartree-Fock model for short-range quantum crystals with
  nonlocal defects}.
\newblock {\em Ann. Henri. Poincar{\'e}}, 15(7):1403--1452, 2014.

\bibitem{LewinLiebSeiringer}
M.~Lewin, E.H. Lieb, and R.~Seringer.
\newblock The local density approximation in {D}ensity {F}unctional {T}heory.
\newblock {\em arXiv preprint}, 1903.04046, 2019.

\bibitem{LewinSabin1}
M.~{Lewin} and J.~{Sabin}.
\newblock {The Hartree equation for infinitely many particles I. Well-posedness
  theory}.
\newblock {\em Commun. Math. Phys.}, 334:117--170, 2015.

\bibitem{Lieb2005}
E.~H. Lieb and W.~E. Thirring.
\newblock {\em Inequalities for the moments of the eigenvalues of the
  Schr{\"o}dinger Hamiltonian and their relation to Sobolev inequalities},
  pages 205--239.
\newblock Springer Berlin Heidelberg, Berlin, Heidelberg, 2005.

\bibitem{MartinReiningCeperley}
R.M. Martin, L.~Reining, and D.M. Ceperley.
\newblock {\em Interacting Electrons: {T}heory and Computational Approaches}.
\newblock Cambridge University Press, 2016.

\bibitem{Mourre1981}
E.~Mourre.
\newblock Absence of singular continuous spectrum for certain self-adjoint
  operators.
\newblock {\em Commun. Math. Phys.}, 78(3):391--408, 1981.

\bibitem{RevModPhys.50.797}
S.~T. Pantelides.
\newblock The electronic structure of impurities and other point defects in
  semiconductors.
\newblock {\em Rev. Mod. Phys.}, 50:797--858, 1978.

\bibitem{ReeSim4}
M.~Reed and B.~Simon.
\newblock {\em {Method of Modern Mathematical Physics. Vol.4: Analysis of
  Operators}}.
\newblock Academic Press, San Diego, 1978.

\bibitem{ReeSim80}
M.~Reed and B.~Simon.
\newblock {\em {Method of Modern Mathematical Physics. Vol.1: Functional
  Analysis}}.
\newblock Academic Press, San Diego, 1980.

\bibitem{rumin2011balanced}
M.~Rumin.
\newblock Balanced distribution-energy inequalities and related entropy bounds.
\newblock {\em Duke Math. J.}, 160(3):567--597, 2011.

\bibitem{BarryS}
B.~Simon.
\newblock {\em Trace Ideals and their Applications}.
\newblock Amer. Math. Soc., 1979.

\bibitem{Solovej1991}
J.~P. Solovej.
\newblock {Proof of the ionization conjecture in a reduced Hartree-Fock model}.
\newblock {\em Invent. Math.}, 104(1):291--311, 1991.

\bibitem{stoneham2001theories}
M.~Stoneham.
\newblock {\em Theories of Defects in Solids}.
\newblock Oxford Classic Texts in the Physical Sciences. Clarendon Press, 2001.

\bibitem{vonNeuman}
J.~{von Neuman} and E.~{Wigner}.
\newblock {Uber merkw{\"u}rdige diskrete Eigenwerte. Uber das Verhalten von
  Eigenwerten bei adiabatischen Prozessen}.
\newblock {\em Physikalische Zeitschrift}, 30:467--470, 1929.

\end{thebibliography}

\end{document}